\documentclass[11pt]{article}
\usepackage{times}
\usepackage{fullpage}
\usepackage{epsf}
\usepackage{amsmath,amsfonts}

\usepackage{algorithmic}
\usepackage[ruled,vlined,commentsnumbered,titlenotnumbered]{algorithm2e}

\newfont{\mycrnotice}{ptmr8t at 7pt}
\newfont{\myconfname}{ptmri8t at 7pt}

\SetCommentSty{mycommfont}
\SetKwBlock{Loop}{loop}{}
\newcommand\asgn{\leftarrow}
\newcommand\fct{\rightarrow}

\newcommand\cC{{\mathcal C}}
\newcommand\ceil[1]{\lceil{#1}\rceil}

\newcommand{\bbar}[1]{\overline{{#1}}}
\newcommand{\Cdiff}{\gamma}
\newcommand\vol{\textnormal{vol}}
\newcommand\volC{\textnormal{vol}_{C}}
\newcommand\PhiC{\Phi_{C}}
\newcommand\boundC{\boundary_{C}}
\newcommand\excess[2]{\textnormal{excess}(#1,#2)}
\newcommand\excessC[2]{\textnormal{excess}_{C}(#1,#2)}
\newcommand\Reals{{\mathbb R}}
\newcommand\RealsP{{\mathbb R}_{\geq 0}}
\newcommand\ApprPR{\textnormal{PageRank}}
\newcommand\PR{\textnormal{PR}}
\newcommand\PRC{\textnormal{PR}_{C}}

\newcommand\boundary{{\partial}}

\newcommand\tO{{\widetilde O}\/}

\newcommand\tOmega{{\widetilde\Omega}}
\newcommand\tTheta{{\widetilde\Theta}}

\newcommand\pd[1]{p({#1})/d({#1})}
\newcommand\volp[1]{\vol\!\left(V^p_{{#1}}\right)}

\newcommand\boundp[1]{\left|\boundary\!\left(V^p_{{#1}}\right)\right|}
\newcommand\Vp[1]{V^p_{{#1}}}
\newcommand\ppd[1]{{p^*}({#1})/d({#1})}
\newcommand\volpp[1]{\vol\!\left(V^{p^*}_{{#1}}\right)}

\newcommand\Vpp[1]{V^{p^*}_{{#1}}}
\newcommand\ppV[1]{p^*\left(V^{p^*}_{{#1}}\right)}
\newtheorem{proposition}{Proposition}
\newtheorem{lemma}[proposition]{Lemma}

\newtheorem{corollary}[proposition]{Corollary}
\newtheorem{theorem}[proposition]{Theorem}

\newtheorem{observation}[proposition]{Observation}
\newtheorem{fact}[proposition]{Fact}

\newcommand{\qed}{\hbox{\rule{6pt}{6pt}}}
\newenvironment{proof}[1][]{\paragraph{Proof{#1}}}{\hfill\qed\medskip\\}

\newenvironment{lemmanew}[1]{\paragraph{Lemma {{#1}}'}\hspace{-2ex}\em }{}

\newcommand\drop[1]{}
\newcommand\confdrop[1]{}

\newcommand\req[1]{(\ref{#1})}
\newcommand\eps{\varepsilon}

\newcommand\packlist{\setlength{\itemsep}{1pt}
\setlength{\parskip}{0pt}\setlength{\parsep}{0pt}}

\newcounter{invari}
\begin{document}
%
\title{Deterministic Edge Connectivity in Near-Linear Time\footnote{A preliminary version \cite{KT14-C} was presented at the 47th ACM Symposium on Theory of Computing (STOC 2015).}}

\author{{\em Ken-ichi Kawarabayashi}\thanks{Ken-ichi Kawarabayashi's
    research is partly supported by JST ERATO Kawarabayashi Large Graph 
Project JPMJER1201 and
    JSPS KAKENHI JP18H05291.}\\National Institute of Informatics, 
   Tokyo, Japan\\
  \texttt{k\_keniti@nii.ac.jp}
  \and {\em Mikkel Thorup}\thanks{
Mikkel Thorup's research is supported by his
Advanced Grant DFF-0602-02499B from the Danish Council for Independent
Research and by his Investigator Grant 16582, Basic Algorithms Research Copenhagen (BARC), from the VILLUM Foundation.}\\
  BARC, University of Copenhagen\\
  \texttt{mikkel2thorup@gmail.com}}

\maketitle

\begin{abstract}
We present a deterministic algorithm that computes the
edge-connectivity of a graph in near-linear time. This is for a simple
undirected unweighted graph $G$ with $n$ vertices and $m$ edges.  This is
the first $o(mn)$ time deterministic algorithm for the problem.  Our
algorithm is easily extended to find a concrete minimum edge-cut.
In fact, we can construct the classic cactus representation of all
minimum cuts in near-linear time.

The previous fastest deterministic algorithm by Gabow from STOC'91
took $\tO(m+\lambda^2 n)$, where $\lambda$ is the edge connectivity,
but $\lambda$ can be as big as $n-1$.
Karger presented a randomized near-linear time Monte Carlo
algorithm for the minimum cut problem at STOC'96, but the returned cut is only minimum 
with high probability.

Our main technical contribution is a near-linear time
algorithm that contracts vertex sets of a simple input
graph $G$ with minimum degree $\delta$, producing
a multigraph $\bbar G$ with $\tO(m/\delta)$ edges
which preserves all minimum cuts of $G$ with at least two vertices on each side.

In our deterministic near-linear time algorithm, we will decompose the
problem via low-conductance cuts found using PageRank a la Brin and
Page (1998), as analyzed by Andersson, Chung, and Lang at
FOCS'06. Normally such algorithms for low-conductance cuts are
randomized Monte Carlo algorithms, because they rely on guessing a
good start vertex. However, in our case, we have so much structure
that no guessing is needed.
\end{abstract}



\section{Introduction}
In this paper we consider classic undirected graphs (i.e., no orientation for edges) where the edges are a set of
unordered pairs of vertices. We refer to them as a {\em simple graphs\/} to
distinguish them from {\em multigraphs} (or pseudographs) allowing parallel edges.
For both cases, the {\em edge-connectivity\/} is the smallest number of edges whose removal
disconnects the graph. This is a classic global reliability measure for
the connectivity of a graph. The set of edges removed are the {\em cut edges\/}
of a {\em (global) minimum cut}, or for short, a {\em min-cut}, and the two components we get
when removing them are the {\em sides of the cut\/}. In this paper, we are
assuming that the graph is connected, which is trivially
checked in linear time.

Our main result is a deterministic near-linear time algorithm to find the
edge connectivity and a global minimum cut of a simple graph. It is
based on a new understanding of the cuts in simple graphs that does not
hold for multigraphs.

\subsection{Previous work}
We will now discuss previous work on global min-cut algorithms. For
the bounds we have $n$
vertices, $m$ edges, and (unknown) edge-connectivity $\lambda$.
In the discussion, we consider
both simple graphs and multigraphs, but our own results are only for 
simple graphs. The discussion also considers {\em
  weighted graphs}, where edges have weights. Then edge-connectivity
is no longer relevant, but the {\em size of a cut\/} is the total
weight of the cut edges. For weighted graphs, parallel edges can be
merged adding up the weights, so weighted graphs may be assumed
simple.


In 1961, Gomory and Hu~\cite{GH61} showed that the global minimum cut problem can
be solved by computing $n-1$ independent minimum $s$-$t$ cuts, that is, cuts with
$s$ and $t$ on different sides. They let $s$ be an arbitrary vertex,
and try with $t$ being any of other vertices. The point is
that to find a minimum cut, they just have to guess a vertex $t$ on the
side that $s$ does not belongs to. The $s$-$t$ cuts
are understood via Menger's classic theorem \cite{menger}. We can thus use
{\em any\/} $s$-$t$ cut algorithm. On a multigraph, if we use the classic augmenting path
algorithm of Ford and Fulkerson \cite{FF56},  we should work in
parallel on all $t$, doing the same number
of augmenting paths to each $t$. After $\lambda$ augmentation rounds,
for some $t$, we find an $s$-$t$ cut of size $\lambda$ which
is also a global min-cut. The  total
time is $O(\lambda nm)$. We could also apply the $O(m^{3/2})$ time $s$-$t$ min-cut
algorithm of Even and Tarjan \cite{ET75},
and solve the global min-cut problem for multigraphs in $O(nm^{3/2})$
time. This improves Ford-Fulkerson when $\lambda=\omega(m^{1/2})$.

The first algorithm to compute a global minimum cut faster than $n$ independent
$s$-$t$ cuts is the $O(\lambda n^2)$ time\footnote{We know $\lambda n=O(m)$, and
this implies $\lambda n^2=O(mn)$} algorithm of Podderyugin \cite{Po73}
for simple graphs from 1973. For many years, this algorithm did not receive
attention until it was rediscovered by Karzanov and Timofeev
\cite{KT86} and by Matula \cite{Ma87}, independently.

In the 1990s, the above bounds for simple graphs were generalized to
multigraphs and weighted graphs. In 1990, Nagamochi and
Ibaraki \cite{NI92A} gave an $O(m+\min\{\lambda n^2,\,pn+n^2\log n\})$
time global min-cut algorithm for multigraphs where $p \leq m$ is the number of pairs
of vertices between which the graphs has an edge. For weighted graphs, they got a general bound of $O(nm+n^2\log n)$. Hao and
Orlin~\cite{HO94} obtained an $O(nm \log (n^2/m))$ time algorithm for the
directed weighted case.  Stoer and Wagner \cite{SW97} and Frank
\cite{frank94}, independently, presented a very simple algorithm finding
a global min-cut of an undirected weighted graph within the
same $O(nm+n^2\log n)$ time bound as in \cite{NI92A}.


The current best deterministic algorithm for simple graphs is from 1991
due to Gabow~\cite{Gab95} who, using \cite{NI92} for preprocessing, 
gets down to $O(m+\lambda^2n \log (n/\lambda))$ time. Gabow \cite[pp. 268-269]{Gab95} also discuss multigraphs with a slightly
worse bound of $O(m+\lambda^2n \log n)$. A linear time
$(2+\eps)$-approximation of the edge-connectivity was presented by Matula \cite{Mat93}.


All the above-mentioned algorithms have been deterministic.
The study of randomized algorithms for the global minimum cut problem was
initiated by Karger~\cite{Kar99}. We distinguish between \emph{Las
  Vegas} algorithms with guaranteed correct answers but expected
running times, and \emph{Monte Carlo} algorithms with bounded running
time but some small probability of incorrect answers. Las Vegas
algorithms can be converted to Monte Carlo algorithms (making an error
if the algorithm doesn't terminate within a certain time bound) so Las
Vegas is generally prefered. If we have a way to certify correct
answers, then we can also convert Monte Carlo to Las Vegas by
rerunning if the certifyer fails.  Karger~\cite{Kar99} presented Las
Vegas algorithms yeilding a $(1+o(1))$-approximation to the global
minimum cut in $˜O(m)$ expected time and an exact global minimum cut
in $˜O(m\sqrt{\lambda})$ expected time.  Karger and Stein~\cite{KS96}
showed that random edge contraction works well for the global minimum
cut problem, leading to a Monte Carlo algorithm running in $O(n^2
\log^3 n)$ time.  Finally, Karger~\cite{Kar00} gave a randomized $O(m
\log^3 n)$ time Monte Carlo algorithm for the global minimum cut
problem. Karger~\cite{Kar00} points out that we are lacking a
correspondingly efficient way to certify the minimality of the
returned cut.

For more detailed history for the global minimum cut problem, we refer the reader to the book by Schrijver \cite{Lex03}. We note that a deterministic
near-linear time min-cut algorithm is known for planar graphs \cite{CFN04:plan-mincut}.

\subsection{Main results}
In this paper, we present a deterministic near linear time algorithm
for computing the edge connectivity and a global minimum cut for a
simple graph. By near-linear time, we mean $\tO(m)$ time where $\tO$ hides log factors. 
This is the first $o(mn)$ time deterministic algorithm
for the problem.  The previous best $\tO(m+\lambda^2
n)$ time bound of Gabow \cite{Gab95} is good if
$\lambda$ is small, but we may have $\lambda=\Omega(n)$.

We note that in this paper, we are not trying to minimize the number
of log factors hidden in the $\tO$-notation, nor do we count them
explicitely, e.g., we will freely use reductions like
$\tO(f(n))\tO(g(n))=\tO(f(n)g(n))$, At the end, we will loosely
estimate the number to $12$, thus estimating the running time of our
algorithm to $O(m\log^{12} n)$. The purpose of this estimate is only
to encourage other researchers to get down to a more reasonable number
of log factors, providing them a concrete bound to improve on. Indeed
Henzinger et al.~\cite{HRW17} have recently improved important parts
of our construction bringing the total running time down to $O(m(\log
n)^2(\log\log n)^2)$. 

In near-linear time we can also compute the \emph{cactus
  representation} of all global minimum cuts introduced in
\cite{DKL76}. To do so we involve the previous fastest $\tO(\lambda
m)$ time algorithm by Gabow~\cite{Gab16} as a black-box.  We note here
that Karger and Panigrahi \cite{DBLP:conf/soda/KargerP09} have presented a
near-linear time Monte Carlo algorithm for constructing the cactus
data structure.

\subsection{Technical Result}
Henceforth, we are only considering unweighted graphs. We are given a
simple graph, and we want to find a min-cut in near-linear time. We
may assume that the minimum degree $\delta$ is at least
polylogarithmic; for otherwise, we can just use Gabow's \cite{Gab95},
algorithm to find the min-cut in $\tO(\lambda m)=\tO(\delta m)=\tO(m)$
time.  For our purposes, it suffices to assume $\delta\geq\log^6 n$.

By a {\em trivial cut\/}, we mean a cut where one side consists of a
single vertex. We note that the minimum degree $\delta$ is an upper bound on the edge-connectivity $\lambda$ since it is the
smallest size of a trivial cut. Finding $\delta$ is trivial, but
we could have $\lambda < \delta$.

What makes simple graphs special is that if you have a non-trivial cut in a
simple graph, then each side needs to have at least $\delta$ vertices. This
observation and its relation to conductance will be discussed below in Section 
\ref{sec:min-cut-low-conductance}.

By {\em contracting a vertex set $U\subseteq V$}, we mean identifying
the vertices in $U$ while removing the edges between them. We may not
check that $U$ is connected, so this may not correspond to edge
contractions. The identity of edges not removed are preserved, and a
cut is preserved if all the cut edges are preserved. 

The basic idea in our min-cut algoritm is do contractions
in the graph while preserving all non-trivial min-cuts, continuing until the graph is so small
that Gabow's algorithm can finish in $\tO(m)$ time where $m$ is the number of edges in the
original simple input graph. The contractions may introduce parallel edges,
so internally, the algorithm works with a multigraph.

Our main technical contribution is to prove
the following theorem:
\begin{theorem}\label{thm:main-tech}
Given a simple input graph $G$ with $n$ vertices, $m$ edges, and minimum degree
$\delta$,  in near-linear time, we can contract vertex sets
producing a multigraph $\bbar G$ which has only
$\bbar m=\tO(m/\delta)$ edges, yet which preserves all non-trivial min-cuts of
$G$.
\end{theorem}
From Theorem \ref{thm:main-tech},  we easily get our near-linear min-cut algorithm:
\begin{corollary}\label{cor:min-cut}
We can find a minimum cut of a simple graph $G$ in near-linear time.
\end{corollary}
\begin{proof}
Let $\delta$ be the minimum degree of $G$. We apply the
Theorem \ref{thm:main-tech} to $G$
producing the graph $\bbar G$. We now run Gabow's min-cut algorithm  \cite{Gab95} on $\bbar G$, asking it to fail if the edge-connectivity is
above $\delta$. This takes  $\tO(\delta \bbar m)=\tO(m)$ time,
and now we compare the output with the minimum degree $\delta$.
\end{proof}
Likewise, in near-linear time, we can obtain the cactus representation
of all global minimum cuts from \cite{DKL76} by applying the
cactus algorithm of Gabow~\cite{Gab16} to $\bbar G$. Having produced
the cactus $\bbar C$ of $\bbar G$, we just need to add min-degree vertices as extra
needles so as to get the cactus of the input graph $G$.
A description of this including the definition of the min-cut cactus is given in Section \ref{cactuss}.

We note that Theorem \ref{thm:main-tech} cannot hold if the input
graph is a multigraph. To see this, consider a cycle of length $n\geq
4$, but where every edge is replaced by $k=(\log n)^{\omega(1)}$
parallel edges. Now every edge is involved in a non-trivial min-cut
where each side consists of at least two consecutive vertices, and
therefore no edges can be contracted. This shows that the contractions
of Theorem \ref{thm:main-tech} are very specific to simple
graphs. Also, they can only preserve non-trivial min-cuts, for if we,
for example, take a complete graph, then every edge is in a trivial
min-cut. However, a complete graph can be contracted to a single
vertex since it has no non-trivial min-cut.

\paragraph{Further results}
While the reduction in Theorem \ref{thm:main-tech} of the number of
edges looks like a typical sparsification, it is not, for edges are
contracted, not deleted, and the resulting $\bbar G$ will have much
fewer vertices than $G$. In fact, combining with techniques of
Nagamochi and Ibaraki \cite{NI92}, we can make sure that the number of
vertices remaining after the contractions is $\tO(n/\delta)$. Since
the number of min-cuts in any multigraph is at most quadratic, a nice
consequence is that that the number of min-cuts in a
simple graph with $n$ vertices and minimum degree $\delta$ is at most
$n+\tO((n/\delta)^2)$. 
We are not aware of anyone else that has
observed that a large minimum degree in a simple graph implies few
minimum cuts, though it does apppear that this fact could also
be derived from the cactus representation \cite{DKL76}.

Our contraction technique can be strengthend to also preserve
approximately min-cuts. More precisely, with min-degree $\delta$ and 
edge connectivity $\lambda$, we will strengthen
Theorem \ref{thm:main-tech} to preserve all non-trivial cuts
of size at most $\lambda+(1-\eps)\delta$, where $\eps$ is an arbitrarily small
positive constant. Since $\lambda\leq\delta$, this implies 
that we preserve all $(2-\eps)$-approximate min-cuts.
Formally, we will prove:
\begin{theorem}\label{thm:main-tech-strong}
Let $\eps\in(0,1]$ be a constant. Given a simple input graph $G$
with $n$ vertices, $m$ edges, minimum degree $\delta$, and (unknown) edge connectivity $\lambda$, in $\tO(m)$ time, we can
contract vertex sets producing a multigraph $\bbar G$ which has only
$\tO(n)$ edges and $\tO(n/\delta)$ vertices, yet which preserves all
non-trivial cuts of size below $\lambda+(1-\eps)\delta$.
\end{theorem}
We know from \cite{HW96} that for a multigraph with $n$ vertices and
edge connectivity $\lambda$, the number cuts of size $3\lambda/2$
is $O(n^2)$. Applying this to the contracted graph from Theorem 
\ref{thm:main-tech-strong} with $\eps=1/2$, we get 
\begin{corollary}\label{cor:cut-number} In simple graph with
with $n$ vertices, minimum degree $\delta$, and edge
connectivity $\lambda$, there are at most
$n+\tO((n/\delta)^2)$ cuts
of size at most $3\lambda/2$. 
\end{corollary}
The ability to preserve approximate min-cuts with Theorem \ref{thm:main-tech-strong} has been used for the amortization in a recent efficient
algorithm \cite{GHT16:inc-edge-conn} 
to maintain a min-cut of a simple graph incrementally, paying
only polylogarithmic time per edge insertion.

\subsection{Minimum cuts and low conductance}\label{sec:min-cut-low-conductance}
Our approach to finding a minimum cut involves cuts of low conductance,
defined below. Generally we {\em define a cut by specifying
one side $U\subset V$}. Then the other side $T=V\setminus U$ is implicit.
No side is allowed to be empty. Algorithmically, it will typically be the smaller side that we specify
explicitly. The
edges leaving $U$ are the \emph{cut edges}, and the set of cut edges
is denoted $\boundary U=\boundary T$. The {\em size of the cut\/} is the
number of cut edges $|\boundary U|$. We do not require that any side
of a cut remain connected if the cut edges are removed.

We are also interested in the sum of the degrees of vertices in $U$ called the {\em volume of $U$\/} defined as
\[\vol(U)=\sum_{v\in U} d(v)\]
Edges with both end-points in $U$ are called {\em internal to $U$\/},
and they are counted twice in the volume of $U$.

Now the {\em conductance of $U$\/} is
defined by
\[\Phi(U)=\frac{|\boundary U|}{\min\{\vol(U),\vol(T)\}}=\Phi(T).\]
\begin{observation}\label{obs:min-sparse}
Let $S$ be the smaller side of a min-cut of our simple graph $G$. Then
either the cut is trivial with $S$ consisting of a single vertex, or $S$ has volume at least
$\delta^2$ and the conductance is $\Phi(S)\leq1/\delta$.
\end{observation}
\begin{proof}
The graph has minimum degree
$\delta$ so the min-cut has at most $\delta$ edges. Since $G$ is
simple, a vertex $v\in S$ has at least $\delta-(|S|-1)$ edges leaving
$S$. The total number of edges leaving $S$ is thus at least
$|S|(\delta+1-|S|)$, and for this to be at most $\delta$, we need
$|S|=1$ or $|S|\geq \delta$. In the latter case, we have
$\vol(S)\geq \delta^2$, so $\Phi(S)\leq
1/\delta$.
\end{proof}
\subsection{Certify-or-cut}
In our algorithm, we are going to assume that the
simple input graph $G$ has minimum degree
\[\delta\geq\lg^6 n.\]
By Observation \ref{obs:min-sparse}, this means that any non-trivial
min-cut has very low conductance. With this in mind, we
are going to devise a near-linear time
deterministic ``certify-or-cut'' algorithm that will either
\begin{enumerate}
\item Certify that there are no non-trivial min-cuts. In
particular, this witnesses that any min-degree vertex forms
the side of a global min-cut, or
\item Find a low-conductance cut.
\end{enumerate}
We note that each of the above tasks alone is beyond our current
understanding of deterministic algorithms.  For the first certification task,
recall
the issue mentioned by Karger~\cite{Kar00} that we have no
efficient deterministic way of certifying that a proposed minimum cut is
indeed minimum. Our task
is no easier, for if it was, to certify that a cut of size $k\leq \delta$
is minimum, we could attach a
complete graph on $k$ vertices, where $k-1$ of the vertices are new.
Each new vertex defines a trivial cut of size $k-1$, and the
edge connectivity of the original graph is $k$ if and only if there
is no non-trivial minimum cut in the new graph.

For the second task, we want to find a low-conductance cut, e.g.,
using PageRank \cite{Google98:pagerank} as analyzed by Andersson,
Chung, and Lang \cite{ACL07:pagerank}. However, such algorithms for
low-conductance cuts are randomized Monte Carlo algorithms, because
they rely on guessing a good start vertex. For cut-or-witness,
however, we only have to find a low conductance cut if we fail to
witness the minimality of the trivial cuts, but then we will have so
much structure that no guessing is needed.

Our certify-or-cut algorithm will illustrate some of the basic techniques presented in this paper, including a study of what happens in the endgame of PageRank when most mass has been distributed, yet some vertex is still left out.

\subsection{The overall algorithm}\label{ssec:overall}
We will now sketch the basic ideas by using a more elaborate
certify-or-cut algorithm for finding a minimum cut, and also point
to the issues that arise.

Given a component $C$ of subgraph $H$ of $G$, suppose we can either
\begin{enumerate}
\item certify that $C$ is a so-called ``cluster''
implying that no min-cut of $G$ induces
a non-trivial cut of $C$,
or
\item find a cut of $C$ of conductance $o(1/\log m)$.
\end{enumerate}
Then, starting from $H=G$, we will recursively remove the
low-conductance cuts, until we have a subgraph $H$ of $G$ where all
the components are certified clusters. Inside each cluster $C$, we
will identify a so-called ``core'' $A$ with the property that no
non-trivial min-cut of $G$ makes any cut of $A$ (let us observe that
$A$ may not be all of $C$ because a non-trivial min-cut of $G$ could
induce a trivial cut of $C$). Cores can therefore be contracted
without affecting any non-trivial min-cut of $G$.

The important observation here is that when removing the
low-conductance cuts, most edges survive in $H$. 
The same observation was used in Spielman and Teng's spectral sparsifiers \cite{ST11:spectral}, though they used randomization to find the low-conductance
cuts. The reason that only few edges get
removed by recursive low-conductance cuts is that we
can amortize the edges removed over the edges incident to the smaller
side where smaller is measured in terms of volume, that is, number of
incident edges. Each edge incident to the smaller side pays $o(1/\log
m)$ (because of the low-conductance cuts), and it can end on the
smaller side at most $\lg m$ times, where $\lg=\log_2$. The total
fraction of edges cut is thus $o(1)$, so most edges remain when we are
done removing low-conductance cuts, certifying that each remaining component of $H$ is a cluster. This is important because we
want many edges to be contracted when we contract the cores of the
clusters in $H$.

We now point out the issues we have to address.
The first issue is that as edges get removed, the degrees of
the remaining vertices will decrease, and then the minimum degree
could fall below $\lg n$, so we can no longer use
Observation  \ref{obs:min-sparse} to conclude that a non-trivial cut
has conductance $o(1/\log m)$. Our fix to this issue will be to not
only remove cut edges, but also ``trim'' the resulting components, removing
all vertices that have lost $3/5$ of their original edges. As we shall see,
this will only increase the number of edges removed by a factor $5$, so most
edges will still  remain in the final clusters.

The second issue happens when we contract the cluster cores in a graph $\bbar G$ that preserves
all the non-trivial min-cuts of $G$. This  may introduce
parallel edges, and hence Observation~\ref{obs:min-sparse} fails
completely, e.g., consider a path of length $4$ where consecutive
vertices are connected by $\delta$ parallel edges. A non-trivial
min-cut with two vertices on each side has conductance $1/2$. We
will, however, argue that if a vertex is dominated by parallel edges,
then it is somehow done and can be ignored.

Handling the above
two complications will also force us to adopt a more complicated notion of
a cluster, but our algorithm will still follow the basic pattern of the
above sketch.

　
The goal is to contract cluster cores until $\bbar G$ has only
$\tO(m/\delta)$ edges, yet preserves all non-trivial min-cuts from $G$,
as desired for Theorem \ref{thm:main-tech}. To find a minimum cut
of $G$, we finish by applying Gabow's algorithm \cite{Gab95} as described
in Corollary \ref{cor:min-cut}.

\subsection{Recent improvement of Henzinger et al.} After this work was announced at STOC'15 
\cite{KT14-C}, Henzinger et al.~\cite{HRW17} have improved the
concrete running time from $O(m\log^{12} n)$ to $O(m(\log
n)^2(\log\log n)^2)$, which also improves the $O(m\log^3 n)$ time
bound of Karger's Monte Carlo algorithm \cite{Kar00}. The algorithm of
Henzinger et al.~uses the overall approach developped in this
paper. However, instead of using a diffusion based PageRank like us to find
low-conductance cuts, they use an interesting local flow based
algorithm.  We note that diffusion and flow based algorithms have
different advantages in different settings, and we hope that our novel
use and analysis of the PageRank diffusion will inspire other algorithmic
applications.

\subsection{Notations}\label{sec:notation} 
As a generic notation, if we have some graph parameter like the edge
connectivity $\lambda$, we may use a subscript to specify which graph
it is measured on as in $\lambda_H$ for the edge connectivity of the
graph $H$. We may also put $H$ in paranthesis, e.g., we use
$V(H)$ and $E(H)$ to denote the vertices and edges in $H$, and let
$n(H)=|V(H)|$ and $m(H)=|E(H)|$ denote the number of vertices and edges in $H$.

To simplify calculations, we will make use of $O$-, 
$o$-, and $\tO$-notation to hide constants
and log factors when they are not important to our results. Starting with
$O$-notation, by definition, $f_i(n)=O(g_i(n))$ means that there
constants $n_i$ and $c_i$ such that $n\geq n_i$ implies $f_i(n)\leq
c_i g(n)$.  If we have this for $i=1,2$, then
$f_1(n)f_2(n)=O(g_1(n)g_2(n))$ since $n\geq\max\{n_1,n_2\}$ implies
$f_1(n)f_2(n)=c_1c_2\,g_1(n)g_2(n)$. Instead of making $f_1$ and $f_2$
explicit, we can write this rule as
$O(g_1(n))O(g_2(n))=O(g_1(n)g_2(n))$.  Likewise, we have
$O(g_1(n))+O(g_2(n))=O(g_1(n)+g_2(n))$ and $O(O(g(n)))=O(g(n))$, all
illustrating how $O$-notation simplifies calculations. When we say
that $n$ is large enough, we mean that it is bigger than any of the
$n_i$ used in our analysis.

The assumption of a large enough $n$ becomes more important when we
combine $O$-notation with $o$-notation. By definition,
$f_i(n)=o(g_i(n))$ means that for any constant $c_i$ there is a
constant $n_i$ such that $n\geq n_i$ implies $f_i(n)< g(n)/c_i$.
Suppose we have $f(n)=o(g(n))$ and $g(n)=O(h(n))$. Then there exists
an $n_0$ such that $n\geq n_0$ implies $f(n)<h(n)$. More precisely,
from $g(n)=O(h(n))$ we get that there are constants $n_1$ and $c_1$
such that $n\geq n_1$ implies $g(n)\leq c_1 h(n)$. Next, from
$f(n)=o(g(n))$, we get that there is a constant $n_2$ depending on
$c_1$ such that $n\geq n_2$ implies $f(n)< g(n)/c_1$.  Thus $n\geq
n_0=\max\{n_1,n_2\}$ implies $f(n)<h(n)$. Our most common use of $n$
being sufficiently large is that $o(1)$ becomes
smaller than any concrete constant.

As usual, we have the derived notations $f(n)=\Omega(g(n))\iff g(n)=O(f(n))$,
$f(n)=\Theta(g(n))\iff f(n)=O(g(n))\wedge f(n)=\Omega(g(n))$, and $f(n)=\omega(g(n))\iff g(n)=o(f(n))$.

Finally, we have the $\tO$-notation, where $f_i(n)=\tO(g_i(n))$ means
that there constants $n_i$ and $c_i$ such that $n\geq n_i$ implies
$f_i(n)\leq g(n)\lg^{c_i} n$.  As for the $O$-notation, we get
simplyfying rules like $\tO(g_1(n))\tO(g_2(n))=\tO(g_1(n)g_2(n))$,
$\tO(g_1(n))+\tO(g_2(n))=\tO(g_1(n)+g_2(n))$ and
$\tO(\tO(g(n)))=\tO(g(n))$. As for $O$-notation, we are going to use
the derived notations $f(n)=\tOmega(g(n))\iff g(n)=\tO(f(n))$ and
$f(n)=\tTheta(g(n))\iff f(n)=\tO(g(n))\wedge f(n)=\tOmega(g(n))$. We
are only going to use $\tO/\tOmega/\tTheta$-notation to hide $\lg n$
factors where $n$ denotes the number of vertices in the simple input
graph for which we want to find a minimum cut.

\subsection{Contents}
This paper is structured as follows. First we will show how to
implement the certify-or-cut algorithm described above, since it
introduces most of the interesting new ideas in a quite clean form. To
do so, we will first describe our view of PageRank in Section
\ref{sec:PageRank}, which includes a new theorem on the endgame.  Next
we describe the certify-or-cut algorithm in Section
\ref{sec:certify-or-cut}. After this warm-up, we are ready to present the recursive
set-up for our contraction based min-cut algorithm in Section \ref{sec:min-cut}.
To complete the min-cut algorithm, we present the use of low-conductance
cuts in Section \ref{sec:cluster-or-cut}.  In Section \ref{sec:approx}, we show how we can also preserve approximate
min-cuts.
Finally, in Section \ref{sec:pagerank-anal}
we prove the PageRank theorems claimed in Section \ref{sec:PageRank}.
A cactus construction is given in Section \ref{cactuss}.

\section{Sparse cuts by PageRank}\label{sec:PageRank}
We are going to find sparse cuts using the PageRank algorithm from 
\cite{ACL07:pagerank}.  We will be running it on a multigraph with
$m$ edges. This will be a subroutine of our min-cut algorithm, applied
to different minors of the original simple input graph.

The PageRank algorithm is operating with a mass distributions $p\in \RealsP^V$ assigning
non-negative mass to the vertices. Given a subset $U$ of the
vertices, $p(U)=\sum_{v\in U} p(v)$ denotes the \emph{total mass} on the subset.
We refer to $p(U)/\vol(U)$ as the \emph{density} on $U$. For an individual
vertex $v$, the density is $p(v)/d(v)=p(v)/\vol(\{v\})$.

We start with some initial mass distribution $p^\circ\in\Reals^V$ on
the vertices. Most of the time, the total mass is normalized to 1,
corresponding to a probability distribution.

The algorithm has a parameter $\alpha$ called the \emph{teleportation} constant,
and we assume $\alpha\leq 1/3$. In our min-cut algorithm, we will have $\alpha=\tOmega(1)$,
but this is not assumed for the  results in this section.

The algorithm operates by moving
mass between two mass distributions: a {\em residual mass\/} $r$ which is
initialized as the initial distribution $p^\circ$, and a {\em settled mass\/} $p$
which is initially zero on all vertices. Generally we say that the
{\em density of mass\/} on a vertex is the mass divided by the degree

The algorithm works by {\em pushing} residual mass from vertices. To push
the residual mass on $u$, we first settle a fraction
$\alpha$ of the residual mass on $u$, and then we spread half the
remaining residual mass evenly to the neighbors of $u$. This is
described in Algorithm \ref{alg:push}.
\begin{algorithm}\label{alg:push}
\caption{Push$(\alpha,u)$}
$p(u)\asgn p(u)+\alpha r(u)$\;
\lFor{$(u,v)\in E$}{
$r(v)\asgn r(v)+(1-\alpha)r(u)/(2d(u))$}
$r(u)\asgn (1-\alpha)r(u)/2$.
\end{algorithm}
The overall algorithm is flexible in that we can apply pushes to the
vertices in any order we want. To control the amount of work done,
\cite{ACL07:pagerank} introduces a parameter $\eps$, and they only
push from a vertex $u$ if the \emph{residual density} $r(u)/d(u)$ is at least
$\eps$. The resulting {\em PageRank\/} algorithm is
described in Algorithm \ref{alg:apprpr}.
\begin{algorithm}\label{alg:apprpr}
\caption{\ApprPR$(\alpha,\eps,p^\circ)$}
$r\asgn p^\circ$;$\quad p\asgn 0^V$\;
\lWhile{$\exists u: r(u)/d(u)\geq \eps$}{Push$(\alpha,u)$}
\end{algorithm}
As noted in \cite{ACL07:pagerank}, the time to do a push at $u$ is
$d(u)$ and it settles $\alpha r(u)\geq \alpha d(u)\eps$ of the
residual mass. If we thus start with a total residual mass at most 1,
the total amount of work is $O(1/(\alpha \eps))$.
This does assume, however, that $p^\circ$ is
presented in such a way that we have direct access to vertices density
$\eps$ or more. In fact, we typically assume that the vertices with positive mass are
listed in order of non-increasing initial density. Then, for any
$\eps$, we find those with initial density above $\eps$ as a prefix of
this list.

As $\eps$ approaches $0$, the residual mass vanishes, and then, as proved
in  \cite{ACL07:pagerank}, the
settled mass approaches a unique limit denoted $\PR(\alpha,p^\circ)$ that we
refer to as the {\em limit mass distribution}. The limit mass distribution
will play an important role in our analysis, but algorithmically, we will
only run the PageRank from Algorithm \ref{alg:apprpr}
with $\eps=1/\log^{O(1)}n$.
From \cite{ACL07:pagerank}, we get that pushes
maintain the following invariant:
\begin{equation}\label{eq:inv-push}
\PR(\alpha,p^\circ)=p+\PR(\alpha,r).
\end{equation}
From \cite{ACL07:pagerank} we know that
$\PR(\alpha,\cdot)$ is a {\em non-negative linear transformation\/} $\Reals^n\rightarrow
\Reals^n$, that is, for any teleportation constant $\alpha$, there is 
an $n\times n$ matrix $M_\alpha$ with
real non-negative entries such that for any initial distribution
vector $p^{\circ}$, we get the limit distribution  $\PR(\alpha,p^\circ)=p^{\circ} M_\alpha$. 
For any $\sigma\in\Reals$,
let $\bbar\sigma$ be the distribution where all vertices have density
$\sigma$. From \cite{ACL07:pagerank} we know that
$\bbar\sigma$ is a fix-point for $\PR(\alpha,\cdot)$, that is,
$\PR(\alpha,\bbar\sigma)=\bbar\sigma$, and we call it a {\em
  stationary\/} distribution.

Mass can only be moved and settled via pushes.
Consider an edge $(u,v)\in E$. Viewing it as
directed from $u$ to $v$, we get a positive flow when we push from $u$,
pushing $(1-\alpha)r(u)/(2d(u))$
mass over $(u,v)$ to $v$ while settling $\alpha r(u)$ mass at $u$. Likewise we get
a negative flow over $(u,v)$ when we push from $v$. Hence
\begin{fact}\label{fact:flow}  Recall that $p(.)$ is the settled mass. After any sequence of pushes for any
$(u,v)\in E$, the total net flow of mass over $(u,v)$ is $\frac{1-\alpha}{2\alpha}
\left(\pd u- \pd v\right)$.
\end{fact}
An important consequence is
\begin{lemma}\label{lem:stationary}
If at some point all residual densities are bounded by $\sigma$, then
from this point forward, the net flow over any edge is at most
$\sigma/(2\alpha)$.
\end{lemma}
\begin{proof} The residual distribution $r$ is bounded
by the stationary distribution $\bbar\sigma$ with densities $\sigma$, so
$\PR(\alpha,r)\leq \PR(\alpha,\bbar\sigma)=\bbar\sigma$ where $\leq$ is vector domination. If $p$ is a
mass distribution settled from $r$, then $p\leq \PR(\alpha,r)\leq
\bbar\sigma$, so $p(u)/d(u)-p(v)/d(v)\leq \sigma$ for every possible
edge $(u,v)\in E$. By Fact \ref{fact:flow}, the net flow over $(u,v)$ based on
$r$ is therefore at most $\sigma/2\alpha$.
\end{proof}
We are going to find the side $S$ of a low-conductance cut via a so-called ``sweep'' over
the settled mass distributions $p$. To describe the sweep, as general
notation, for any comparison operator
$\circ\in\{=,<,>,\leq,\geq\}$ and $t\in \Reals$, define
\[\Vp {\circ t}=\{u\in V\mid \pd u\circ t\}\textnormal,\]
e.g, $V^p_{\geq t}=\{u\in V\mid \pd u\geq t\}$. Now
let $\Phi(p)$ be the smallest conductance we can obtain by picking
some threshold $\tau\in[0,1]$, and considering the
set of vertices with density at least $\tau$, that
is,
\[\Phi(p)=\min_{\tau\in[0,1]} \Phi(\Vp{\geq \tau}).\]
To find $\Phi(p)$, we {\em sweep\/} over the vertices in order of
non-increasing settled density. We only have to consider vertices with
positive settled mass, including their incident edges, of which there
are only $O(1/(\alpha \eps))$ assuming that the total initial mass is
1.  As described in \cite{ACL07:pagerank}, we can implement the sweep
in $O((\log n)/(\alpha \eps))$ time, and we shall further bring the sweep
time down to $O(1/(\alpha \eps))$
in Section \ref{sec:sweep}.
The important question is, however, when does the
sweep give us a cut of low conductance? We will give some sufficient
conditions in the next subsection.

\subsection{Limit concentration and low conductance}\label{sec:thm:conc}
We now state conditions under which a PageRank algorithm starting from
an initial distribution $p^0$ can find a low conductance cut.  The
conditions are all based on the limit mass distribution
$p^*=\PR(\alpha_0,p^\circ)$. While this limit is unique, there are
different ways of running a PageRank, e.g., the choice and use of $\eps$ in
Algorithm \ref{alg:apprpr} and which cut we choose to return from
the sweep.

As in \cite{ACL07:pagerank}, we first study situations where the
limit mass on some set $S$ deviates significantly from the uniform
$\vol(S)/(2m)$, as quantified by
\[\excess {p^*} S=p^*(S)-\vol(S)/(2m).\]
It may seem surprising that we look at this additive excess, rather
than the multiplicative difference, but imagine that we have an
initial distribution placing the mass 1 on a single vertex $v$ of
minimum degree.  Then the first push will settle mass $p(v)=\alpha$ on $v$,
and in the limit $p^*>p(v)=\alpha$. However, on the average $v$ should only have
mass $\vol(v)/(2m)\leq 1/n$ and in our min-cut algorithm, we will have
$\alpha=\tOmega(1)$, so the multiplicative difference on
$\{v\}$ is huge.  

We think of the additive excess as the mass that
gets trapped in $S$ when we push to the limit. As proved in
\cite{ACL07:pagerank}, a large excess can only happen if there is a
low conductance cut somewhere, and then we can find some low
conductance cut efficiently.

The basic result, formalized below in Theorem \ref{thm:ACL}, is 
that if there exists a set $S$ with limit excess at least $\Cdiff$, then using a PageRank algorithm, we can find a set $T$ which is the smaller (in volume) side
of a cut with 
conductance
\[\Phi(T)=O(\sqrt{(\alpha\log m)/\Cdiff})\]
in time $O(\vol(T)(\log m)/(\Cdiff \alpha))$. Inside our min-cut algorithm, we
will have $\Cdiff,\alpha=\tOmega(1)$ and then $T$ is found in time $\tO(\vol(T))$, so
the algorithm is fast if $T$ is small. It is here important that $T$ is the smaller
side of the cut. We are not allowed to spend a long time reporting the bigger side instead.

If we are further given an upper bound $s\leq m\Cdiff/16$ on the volume of $S$, then 
the volume of the returned set $T$ has volume bounded by $8s/\Cdiff=\tO(s)$.
In this case, it is further guaranteed that $T$ has excess at least
$\Cdiff/(16\lg (4s))$.

When we start running the algorithm, we may not know if the set $S$ 
exists. If it doesn't
exist, the algorithm will either return $T$ as described above, or
certify that no such $S$ exists. Inside our min-cut algorithm, there
will be cases where certifying the non-existence of $S$ is the
prefered outcome.  We note that if there is no set $T$ with
$\vol(T)\leq 8s/\Cdiff$ and $\excess {p^*} T\geq\Cdiff/(16\lg (4s))$,
then we must get the certificate that there
is no set $S$ with $\vol(S)\leq s$ and $\excess {p^*} S\geq \Cdiff$.

The maximal running time of the algorithm
is $O(m/(\Cdiff \alpha))=\tO(m)$ without a volume bound, and
$O(s/(\Cdiff \alpha))=\tO(s)$ time with the volume bound $s$. This is
all formalized in the following theorem which is similar to results
proved in \cite{ACL07:pagerank}.
\begin{theorem}\label{thm:ACL}
We are given a multigraph with $m$ edges,
an initial mass distribution $p^\circ$ of total mass $1$, and
a listing of the vertices with positive mass in order of non-increasing density. 
Let
$p^*=\PR(\alpha,p^\circ)$. We are also given an excess parameter
$\Cdiff<1$.  We have a PageRank algorithm that staring from $p^\circ$ 
will either find a set $T$ with $\vol(T)\leq m$ and conductance
\[\Phi(T)=O(\sqrt{(\alpha\log m)/\Cdiff})\textnormal,\]
or certify that there is no set
$S$ with 
\[\excess {p^*} S\geq \Cdiff.\]
The maximal running time is $O(m/(\Cdiff \alpha))$,
but if a set $T$ is returned, then the time is also bounded by $O(\vol(T)(\log m)/(\Cdiff \alpha))$.

If we are further given a volume parameter $s\leq m\Cdiff/16$, the 
algorithm will either find the above $T$ with the additional guarantees that
$\vol(T)\leq 8s/\Cdiff$ and $\excess {p^*} T\geq \Cdiff/(16\lg (4s))$,
or certify 
there is no set $S$ with $\vol(S)\leq s$ and $\excess {p^*} S\geq \Cdiff$.
The maximal running time is $O(s/(\Cdiff \alpha))$,
but if a set $T$ is returned, then the time is also bounded by $O(\vol(T)(\log m)/(\Cdiff \alpha))$.
\end{theorem}
The proof of Theorem \ref{thm:ACL} is deferred to Section
\ref{sec:pagerank-anal}.  Without the running time and the part with
the volume parameter $s$, Theorem~\ref{thm:ACL} follows from Theorem
4.1 in \cite{ACL07:pagerank}. However, when it comes to finding
low-conductance cuts, \cite{ACL07:pagerank} is focussed on the case
where the initial distribution has all mass on a single ``good'' vertex. Here
we need to find low-conductance cuts starting from initial distributions
spreading the mass on many vertices, as supported by our Theorem
\ref{thm:ACL}. Another new point in the last part of Theorem
\ref{thm:ACL} is that the set $T$ has $\excess {p^*} T\geq
\Cdiff/(16\lg (4s))$. This will later be critical to the proof of Lemma
\ref{lem:s-captured} and \ref{lem:many-non-captured}. The statement of of Theorem \ref{thm:ACL}
is thus tailored for the needs of our min-cut algorithm, but it can be
proved using techniques from \cite{AC07:sharp-drop,ACL07:pagerank}.

\paragraph{The endgame}
More interesting and novel, we study here the ``endgame'' of the PageRank algorithm.
Suppose there is a single vertex $u$ that even in the limit receives too little mass.
Then this is because there is some low conductance cut that prevents mass from
reaching $u$. More precisely, suppose there is just a single vertex $u$ with low density
\[\ppd u\leq (1-\Cdiff)/(2m).\]
Then we will find a low conductance cut like that in Theorem \ref{thm:ACL}, that is,
a set $T$ which is the smaller side of a cut with conductance
\[\Phi(T)=O(\sqrt{(\alpha\log m)/\Cdiff}).\]
We will either find $T$ efficiently in time $O(\vol(T)(\log
m)/(\gamma\alpha))$ as in Theorem \ref{thm:ACL}, or in
time $O(m/(\gamma\alpha))$ but with the guarantee that $T$ contains
all small density vertices $u$ with $\ppd u\leq (1-\Cdiff)/(2m)$.

If there is no low density vertex $u$, then the algorithm will either return $T$ as
described above, or certify that that there is no low density vertex $u$. The maximal running time of the algorithm is $O(m/(\Cdiff \alpha))$. The full formal details are presented in the theorem below.
\begin{theorem}\label{thm:endgame}
We are given a multigraph with $m$ edges,
an initial mass distribution $p^\circ$ of total mass $1$, 
and a listing of the vertices with positive mass in order of non-increasing density. 
Let $p^*=\PR(\alpha,p^\circ)$. We are also given a parameter $\Cdiff<1$.
We have a PageRank algorithm that staring from $p^\circ$ 
will either 
find a set $T$ with $\vol(T)\leq m$ and
conductance 
\[\Phi(T)=O(\sqrt{(\alpha\log m)/\Cdiff})\textnormal,\]
or certify that there is no vertex $u$ with 
\[\ppd u\leq (1-\Cdiff)/(2m).\]
The running time of the algorithm is $O(m/(\gamma\alpha))$, and, depending
on the input, it will always end in one of the following cases:
\begin{itemize}
\item[(i)] The set $T$ is found in time $O(\vol(T)(\log m)/(\gamma\alpha))$ and
has $\excess {p^*} T\geq \Cdiff/(64\lg (8m))$.
\item[(ii)] The set $T$ is guaranteed to contain all small density vertices $u$ with
$\ppd u\leq (1-\Cdiff)/(2m)$.\\
In this case, even if $T$ is small, we have no better time bound than $O(m/(\gamma\alpha))$.
\item[(iii)] A certificate that there is no vertex
$u$ with $\ppd u \leq (1-\Cdiff)/(2m)$.
\end{itemize}
\end{theorem}
The proof of Theorem \ref{thm:endgame} is deferred to Section
\ref{sec:pagerank-anal}.
We note that if we just want a condition for finding a cut
with conductance
\[O(\sqrt{(\alpha\log m)/\Cdiff}),\]
then Theorem \ref{thm:endgame} implies Theorem \ref{thm:ACL}; for if
there is a set $S$ with excess $\gamma$ (the condition for Theorem \ref{thm:ACL}),
then the average density outside $S$ is at most $(1-\gamma)/(2m)$, so
there must exist a vertex $u$ outside $S$ with low density $\ppd u
\leq (1-\Cdiff)/(2m)$ (the condition for Theorem \ref{thm:endgame}). 
This does not
mean that Theorem \ref{thm:endgame} replaces Theorem \ref{thm:ACL},
for, on a more detailed level, we will make heavy use of the volume
parameter $s$ in Theorem \ref{thm:ACL}.

Note the big asymmetry in the conditions for the theorems. For Theorem \ref{thm:ACL},
we need 
\[p^*(S)-\frac{\vol(S)}{2m}\geq\Cdiff.\]
The volume of a vertex is its degree, so our new condition for Theorem \ref{thm:endgame} 
can be written as
\[\frac{\vol(u)}{2m}-p^*(u)\geq \Cdiff\; \frac{\vol(u)}{2m}.\]
Thus it takes much less missing mass than excess mass to find a low conductance cut. This asymmetry is
necessary, for consider an expander graph where all cuts have conductance $\Omega(1)$.
As in the example we gave before Theorem \ref{thm:ACL} against
multiplicative differences, consider an initial distribution $p^\circ$ that places
all mass on a vertex $u$ of minimum degree $\delta$. Then after the first
push, we have settled mass $p(u)=\alpha$, and in the limit $p^*(u)\geq p(u)=\alpha$.
Assuming $\alpha\geq 2/n$, we get 
\[p^*(u)-\frac{\vol(u)}{2m}\geq \alpha/2\geq \Cdiff\;\frac{\vol(u)}{2m}\mbox{ with }\gamma=n\alpha/2\textnormal,\] 
but this does imply any low conductance cut. 

The conditions of Theorem \ref{thm:ACL} and Theorem \ref{thm:endgame}
are used in a complimentary fashion. Suppose we have a low conductance
cut around a set $S$. Informally speaking, if our initial distribution
concentrates the mass well inside $S$, then a lot of mass will also be
trapped in $S$ in the limit, giving us the excess for Theorem
\ref{thm:ACL}. Conversely, if we concentrate the initial mass well
outside $S$, then only little mass will reach $S$ in the limit, and
then there will be some low density vertex $u$ in $S$ for Theorem
\ref{thm:endgame}. We still have the problem of ensuring that we don't
start somewhere in the middle, where the initial mass is neither well
inside, nor well outside $S$. This will be illustrated in the next
section when $S$ is a non-trival min-cut.

\subsection{PageRank in our applications}
In our applications, we are always going to use the same teleportation constant
\[\alpha_0=1/\lg^5 n.\]
Recall that our minimum degree is $\delta\geq \lg^6 n$, so $\delta\alpha_0\geq\lg n$.

Our initial distribution $p^\circ$ will almost always be obtained by
distributing all the mass with uniform density on some set $X$ of vertices, 
that is,
$p^\circ(v)=1/\vol(X)$ if $v\in X$; otherwise $p^\circ(v)=0$.  For
short, we call this the {\em uniform density distribution on $X$}.
We use $p^\circ_X$ to denote this initial distribution and
$p^*_X=\PR(\alpha_0,p^\circ_X)$ to denote the corresponding limit
distribution. In case we start with all mass on a single vertex $v$,
we will write $p^\circ_v=p^\circ_{\{v\}}$ and $p^*_v=p^*_{\{v\}}$. If we
need to make explicit which graph $H$ we run PageRank on, we write it
as a second subscript as in $p^\circ_{X,H}$ and $p^*_{X,H}$.

Note that Theorem \ref{thm:ACL} and \ref{thm:endgame} both
require that
the initial mass distribution $p^\circ$ is presented
with a listing of the vertices with positive mass in order of 
non-increasing density. For $p^\circ_X$, we just need to
list the vertices in  $X$ since they all have the same density,
while all other vertices have zero initial mass.


\section{Certify-or-cut}\label{sec:certify-or-cut}
In this section, using PageRank as described in Theorems \ref{thm:ACL} and
\ref{thm:endgame}, we will implement
the ``certify-or-cut'' algorithm from the introduction, proving
\begin{proposition}\label{prop:cert-or-cut} Given a simple graph with minimum degree
$\delta\geq \lg^6 n$, in near-linear time, we can either
\begin{enumerate}
\item certify that there are no non-trivial min-cuts, or
\item find a cut with conductance $o(1/\log m)$.
\end{enumerate}
\end{proposition}
Recall from the introduction that the point of the certify-or-cut is to illustrate our techniques in a
simple form on a non-trivial problem.  This is also why we will use the
same parameters as in the rest of the paper even though the min-degree
bound of Proposition \ref{prop:cert-or-cut} could easily be
reduced. When we get to our real recursive min-cut algorithm,
everything will become far more complicated.

\subsection{Starting on the small side of a min-cut}\label{sec:inside0}
Our first important observation is that if we start with a point mass on
{\em any\/} vertex $v$ on the small side $S$ of a non-trivial min-cut, and the small
side is not too large, e.g., $\vol(S)\leq m/2$, then in the limit,
we get a mass concentration on  $S$ so that Theorem \ref{thm:ACL} applies.
This should be contrasted with the results from \cite{ACL07:pagerank} which
say that if $S$ is a side of a low conductance cut, then a large fraction of
the vertices can be used as starting points leading to mass concentration.
In \cite{ACL07:pagerank} they have to guess such a good starting vertex, resulting in  a randomized algorithm. However, in our min-cut case, any vertex in $S$
will do, which is why we have a chance of a deterministic algorithm.

Note that since $S$ is a min-cut, $v$ can have at most half its edges leaving $S$, for otherwise
$S\setminus \{v\}$ would have a smaller cut around it. The result therefore
follows from the following more general lemma.
\begin{lemma}\label{lem:small-start} 
Consider a vertex $v$ in a non-trivial min-cut side $S$ with $\vol(S)\leq s$.
Suppose $S$ contains a fraction $\eps=s/(2m)+\Omega(1)$ of $v$'s neighbors
(this condition is satisfied for every vertex $v\in S$ if  $s\leq m/2$).
Then, in $\tO(s)$ time, we can find a cut with conductance
$o(1/\log m)$. If there is no such min-cut side $S$ containing $v$, 
in $\tO(s)$ time, we
will either find the cut with conductance $o(1/\log m)$, or report an error.
\end{lemma}
\begin{proof}
We start PageRank from the initial distribution $p^\circ_v$ with all mass on $v$. Then we
repeatedly push mass from $v$ until its residual mass $r(v)$ is
at most $1/\delta$. The mass from $v$ will be spread evenly to its neighbors,
so at the end, we have more than $\eps$ mass staying in $S$. Moreover,
the residual mass on any vertex is now bounded by $1/\delta$. Next
we apply the following lemma with $\mu=1/\delta$:
\begin{lemma}\label{lem:spread} If at some point the residual mass on every vertex is
bounded by $\mu$, then from this point forward, at most $\mu/(2\alpha_0)$
mass can move across any specific min-cut.
\end{lemma}
\begin{proof}
Since the minimum degree is $\delta$, the maximal residual density
is bounded by $\mu/\delta$. By Lemma \ref{lem:stationary}, from this point forward, the net
flow over any edge is at most $\mu/(2\alpha_0\delta)$. A min-cut
has at most $\delta$ edges, so the net flow across any min-cut is
therefore at most $\mu/(2\alpha_0)$.
\end{proof}
After pushing the residual mass from its starting point $v$,
by Lemma \ref{lem:spread}, the mass leaving $S$ is
at most $1/(2\alpha_0\delta)=o(1)$ since $\alpha_0=1/(\lg n)^5$ while
$\delta\geq \lg^6 n$. Thus, in the limit,
the mass staying in $S$ is $p^*_v(S)=\eps-o(1)=s/(2m)+\Omega(1)$, so
$\excess {p^*_v} S=\Cdiff=\Theta(1)$. By Theorem \ref{thm:ACL}, we now
get a set $T$ with
\[\Phi(T)=O(\sqrt{(\alpha_0\log m)/\Cdiff})=o(1/\log m).\]
in time  $O(s/\alpha_0)=\tO(s)$. This time bound is immediate from
Theorem \ref{thm:ACL} with a bound $s\leq m\Cdiff/16$, but otherwise
$s>m\Cdiff/16=\Omega(m)$, and then the general
time bound is $O(m/(\Cdiff\alpha_0))=\tO(s)$, as desired.

Since every
vertex $v$ has at least half its neighbors on its side $S$ of a non-trivial
min-cut, the conditions of the lemma are satisfied if
$1<\vol(S)\leq m/2$.
\end{proof}

\subsection{Balanced min-cut}\label{sec:balanced0}

We now consider the situation where both sides of some specific min-cut have
volume between $m/2$ and $3m/2$. We claim that there are less than $16$ vertices
that we can start from without finding a low-conductance cut.
There are at most $2\delta$ end-points of the cut edges, so there
are less than $16$ vertices incident to more than $\delta/8$ cut edges. These
are the only bad vertices. Any other vertex $v$ has at least
a fraction $\eps=7/8$ of its neighbors on its side $S$ of the min-cut.
Moreover $\vol(S)\leq s=3m/2$, so $\eps=s/(2m)+1/8$. Thus, if we
apply Lemma~\ref{lem:small-start} to a vertex $v$, in $\tO(s)=\tO(m)$ time,
we either find a cut of conductance $o(1/\log m)$, or conclude that $v$ is
bad. We run from 16 vertices. If they are all bad, we conclude that there is
no min-cut where both sides have volume at least $m/2$. All this
takes near-linear time, so to finish the proof of
Proposition~\ref{prop:cert-or-cut}, it suffices to look for non-trivial
min-cuts where the small side $S$ has $\vol(S)\leq m/2$.

\subsection{Handling any non-trival min-cut using the endgame}\label{sec:outside0}
We will now assume that we have a bound $s\leq m/2$ on volume of the small
side of any min-cut. If there is a min-cut where one side
has volume between $s/2$ and $s$, then we will find a sparse cut.
We are only interested in non-trivial min-cuts. By Observation \ref{obs:min-sparse}, the smaller side has volume at least $\delta^2$, so we
will consider $s=m/2^i$ for $i=1,...,\ceil{\lg (m/\delta^2)}$. For
a given $s$, consider a min-cut $(S,V\setminus S)$ where
$s/2\leq \vol(S)\leq s$.
We will either find a low-conductance cut, or falsify the
existence of  $(S,V\setminus S)$.

We pick an arbitrary set $U$ of $4m/(\alpha_0 s)=\tO(m/s)$
vertices. For each $v\in U$, we apply Lemma \ref{lem:small-start},
either finding a desired low-conductance cut, or determining that $v$ is
not in $S$.
The check from $v$ takes $\tO(s)$ time,
so checking all $v\in U$ takes $\tO(m)$ time. We now know that $U$
is contained in the big side $V\setminus S$ of our min-cut.

Next, we create the initial distribution $p^\circ_U$ with
uniform density on $U$. None of this mass is in $S$,
and the maximal density on any vertex in $G$ is $1/\vol(U)\leq 1/(\delta |U|)$.
Therefore, by Lemma \ref{lem:spread}, the netflow over any edge is at most
$1/(2\delta |U|\alpha_0)$, so the total mass that can move into $S$
through the at  most $\delta$ cut edges is at most
$1/(2|U|\alpha_0)=s/(8m)$, bounding the limit mass $p^*_U(S)$ on $S$.

Since $\vol(S)\geq s/2$, the average limit density on $S$ is thus
at most $1/(4m)$. It follows that some vertex $w\in S$ has limit density $p^*_U(w)\leq
1/(4m)$. This is the endgame considered in
Theorem \ref{thm:endgame}. In $\tO(m/\alpha_0)$ time, it finds a cut with conductance
$O(\sqrt{\alpha_0\log m})=o(1/\log m)$. Otherwise we conclude that
$S$ did not exist.

For each of the logarithmic number of values of $s$, we thus spend
near-linear time, so our total time bound is near-linear. If
no low-conductance cuts are found, we conclude that there is no non-trivial
min-cuts.
This completes the proof of Proposition \ref{prop:cert-or-cut}.
\drop{
In our min-cut algorithm, we are going
to recurse based on low-conductance cuts, but then it becomes important
that the time spent on finding the low-conductance cut is bounded in terms
of the volume of the smaller side unless we end in case (ii) of Theorem \ref{thm:endgame}. Changing
the parameters above, we can make sure that half the volume of $S$
gets density $\leq 1/(4m)$, and then Theorem \ref{thm:endgame} (ii) leaves
less than half the volume of $S$ on the large side.}

\subsection{Relation to the overall min-cut algorithm}
At first sight, it may seem that Proposition \ref{prop:cert-or-cut} is
a major step forward in the direction of implementing our min-cut
algorithm sketched in Section \ref{ssec:overall}. However, in the
min-cut algorithm we have a recursion where if there is a non-trivial
min-cut, then we will find a low-conductance cuts and split off the
smaller side. In order for this to be efficient recursively, the time
spent has to be near-linear in the volume of the smaller side that we
split off. However, in Proposition \ref{prop:cert-or-cut}, we may
spend time near-linear in the graph size even if we find a
low-conductance cut with a very small side. Indeed we may spend time
near-linear in the graph size when we prepare for the endgame in Section \ref{sec:outside0},
checking all the $\tO(m/s)$ vertices in the set $U$.  We spend
$\tO(s)$ time on each vertex $v$, finding either a low-conductance cut
where the small side has volume $\tO(s)$ (c.f. Theorem \ref{thm:ACL}),
or discovering that $v$ is not in $S$. The worst case would be if we only found
a low-conductance cut for the last vertex checked.  To steer around
this issue, we will carefully exploit the guarantee from Theorem
\ref{thm:ACL} that the smaller side $A$ has $\excess {p^*} A\geq
\Cdiff/(16\lg (4s))$. The details are found in the proofs of Lemma
\ref{lem:s-captured} and \ref{lem:many-non-captured}, but first we
have to set up a framework to formalize our min-cut algorithm
which has to handle many other tricky issues such as the parallel edges
created by contractions.

\section{The min-cut algorithm: the recursive set-up}\label{sec:min-cut}

In this section, we are going to present the recursive set-up
for our min-cut algorithm, doing the min-cut preserving contractions
described in Theorem \ref{thm:main-tech}:
\begin{quote}\em
Given a simple input graph $G$ with $n$ vertices, $m$ edges, and minimum degree
$\delta$,  in near-linear time, we can contract vertex sets
producing a multigraph $\bbar G$ which has only
$\bbar m=\tO(m/\delta)$ edges, yet which preserves all non-trivial min-cuts of
$G$.
\end{quote}
When first we have the contracted graph $\bbar G$, we can apply 
Gabow's min-cut algorithm \cite{Gab95} as described
in Corollary \ref{cor:min-cut} and find a min-cut of $G$ in
$\tO(\delta\bbar m)=\tO(m)$ time.

The reader may at this point want to review the sketch of our
deterministic min-cut algorithm from Section \ref{ssec:overall}.
The pseudo-code for the real algorithm is found in
Algorithm \ref{alg:min-cut}. In the top level repeat-loop,
 it works with a  multigraph
$\bbar G$ obtained from $G$ by contracting vertex sets while preserving all
non-trivial min-cuts of $G$. The only edges removed from $G$ in
the construction of $\bbar G$ are those with contracted end-points who
would otherwise be loops. However, to find out which vertex sets that
can be contracted in $\bbar G$, in each iteration of the repeat-loop, 
the algorithm works with a subgraph $H$ of $\bbar G$.

The edge connectivity of $G$ is at most $\delta$,
so if the edge connectivity of $\bbar G$ becomes
bigger than $\delta$, then there cannot be any non-trivial min-cuts
in $G$, and then we can contract $\bbar G$ to a single vertex.

We note that if there are more than $\delta$ parallel edges between vertices
$u$ and $v$, then we can trivially contract $\{u,v\}$. There are therefore
never more than $\delta$ parallel edges between two vertices in $\bbar G$.

When a vertex set is contracted to a single vertex, we call it a {\em
  super vertex\/} while the original vertices from $G$ are called {\em
  regular vertices}. If we just say a vertex it can be of either
kind. The degrees of the regular vertices (possibly including parallel edges to
super vertices) do not decrease, so regular vertices will always
have degree at least $\delta$.

\begin{algorithm}\label{alg:min-cut}
\caption{Min-cut$(G)$. Here $G$ is a simple graph with $m$ edges and
  minimum degree $\delta$}

\If{$\delta\leq\lg^5 m$}{find min-cut in $G$ using Gabow's algorithm~\cite{Gab95}.}
$\bbar G\asgn G$\tcp*{$\bbar G$ preserves non-trivial min-cuts}
\Repeat{$\geq 1/20$ of edges in $\bbar G$ 
are incident to passive super vertices}{
$H\asgn \bbar G$\;
Remove passive super vertices from $H$ and trim $H$\tcp*{Section  \ref{sec:clusters} and \ref{sec:super-degree}}
\While{some component $C$ of $H$ is not known to be a cluster}{
Find cut $(A,B)$ of $C$ with conductance $\leq \Phi_0=1/(20\lg m)$\tcp*{Section \ref{sec:cluster-or-cut}}
Remove cut edges from $H$ and trim $H$}
Take each cluster component of $H$ and contract its core to a super vertex
in $\bbar G$\;
\tcp{Preserves all non-trivial cuts from $G$ by Lemma \ref{lem:core-contract} in Section \ref{sec:semi-contract}}
\tcp{Contracts $\geq 1/2$ the edges in $\bbar G$ by Lemma \ref{lem:edges-left} in
Section \ref{sec:cut-trim}}
While there are vertices $u$ and $v$ connected by more than $\delta$ parallel
edges, contract $\{u,v\}$.
}
\tcp{$\tO(m/\delta)$ edges left in $\bbar G$ by Lemma \ref{lem:passive} in Section \ref{sec:super-degree}}
\end{algorithm}

When we analyze our algorithm, which recursively contracts vertex sets
in $\bbar G$, we will always have $n$, $m$, and $\delta$ denote the
number of vertices and edges, and minimum degree of the simple input
graph $G$. This means that our lower bound $\lg^6 n$ on the minumum
degree $\delta$ is fixed. Likewise, we will never change
$\Phi_0=1/(20\lg m)$ defining low conductance, or our teleportation
constant $\alpha_0=1/\lg^5 n$.

We assume that $n\geq n_0$ for some big enough constant $n_0$ so that $o(1)$ is smaller
any concrete constant, e.g., that $9/5+o(1)<2$ (c.f.~discussion in Section \ref{sec:notation}). 

In the presentation of the our min-cut algorithm, many (but not all)
statements will be made in terms of cuts of $\bbar G$ of size at most
$\delta$. Trivially, this includes all cuts corresponding to non-trivial min-cuts
of $G$. When we want to argue that all non-trivial min-cuts
of $G$ are preserved in $\bbar G$, we will further use that every vertex has
at least half its vertices on the same side of a non-trivial 
min-cut in $G$ (c.f.~Lemma
\ref{lem:core-contract}).

In the rest of this section, we will describe the different elements
of Algorithm \ref{alg:min-cut}, and how they work together to produce
the contracted graph $\bbar G$ from Theorem \ref{thm:main-tech}.  For
now we pretend we have an oracle to find the low-conductance cut in the
while-loop, postponing its implementation to Section
\ref{sec:cluster-or-cut}.

\subsection{Trimmed clusters}\label{sec:clusters}

Our min-cut algorithm is centered around finding clusters in $\bbar G$ defined below.

First, a non-empty set $C\subseteq V$ of vertices is called {\em trimmed\/} if
for each $v\in C$, at least $2/5$ of the edges from $v$ in $\bbar G$
don't leave $C$. The set $C$ is called a {\em cluster\/} if it is trimmed
and for every cut of size at most $\delta$ in $\bbar G$, one side
contains at most two regular vertices and no super vertices from $C$.

We say a {\em cluster $C$ belongs to a side of a cut of size at most $\delta$\/} if
that side has a super vertex or more than two regular vertices from
$C$. 
\begin{observation}\label{obs:one-side}
Any cluster $C$ belongs to unique side of any cut of $\bbar G$ of size
at most $\delta$.
\end{observation}
\begin{proof}
From the
definition of a cluster, it follows that $C$ belongs to at most one
side. However, if $C$ belonged to no side, then each side would have
at most $2$ regular vertices and no super vertices, implying that $C$
consists of at most $4$ regular vertices. However, $C$ is non-empty,
so it has at least one regular vertex which has at least $2\delta/5$
edges into $C$. If all vertices in $C$ are regular, there are no
parallel edges between them, so then we must have at least
$2\delta/5+1>4$ regular vertices in $C$. Thus we conclude that $C$
belongs to exactly one side of the cut.
\end{proof}
The condition of having all but at most two regular vertices from $C$ on the
same side of any min-cut may seem a bit ad-hoc, but we have the following
lemma stating how more than two makes a big difference.
\begin{lemma}\label{lem:big-sides} Consider a trimmed vertex set $C$ and
a cut $(T,U)$ of $\bbar G$ of size at most $\delta$. If $T\cap C$ has no super vertices
and at least 3 regular
vertices, then $T\cap C$ has at least $\delta/3$ regular
vertices.
\end{lemma}
\begin{proof}
The proof is very similar to that of Observation \ref{obs:min-sparse}.
Consider $T\cap C$ which has no super vertices.
Since $C$ is trimmed, the internal degree of regular vertices in $C$
is at least $2\delta/5$,  so the number of edges
crossing from $T\cap C$ to $U\cap C$ is
at least $|C\cap T|(2\delta/5+1-|C\cap T|)$, but we
have at most $\delta$ cut edges. Since $\delta=\omega(1)$, we conclude that $|C\cap T|\leq 2$ or
$|C\cap T|\geq 2\delta/5-1>\delta/3$. 
\end{proof}

\subsection{Cores and loose vertices}\label{sec:semi-contract}

The goal of our algorithm will be to find a family $\cC$ of
non-overlapping clusters such that the number of edges not internal to
clusters is $\bbar m=\tilde O(m/\delta)$. Contracting a ``core'' of
each cluster, as defined below, we will produce a graph $\bbar G$ with
$O(\bbar m)$ edges that preserves all non-trivial cuts of size at most
$\delta$. We can then apply Gabow's algorithm~\cite{Gab95}, and find a
minimum cut in $\tO(\bbar m\delta)=\tO(m)$ time.

Note that because the clusters in $\cC$ are required to be non-overlapping, identifying a subset of
vertices in one cluster will not stop any other cluster from being a cluster.

Consider a cluster $C$, and consider any vertex $v\in C$.  Recall
here that since clusters are trimmed, at most $3d(v)/5$ of the edges
from $v$ leave $C$ in $\bbar G$. We say a vertex $v\in C$ is {\em loose\/} if
it is regular and at least $d(v)/2-1$ of its incident edges leave $C$. 

If more than $1/4$ of the edges incident to $C$ in $\bbar G$ are internal to $A$ then we
define $A$ to be the {\em core\/} of $C$; otherwise the core of $C$ is
empty and contracting an empty core has no effect.  

The following lemma states that contracting cores preserves non-trivial min-cuts:

\begin{lemma}\label{lem:core-contract} If a non-trivial min-cut of $G$
has survived in $\bbar G$, then it will also survive  when
we contract the core of any cluster in $\bbar G$.
\end{lemma}
\begin{proof}
First we note that if a non-trivial min-cut of $G$ survives in $\bbar G$,
then it must also be a min-cut $(T,U)$ of $\bbar G$. It was a min-cut of $G$, so it has $\lambda\leq \delta$ cut edges. Also,
because it was a non-trivial cut in $G$ with at least two vertices on each
side, we must have at least two regular vertices or one super vertex both in
$T$ and in $U$.

We now consider a cluster $C$ in $\bbar G$ with a non-empty core.
Since $(T,U)$ has at most $\delta$ cut edges, by the definition of a cluster,
one side, say $T$, has at most two regular vertices and no super vertices
from $C$. We will argue that these vertices in $C\cap T$ must be
loose, hence that the vertices identified by the contraction of the core
are all in $U$, for then this contraction preserves $(T,U)$.

Let $v$ be one of the vertices from $C\cap T$, and assume for
a contradiction that $v$ is not loose.
We will prove that we get a smaller cut by moving $v$ to $U$,
contradicting that $(T,U)$ was a minimum cut. Since $v$ is regular
and both sides have at least one super vertex or two regular vertices,
$v$ is not the only vertex in $T$. Therefore we still have
a cut after moving $v$ to $U$.

Moving $v$ only affects the cutting of
edges incident to $v$. When $v$ is in $T$, we cut all edges
from $v$ to $C$, except possibly one to another regular vertex in $C\cap T$.
Since $v$ is not loose, it has more than $d(v)/2+1$ edges from $v$ into $C$, so
with $v$ in $T$, we cut more than $d(v)/2$ edges incident to $v$.
Moving $v$ to $U$, we stop cutting these edges, so we
cut less than $d(v)/2$ edges incident to $v$, contradicting
that $(T,U)$ was a min-cut.
\end{proof}

Next we want to argue that if we have a bound on the number $\bbar m$
of edges between the clusters, then the number of edges that survive
when we contract the cores is $O(\bbar m)$. This is done by the following
lemma:

\begin{lemma}\label{lem:shaving} If a cluster $C$ has $k$ edges leaving it,
then there are less than $3k$ edges incident to $C$ that are not internal to the
core. In particular, if the core is empty, we have $\vol(C)<3k$.
\end{lemma}
\begin{proof}
Let $A$ be $C$ without the loose vertices, i.e., $A$ is the core unless
the core becomes empty. Let $\ell$ be the number of edges leaving $C$ from loose vertices.
Then we have $k-\ell$ edges leaving $C$ from vertices in $A$.
Other edges incident to $C$ but not internal to $A$ are all incident
to loose vertices.

Consider any loose vertex $v$ in $C$. It has at least $d(v)/2-1=d(v)/(2+o(1))$ incident
edges leaving $C$. Here we used that loose vertices are regular, so
$d(v)\geq \delta=\omega(1)$. 
It
follows that the total number of edges incident to loose vertices is at most $(2+o(1))\ell$. Therefore, the
total number of edges not internal to $A$ is at most $(2+o(1))\ell+(k-\ell)\leq
(2+o(1))k$. This proves the lemma unless the core becomes empty.

The core becomes empty if and only if at most $1/4$ of the edges incident to $C$ are internal
to $A$, but this implies that the number of edges internal to $A$ is at
most $1/3$ of the number of edges not internal to $A$. Thus, if $A$ is not the
core, there are at most $(2+o(1))k/3$ edges internal to $A$, and then
we have at most $(2+o(1))k+(2+o(1))k/3<3k$ edges incident to $C$.
\end{proof}
Finally, we claim that each super vertex represents a lot
of edges from the orginal graph.
\begin{lemma}\label{lem:big-super}
There are $\Omega(\delta^2)$ edges from $G$ contracted in each super 
vertex of $\bbar G$.
\end{lemma}
\begin{proof}
Consider the first time a cluster $C$ with a non-empty core $A$ get
contracted into a super vertex $v^*$. By {\em first} we mean that $A$
itself does not already have any super vertices, that is, all vertices
in $A$ are regular. By definition of a core, only loose vertices from
$C$ are not in $A$, and loose vertices are regular by definition, so we conclude that 
all vertices in
$C$ are regular. But $C$ is also trimmed, so any vertex $v\in C$, has
at least $2/5$ of its incident edges staying in $C$, and they all go
to distinct neighbors since $C$ has no super vertices. Thus $|C]\geq
  2\delta/5$, and hence we have at least $2\delta^2/5$ edge end-points
  in $C$, corresponding to at least $\delta^2/5$ distinct edges. By
  definition of a non-empty core, at least $1/4$ of the edges incident
  to $C$ are internal to $A$, so we conclude that $A$ has at least
  $\delta^2/20=\Omega(\delta^2)$ internal edges that all get
  contracted into $v^*$. Now $v^*$ may later be contracted with other
  vertices, but this can only increase the number of edges contracted
  in $v^*$.
\end{proof}

\subsection{Active and passive super vertices}\label{sec:super-degree}

We say that a super vertex is {\em active\/} if it
has at least
\[\delta^*=(\lg n)\delta/\alpha_0\]
incident edges in $\bbar G$; otherwise we call it {\em passive}. 

Recall that the maximal number of parallel edges between any two vertices
is $\delta$; for if it is higher then we can contract the two vertices.
The high degree of an active super vertex $v$ therefore 
implies that at most a fraction
$\alpha_0/(\lg n)=1/\lg^6 n$ of its outgoing edges go to any single
neighbor. This means that when we push residual mass from $v$, then no
neighbor receives more than this fraction.
Thus we can spread mass from active vertices in a way similar to
what was described in Section \ref{sec:inside0} when we had no parallel
edges.  The point in the low
degrees of passive super vertices is the following good bound on the
total number of edges incident to passive super vertices.
\begin{lemma}\label{lem:passive}
The total number of edges leaving passive super
vertices is $\tO(m/\delta)$.
\end{lemma}
\begin{proof}
By Lemma \ref{lem:big-super}, we have $\Omega(\delta^2)$ internal 
edges contracted into each super vertex $v^*$. When $v^*$ is passive, the
ratio of edges leaving $v^*$ to edges contracted in $v^*$ is
at most 
$\delta^*/\Omega(\delta^2)=O((\log n)/(\alpha_0\delta))=\tO(1/\delta)$,
and this holds for every passive super vertex.
\end{proof}

Our algorithm will terminate successfully if the total number of edges in $\bbar G$ is less than 20 times the number of edges incident to passive
super vertices, for then, by Lemma \ref{lem:passive}, we have only $\tO(m/\delta)$ edges in $\bbar G$, and then, as described in Section \ref{sec:semi-contract}, we can find a min-cut of $G$ in near-linear time.

\subsection{Cut, trim, shave, and scrap}\label{sec:cut-trim}

An iteration of the repeat-loop in Algorithm \ref{alg:min-cut}
generally works by alternation between cutting edges of a subgraph $H$
of $\bbar G$ and trimming the resulting components of $H$ as described
below. We start with $H=\bbar G$. By {\em cutting} we refer to two
cases. One is where we cut out a passive super vertex, removing its
incident edges. The other is where we remove the edges of a
low-conductance cut. By {\em trimming\/} we mean removing any vertex
$v$ from $H$ that has lost more than $3/5$ of the edges it has in
$\bbar G$. When removing $v$, we also remove all its incident edges
from $H$, possibly resulting in more vertices to be trimmed. When no
more trimming is possible, each remaining vertex in $H$ satisfies
$d_H(v)\geq 2 d_{\bbar G}(v)/5$ which means that all components of $H$ are 
trimmed as defined in Section \ref{sec:clusters}.

The while-loop of Algorithm \ref{alg:min-cut} keeps cutting and
trimming until we somehow know that all remaining components $C$ in $H$
are clusters in $\bbar G$. We now {\em shave\/} off loose regular
vertices $v$ that have lost at least $d(v)/2-1$ of their incident
edges. Contrary to trimming, the shaving is not recursive. We only
shave off vertices that were loose before the shaving started. If we
recursively removed vertices $v$ that had lost at least $d(v)/2-1$ of their
incident edges, then we could easily end up losing all vertices in
thee graph.

Let $A$ be what is left of a component $C$ after shaving. If less than
$1/4$ of the edges incident to $C$ are internal to $A$, we {\em
  scrap\/} $A$ so that nothing remains from $C$. Otherwise $A$ is a
{\em core\/} that we contract in $\bbar G$.

We want to bound the number of edges cut, trimmed, shaved, and scrapped from $H$, for these are
the edges that remain in $\bbar G$ when the cores of the cluster components of $H$ 
 are contracted.

\begin{lemma}\label{lem:cut-trim} If the total number of edges cut is $c$,
then the total number of edges lost due to trimming, shaving, and scrapping
is at most $4c$.
\end{lemma}
\begin{proof}
The proof is by amortization. The ``lost degree'' of a vertex $v\in H$
is the number of incident edges in $\bbar G$ that are not in $H$.
A vertex not in $H$ has no lost degree. We are interested in the total lost degree over all vertices in $H$, and it
starts at $0$ when $H=\bbar G$.
When we cut an edge, the total lost degree increases by $2$.
When we trim a vertex $v$, it has at most $2d_{\bbar G}(v)/5$ 
incident edges left, so its lost degree is at least $3d_{\bbar G}(v)/5$.
The trimming removes $v$ from $H$, so its lost degree is saved.
However, each incident edge removed increases
the lost degree of the other end-point by $1$. Thus, when we trim $v$,
we remove at most $2d_{\bbar G}(v)/5$ edges while reducing
the total lost degree by at least $d_{\bbar G}(v)/5$.
The number of edges removed by trimming is therefore at most twice the 
decrease in the total lost degree due to triming.
The total increase lost degree by cutting is $2c$ and if the total lost degree is $d$
after all the trimming is done, then the total number of
trimmed edges is at most $2(2c-d)$.

Consider a cluster $C$. At this point, its lost degree $d_C$ is
exactly the number of edges leaving $C$ in $\bbar G$, and summing over
all clusters $C$ in $H$, we have $\sum_C d_C=d$. By Lemma
\ref{lem:shaving}, there are at most $3d_C$ edges incident to a cluster $C$ in
$\bbar G$ that are not internal to the core of $C$. We have already
removed the $d_C$ edges leaving $C$, so when we shave $C$ down to its
core, scrapping it if too small, then we remove at most $2d_C$
edges. Summing over all clusters, the shaving and scrapping takes out
at most $\sum_C 2d_C=2d$ edges.

Summing up, the trimming, shaving, and scrapping of undersized
cores, takes out a total of at most $2(2c-d)+2d\leq 4c$ edges.
\end{proof}
As mentioned above, we start the round with $H=\bbar G$. As described
at the end of Section \ref{sec:super-degree}, we are done if more
than a fraction $1/20$ of the edges are incident to passive super
vertices. Otherwise, we cut all edges incident to passive 
super vertices, and trim
the sides.

Next we are repeatedly going to cut and trim using cuts of components of $H$
of conductance at most
\[\Phi_0=1/(20\lg m).\]
This is what we henceforth
regard as a {\em low-conductance\/} cut. Later sections will prove that low-conductance
cuts can be found efficiently if a component is not a cluster.

We claim that the total
number of edges cut this way is at most a fraction $1/20$ of the
edges in $\bbar G$. The point is that the number of edges cut
is a fraction $1/(20\lg m)$ of the volume of the small side, and
the same vertex can end on the smaller side only $\lg m$ times. Here
size is measured by volume, that is, number of incident edges.

Including the at most $1/20$ of the edges of $\bbar G$ incident to passive vertices, we thus cut at most a
 fraction $1/10$ of the edges in $\bbar G$.
Hence, by Lemma \ref{lem:cut-trim}, in total,
we lose at most $1/2$
of the edges in  $\bbar G$. Summing up,
\begin{lemma}\label{lem:edges-left} Cutting edges around passive vertices and edges of
low-conductance cuts, trimming, shaving, and scrapping, leaves at least
half the edges of $\bbar G$ in the resulting cluster cores of $H$.
Therefore, when we contract the cluster cores of $H$ in $\bbar G$ (c.f. Algorithm  \ref{alg:min-cut}), we reduce the number of edges in $\bbar G$ by at least 
a factor $2$.
\end{lemma}

\subsection{Graph representation and modification}\label{sec:representation}
We will now discuss how to represent the graphs involved in Algorithm
\ref{alg:min-cut}. Our representations are all fairly standard. They
ensure that all the graph manipulations we do in Algorithm
\ref{alg:min-cut} can be supported in near-linear total time. The
repesenations will also help us when we later search for
low-conductance cuts.

We assume that the simple input graph $G=(V,E)$ is represented with a list
of neighbors for each vertex. We also assume some linear ordering of $V$.

\paragraph{The contracted graph \boldmath{$\bbar G$}}
The graph $\bbar G$ is obtained from $G$ by contraction of vertex
sets.  We think of $\bbar G=(\bbar V,\bbar E)$ as a multigraph with
parallel edges, and the degree $d_{\bbar G}(v)$ of a
vertex in $v$ counts all incident edges, including parallel ones.
However, internally, we represent $\bbar G$ as a simple
graph with multiplicities on the edges. To represent the connection to
$G$, we store with each vertex in $\bbar V$, the set of original vertices it
represents $V$, and likewise for the edges.  Recall here that a vertex
$v\in \bbar V$ is a regular vertex if it represents a single original
vertex and that $v$ is a super vertex if it represents several original vertices.

Initially $\bbar G=G$, and we will always have $\bbar V\subseteq V$.
Each $v\in \bbar V$ will have its current neighbors, each with a positive
multiplicity, stored in a doubly-linked list sorted by the ordering of
$V$. The list ordering is maintained by a balanced binary search trees so that
we can insert and delete neighbors in $O(\log n)$ time.

We will always do the contractions for two vertices at the time, that
is, if we have a set $\{v_1,\ldots,v_s\}$ to be contracted, we create
a list of pairs $\{v_1,v_2\}, \{v_2,v_3\},\ldots \{v_{s-1},v_s\}$ to
be contracted.

To make the contractions efficient, for each $v\in \bbar V$, we
maintain the sum of degrees of the original vertices represented by
$v$. We note that loops created by the contractions are counted twice
even though we remove them from $\bbar G$. When asked to contract $u$
and $v$, we will use the one with the bigger degree sum, say $u$, as
the new representative.  As the first step in the contraction, if
there is an edge $(u,v)$, we remove it as it would otherwise
become a loop. More precisely, we search and remove $v$ from $u$'s incidence
list, and vice versa, all in $O(\log n)$ time thanks to the binary
search trees.  Next we have to take every edge $(v,w)$ from $v$ and
turn it into an edge $(u,w)$ from $u$. This requires replacing
$v$ by $u$ in $w$'s neighbor list, and adding $w$ to $u$'s neighbor list.
Thanks to the ordering, we immediately discover if we already
had a $(u,w)$ edge, in which case we merge them into a single edge,
adding their multiplicities. This is all done
in $O(\log n)$ time per edge incident to $v$. What makes the whole thing
efficient is that
after the contraction, the degree sum of $u$ is at least twice the
degree sum we had for $v$. Since the degree sum cannot exceed $2m$, we
conclude that each edge can be moved at most $\log_2(2m)$ times over
all contractions, adding up to a total cost of $O(m\log^2 n)$.

As we do the contractions, we will build a contraction tree, making
$u$ a parent of $v$ when $u$ and $v$ are contracted into $u$ as above.
The roots are then the current vertices in $\bbar V$ and their
descending nodes are the original vertices they represent. The
contraction tree has height at most $\log_2(2m)$.

Finally, recall that if we get more than $\delta$ parallel edges between
two vertices $u$ and $w$, then we want to contract $u$ and $w$. We easily detect
this as we merge edges, adding up their multiplicities, and this will result
in a list of pairs to be contracted. When get to such a pair 
$\{u,w\}$, we go to the contraction tree and find the roots $u'$ and $w'$
of $u$ and $w$. If they have the same root $u'=w'$, there is nothing
to be done. Otherwise, we perform the contraction of $u'$ and $w'$ as 
described above. It only takes $O(\log n)$ time to find the roots, so
this does not affect our overall time bound of $O(m\log^2 n)$ to
support all contractions in $\bbar G$ done in Algorithm \ref{alg:min-cut}.

\paragraph{The subgraph \boldmath{$H$} of \boldmath{$\bbar G$}}
Inside the repeat-loop of Algorithm \ref{alg:min-cut}, we work with a
subgraph $H$ of $\bbar G$. Initially, we just copy $\bbar G$ into
$H$. Let $\bbar n$ and $\bbar m$ denote the number of vertices and
edges we get from $\bbar G$.

The subgraph $H$ is decremental in the sense that edges are only
removed, not added. We will now describe a dynamic data structure to
maintain information on $H$. The data structure only supports
the deletion of one edge at the time. If a vertex becomes isolated,
it is deleted automatically.

With each vertex $v$, we store its initial degree $d_{\bbar G}(v)$
from $\bbar G$ as well as its current degree $d_H(v)$ in $H$.  For
every component $C$ of $H$, we will have a doubly-linked list of its
vertices and a balanced binary search tree of height $O(\log \bbar n)$ over
this vertex list. For each node in the binary search tree, we maintain the
number of descending vertices, the sum of their degrees in $H$, and the
sum of their degrees in $\bbar G$. For the root node of $C$, these
numbers represent the number of vertices in $C$, $\vol_H(C)$, and
$\vol_{\bbar G}(C)$. Since no edges leave a component of $H$, we
have $m(C)=\vol_H(C)/2$. {\em The root is viewed as representating the component $C$}. To check which component a vertex belongs, we 
simply follow parent pointers to the root in $O(\log n)$ time.

To maintain the components of $H$, we use the dynamic connectivity algorithm
from \cite{HLT01} supporting each edge update in $O(\log^2 \bbar n)$ time.
We will only use it decrementally in $\tO(\bbar m \log^2 \bbar n)$ total time.
The algorithm from  \cite{HLT01} maintains a spanning forest of $H$. In
connection with each edge deletion $(u,v)$, it can tell us if $(u,v)$ was a bridge,
and if so, which of $u$ and $v$ that ends in the smaller component.

When we remove an $(u,v)$ edge from $H$, our first action is to
decrement $d_H(u)$ and $d_H(v)$ as well as the degree sums of their
ancestors in the binary search trees, all in $O(\log \bbar n)$ time. If
$(u,v)$ was a bridge in a component $C$, the algorithm from
\cite{HLT01} points us to (a spanning tree of) the smaller new
component $A$. We
extract these vertices, one by one, from the vertices in $C$ and from
the binary search tree over $C$. Next we build a new binary search
tree for the vertices in $A$. All this is done in $O(\log n)$ time
time per vertex, including all necessary updates to counters in the
binary search trees. Every time a vertex is moved, it ends up in a
component of half the size, so the total of moving vertices
to smaller components is $O(\bbar n\log^2\bbar n)$. Adding up,
the total time spent on our decremental data structure for $H$ is $O(\bbar m\log^2\bbar n)$.

Returning to an iteration of the repeat-loop in Algorithm
\ref{alg:min-cut}, first we copy $\bbar G$ to $H$, initializing the
above decremental data structure. 
We will generally have a deletion list for vertices to be deleted. When
we get to a vertex in the list, we remove all its incident
edges, one by one, before removing the vertex.

The first vertices in the deletion list are the passive
super vertices, that is, super vertices $v$ with 
$d_{\bbar G}(v)<\delta^*$. Later we add vertices
to be trimmed to the list, that is, any vertex $v$ getting
$d_{H}(v)<2d_{\bbar G}(v)/5$. The trimming is not completed
until the deletion list is empty. 

When we in the while-loop find a low conductance cut, we 
remove the cut edges one by one. When all cut edges have been removed,
we do the trimming as described above. 

Recall that we for now ignore the cost of finding the low-conductance
cuts in the while-loop. Then the total time
spent on reducing $H$ to clusters is $O(\bbar
m\log^2\bbar n)$.

Finally, when we are down to clusters, we want to identify the
cores to be contracted. This is easily done sequentially in $O(\bbar m)$ total
time. More precisely, for each cluster $C$, we first go through
all the vertices, identifying the set $L$ of loose vertices, that is,
vertices with $d_H(v)\leq d_{\bbar G}(v)/2+1$. Let $A=V(C)\setminus L$.
Next we check if at least $1/4$ of the edges incident to $C$ in $\bbar G$
are internal to $A$. If so we, $A$ is a core to be contracted in $\bbar G$; 
otherwise the core is empty with nothing to be contracted from $C$.

Summing up what we have proved in this section, we have
\begin{lemma}\label{lem:basic-recurse}
Ignoring the cost of finding the low-conductance cuts, we can implement Algorithm
\ref{alg:min-cut} in $O(m\log^2 n)$ time where $n$ and $m$ are the
number of vertices and edges of the simple input graph $G$. The result
is the contracted graph $\bbar G$ described in Theorem
\ref{thm:main-tech}, that is, $\bbar G$ has only $\tO(m/\delta)$
edges, yet it preserves all non-trivial min-cuts in $G$.
\end{lemma}
\begin{proof}
Concerning the running time, above we implemented all the
contractions in $\bbar G$ in $O(m\log^2 n)$ total time. Moreover, we
implemented each iteration of the repeat-loop in $O(\bbar m\log^2\bbar n)$
time where $\bbar n\leq n$ and $\bbar m\leq m$ are the number of vertices
and edges in $\bbar G$ in the beginning of the iteration. By 
Lemma \ref{lem:edges-left}, $\bbar m$ is reduced by
at least a factor $2$ by the contractions to $\bbar G$ at the end of
each iteration. Therefore, the total time spent is $O(m\log^2 n)$.

Concerning correctness, Lemma \ref{lem:core-contract} ascerts
that the contraction of cluster cores preserve all non-trivial min-cuts of
$G$. Moreover, Lemma \ref{lem:passive} states that we
have only $\tO(m/\delta)$ edges leaving the passive super vertices,
and the repeat-loop only terminates when these edges constitute at least
a fraction $1/20$ of all edges, so we conclude that the
final contracted graph $\bbar G$ has only $\tO(m/\delta)$ edges.
\end{proof}

\section{Certifying clusters and finding low conductance cuts}\label{sec:cluster-or-cut}

We will now, at a high level, describe the process that repeatedly
takes a trimmed component $C$ of $H$, cuts the edges of a
low-conductance cut and trims the sides, stopping only when all
remaining components are known to be clusters. This implements the
while-loop in Algorithm \ref{alg:min-cut}.

In the process, {\em $H$ is a subgraph of $\bbar G$ with no passive
  super vertices from $\bbar G$}. Since $H$ starts trimmed and since
we trim after each cut, we know that all components of $H$ are trimmed
when we look for low-conductance cuts.

We need a measure for how close trimmed components of $H$ are at being clusters, but
it is useful with a more general mearsure that applies to all subgraphs $B$ of $H$.
We say $B$ is {\em $s$-splittable\/}, or
has {\em splittability $s$}, if
every cut $(T,U)$ of $\bbar G$ with at most $\delta$ cut edges
has $\min\{\vol_{B}(T\cap B),\vol_{B}(U\cap B)\}\leq s$. We want
$s$ to be as small as possible, and note that we always
have $s\leq m(B)$. A very important part of this definition is that it is inherited by
subgraphs, that is, if $A$ is a subgraph of $B$ and $B$ is $s$-splittable,
then $A$ is also $s$-splittable. Being $s$-splittable is thus preserved as we cut and trim.
Let
\[s_0=64\delta/\alpha_0\]
Our goal will be to partition $H$ into $s_0$-splittable trimmed components,
for they are then all clusters:
\begin{lemma}\label{lem:super-splittable}
If a trimmed component $C$ of $H$ is $s_0$-splittable,
then $C$ is a cluster.
\end{lemma}
\begin{proof}
Suppose that $C$ is not a cluster. Then there is a cut $(T,U)$ of $\bbar G$
with at most $\delta$ cut edges and such that both $T\cap C$ and
$U\cap C$ contain a super vertex or at least 3 regular vertices.
Consider $T\cap C$. Suppose $T\cap C$ contains a super vertex $v$.
Since $H$ has no passive super vertices, $d_{\bbar G}(v)\geq \delta^*=(\lg n)\delta/\alpha_0$. 
Since $C$ is trimmed, $d_C(v)\geq 2d_{\bbar G}(v)/5$. Therefore
$\volC(T)\geq d_C(v)\geq 2\delta^*/5\gg s_0$. Suppose instead that
$T\cap C$ contains no super vertices but at least three regular vertices. Then
by Lemma \ref{lem:big-sides}, there are at least $\delta/3$ regular vertices in $C\cap T$.
Since $C$ is trimmed, each of them has degree at least $2\delta/5$ in $C$, so
we conclue that $\volC(T)\geq 2\delta^2/5$. Since $\delta\alpha_0\geq \lg n$,
we again conclude that $\volC(T)\gg s_0$. The same argument holds for $S=U\cap C$, so
we conclude that 
$C$ is not $s_0$-splittable.
\end{proof}
As we cut and trim components into clusters, for each component $C$ of
$H$, we record the smallest $s$ for which we have {\em certified\/}
that $C$ is $s$-splittable.  By Lemma \ref{lem:super-splittable}, we 
{\em certify that $C$ is a cluster if $s\leq s_0$}. For larger $s$, we will apply the theorem below.
\begin{theorem}\label{thm:recurse} Let $s\in [s_0,m(C)]$ and $C$ be an $s$-splittable trimmed component of $H$.
We have an algorithm that, depending
on the input, will do one of the following:
\begin{itemize}
\item[(i)] Find a set $A\subseteq V(C)$ with $\volC(A)\leq m(C)$ and $\PhiC(A)
=o(\Phi_0)$
 in time $\tO(\volC(A))$.
\item[(ii)] Find a set $A\subseteq V(C)$ with $\volC(A)\leq m(C)$ and
$\PhiC(A)=o(\Phi_0)$ in time $\tO(m(C))$ certifying that the large side
$B=V(C)\setminus A$ is $s/2$-splittable.
\item[(iii)] Certifying in $\tO(m(C))$ time that $C$ is $s/2$-splittable.
\end{itemize}
\end{theorem}
The proof of Theorem \ref{thm:recurse} is complicated and deferred
to Sections \ref{sec:cross-cut}--\ref{sec:large}. 

In Algorithm \ref{alg:find-clusters} we use Theorem~\ref{thm:recurse}
to efficiently implement the while-loop from Algorithm \ref{alg:min-cut}
while carefully recording the certified splittability of each component.

\begin{algorithm}\label{alg:find-clusters}
\caption{Partitions $H$ into certified clusters as in the while-loop of Algorithm \ref{alg:min-cut}. Initially $H$ is
a trimmed subgraph $H$ of $\bbar G$ with no passive super vertices.}
Each component $C$ of $H$ is certified  $m(C)$-splittable\;
\While{some component $C$ of $H$ certified $s$-splittable with $s>s_0$}{
Apply Theorem \ref{thm:recurse} to $C$. In case (iii) we certify
that $C$ is $s/2$-splittable\;
In case (i) and (ii) we get the small side $A$ of a low-conductance
cut, remove the cut edges and trim the sides\;
\tcp{Note that both sides may now have been split into several
components}
In case (i), all new components $D\subseteq C$  are certified $\min\{s,m(D)\}$-splittable\;
In case (ii), all new components $D\subseteq A$ are certified $\min\{s,m(D)\}$-splittable\\
 and all new components $D\subseteq C\setminus A$ are certified $\min\{s/2,m(D)\}$-splittable}
\end{algorithm}
\begin{lemma}\label{lem:oracle} Using Theorem \ref{thm:recurse}, 
Algorithm \ref{alg:find-clusters} 
implements while-loop of Algorithm \ref{alg:min-cut}: starting from a
trimmed subgraph $H$ of $\bbar G$ with no passive super vertices from
$\bbar G$, it repeatedly removes edges from a low-conductance cuts,
trimming the sides, until all components of $H$ are
certified clusters. The total time spent is near-linear in the initial
number of edges in $H$.
\end{lemma}
\begin{proof} Let $n_H$ and $m_H$ be the initial number of
vertices and edges in $H$. It follows directly from the description of
Algorithm \ref{alg:find-clusters}  that the only action it performs on $H$ is to remove the
edges of a low-conductance cut and trim the sides.
From Section \ref{sec:representation}, we know that the total
time spent on modifying the representation of $H$ is $O(m_H\log^2 n_H)$.
What has been added is the certified splittability of components. When it gets down
to $s_0$-splittable components, we know they are all clusters by Lemma
\ref{lem:super-splittable}, so this a faithful implementation of
the while-loop from Algorithm \ref{alg:min-cut}.

To amortize the cost of cutting and
trimming $H$ into clusters, we say that {\em an edge $e$ pays\/} in
each of the following events: (1) when $e$ is removed from $H$, (2)
the edge $e$ gets into a component of $H$ of half the volume, or (3) the edge
$e$ gets into a component certified to be only half as splittable. 
The first event happens only once, and the last
two events can happen at most $\lg m_H$ times each, so an edge can pay at
most $2\lg m_H+1$ times. If the payment each time is $\tO(1)$, then the
total time is $\tO(m_H)$, as desired.

Looking at the cost of applying Theorem \ref{thm:recurse}, in
case (i), the edges incident to $A$ in $C$ are now either removed or
in a component of at most half the size, so they can pay $\tO(1)$ each
to cover the $\tO(\volC(A))$ time spent in this case.  Case
(iii) certifies that $C$ is $s/2$-splittable, so the $\tO(m(C))$ time
can be paid by the edges in $C$. Finally, case (ii) certifies that
$B$ is $s/2$-splittable. Before removing any edges from $C$,
we had  $\volC(B)\geq m(C)\geq\volC(A)$, hence at least $m(C)/2$
edges incident to $B$. Each of these edges is either removed as
a cut edge or by trimming, or it ends up in a component that is only
half as splittable, so they can pay $\tO(1)$ each
to cover the $\tO(m(C))$ time spent. Thus we conclude that
the total time spent in Theorem \ref{thm:recurse} is $\tO(m_H)$.

Finally we need to consider the cost of maintaining the certified
splittablity of each trimmed component $C$ encountered in the
while-loop. Recall from Section \ref{sec:representation} that $C$ is
represented as the root node of a balanced binary search tree over all
vertices in $C$. This root also knows the number $m(C)$ of edges
in $C$. We now further store the splittability of $C$ with the root
representing $C$. 

For the while-loop we maintain a ``to-do'' list
with all the roots of components with splittability $s>s_0$. We are
done when this list is empty. 
For each iteration of the while-loop, we take a component $C$ from
the to-do list, copy its splittability $s$, and apply Theorem
\ref{thm:recurse}.  If we get case (iii), we simply set the
splittablity of $C$ to $s/2$ in constant time, removing $C$ from
the to-do list if $s/2\leq s_0$.

In case (i) and (ii) we get the small side $A$ of a low-conductance
cut, remove the cut edges, and trim the sides. This can result in many
new components. To track these, recall from Section
\ref{sec:representation} that we via the dynamic connectivity
algorithm from \cite{HLT01} will find out which of the edges removed
were bridges whose removal split a component. We record all the bridges
removed. When the trimming is completed, for each bridge end-point
$u$, in $O(\log m_H)$ time,  we find its root which represents its new component $D$. In case
(i) we just set the splittability of $D$ to $\min\{s,m(D)\}$. In
case (ii), we first check if $u$ is in $A$. If so, we set 
the splittability of $D$ to $\min\{s,m(D)\}$; otherwise we
set it to $\min\{s/2,m(D)\}$. The root of $D$ is added to the to-do
list if its splittability is above $s_0$. We note
that there
may several bridge end-points in $D$, leading us to the same root, but we only need to do the above certification
once.

The total number of bridges removed in the above process is bounded
by $n_H-1$ and the total number of trimmed components considered is 
bounded by $\bbar 2n-1$, so the total time spent on maintaining
the certifiability of trimmed components is $O(n_H\log n_H)=
O(m_H\log n_H)$. Adding up, we conclude that the total cost
of implementing the while-loop is near-linear in the original number
of edges in $H$.
\end{proof}
Combining Lemma \ref{lem:basic-recurse} and \ref{lem:oracle},
we conclude:
\begin{lemma}\label{lem:recurse-main-tech} Theorem \ref{thm:main-tech} follows from Theorem \ref{thm:recurse}.
\end{lemma}
\begin{proof}
Lemma \ref{lem:oracle} includes the cost of finding low-conductance
cuts ignored in Lemma \ref{lem:basic-recurse}. Assuming Theorem
\ref{thm:recurse} it says that we can implement
the while-loop of Algorithm \ref{alg:min-cut} in near-linear time.

The while-loop is inside an iteration of the repeat-loop in Algorithm
\ref{alg:min-cut} where it is applied to a subgraph $H$ of the
contracted graph $\bbar G$. In the first iteration, $\bbar G$ is
identical to the original simple graph with $m$ edges, and by Lemma
\ref{lem:edges-left}, the edge number of $\bbar G$ is halved
between iterations. Therefore, in the while-loop,
over all iterations of the repeat-loop, the total time spent is near-linear in
$m$. From Lemma \ref{lem:basic-recurse}, we know that everything else
is handled in time $O(m\log^2 n)$, so we conclude that Algorithm
\ref{alg:min-cut} is implemented in near-linear total time.

Lemma  \ref{lem:basic-recurse} further states that the contracted graph $\bbar G$ produced by Algorithm \ref{alg:min-cut} has all the desired properties from Theorem \ref{thm:main-tech}, which hence follows if
we can prove Theorem
\ref{thm:recurse}.
\end{proof}
The rest of this section is devoted to the proof of Theorem
\ref{thm:recurse}.

\subsection{Pushing from a vertex across a small cut---the issue of parallel edges}\label{sec:cross-cut}

We are now going to introduce a basic technical lemma that we shall
use to find low-conductance cuts. It corresponds to  Lemma \ref{lem:spread}
from Section \ref{sec:inside0}, but now we have
to handle active super vertices.
A new issue is that a vertex might now have many
parallel edges to a few neighbors. We cannot handle this situation
in general, but in our case, we will argue that it has to be a regular vertex where
the parallel edges all go to super vertices, and this special structure will
be critical to our solution.
\begin{lemma}\label{lem:cross-cut}
Consider a trimmed component $B$ of $H$, and let $S$
be one side of a cut of $B$ with $O(\delta)$ cut edges.
Start PageRank in $H$ with all mass on a
vertex $v$ in $S$ and push
to the limit. If $v$ is a super vertex, the mass leaving
$S$ is $o(1)$. If $v$ is a regular vertex with a fraction $\eps$ of its
edges leaving $S$,
then the mass leaving $S$ is $\eps+o(1)$.
\end{lemma}
\begin{proof} Suppose first that $v$ is a super vertex. 
Since all super vertices are active, $v$ has at least $\delta^*=(\lg n)\delta/\alpha_0$ incident edges
in $\bbar G$, and $B$ is trimmed, so $v$ has at least $2\delta^*/5$
incident edges in $B$. The cut has $O(\delta)$ edges,
so the fraction of edges from $v$ leaving $S$ is 
$O(\alpha_0/\log n)=o(1)$.

We now first push all the initial mass from $v$. The mass is spread
evenly over its incident edges, so the mass escaping $S$ is
$o(1)$. Moreover, since the maximal number of parallel edges between
any pair of vertices is $\delta$, the maximal residual mass ending at
any vertex is $\delta/(2\delta^*/5)=O(\alpha_0/\log n)$. The minimum
degree in $B$ is $2\delta/5$, so we end up with a maximum
residual density of $O(\alpha_0/(\delta\log n))$.

By Lemma
\ref{lem:stationary}, from this point forward, the net flow 
over any edge is bounded by
$O(\alpha_0/(\delta\log n))/(2\alpha_0))=O(1/(\delta\log n))$, so the net flow over the $O(\delta)$  cut edges is bounded by $O(1/\log n)=o(1)$. Adding
in the $o(1)$ mass leaving $S$ directly at the first push from $v$, we get that the
total mass leaving $S$ is $o(1)$.

We now consider the case where $v$ is a regular vertex and where a
fraction $\eps$ of its incident edges leave $S$. As above, we first
push all the mass from $v$, sending a fraction $\eps$ of the mass out
of $S$.  We will now study what happens with the remaining residual
mass. Recall that the mass pushed from $v$ is distributed
evenly along the edges leaving $v$. We now partition the residual
mass, recalling from \cite{ACL07:pagerank} that pushing mass to the
limit is a non-negative linear transformation. We can therefore study
what happens to different parts separately.

Consider the mass $r$ that the regular vertex $v$ pushed to
its regular neighbors. There are no parallel
edges between regular vertices, and $v$ has degree at least
$2\delta/5$, so the residual mass at any regular neighbor 
is at most $5/(2\delta)$ and residual density at most $25/(4
\delta^2)=O(1/\delta^2)$. By Lemma \ref{lem:stationary}, starting
from $r$, the net flow over any edge 
is 
$O(1/(\alpha_0\delta^2))$, so the mass leaving $S$ over the $O(\delta)$ cut
edges is $O(1/(\alpha_0\delta))=o(1)$.

For each super neighbor $v_i$ of the regular vertex $v$, let $r_i$ be
the residual mass first pushed from $v$ to $v_i$.  If $v_i$ is outside $S$, we
already count $r_i$ as lost from $S$ in the initial push from $v$, so
we can assume that $v_i$ is inside $S$. Our analysis above shows that
when we push mass starting from a super vertex in $S$, then the mass
leaving $S$ is only a fraction $ o(1)$, so in this case
$o(r_i)$. However, $\sum_i r_i<1$, so when we add up the limit
distributions, we conclude that only $o(1)$ mass leaves $S$ after the
initial loss of $\eps$ to the neighbors of $v$ outside $S$.
\end{proof}

\subsection{Starting from a captured vertex}
Consider a vertex $v$ in a trimmed component $C$.  We say $v$ is {\em
  captured\/} if there is a set $S\subseteq V(C)$ with $s_0\leq
\volC(S)\leq m(C)$ and $|\boundC(S)|\leq \delta$ such that $S$
contains $v$ and at least $\frac 34$ of the edges incident to $v$. If
$\volC(S)\leq s$, we further say that $v$ is {\em $s$-captured}.

Finding a low-conductance cut is easy if we can somehow guess
a captured vertex. More precisely, using Lemma \ref{lem:cross-cut} and
Theorem \ref{thm:ACL}, we will prove:
\begin{lemma}\label{lem:s-captured} Given a vertex $v$ in $C$ and 
a parameter $s\in [s_0,m(C)]$, we have an algorithm that, depending on the input, 
will do one of the following:
\begin{itemize}
\item[(i)] Find a set $A\subseteq V(C)$ with $\PhiC(A)=o(\Phi_0)$ and
$\volC(A)\leq m(C)$ in time $\tO(\volC(A))$.
If $s\leq m(C)/32$, we will further have $\volC(A)\leq 16s$ and
$\excessC {p^*_{v,C}} A\geq 1/(32\lg(4s))$. Recall that 
$p^*_{v}$ is the limit distribution when we run PageRank on $C$
starting with all mass on $v$.
\item[(ii)] Certify in $\tO(s)$ time that $v$ is not $s$-captured.
\end{itemize}
\end{lemma}
\begin{proof}
By Lemma \ref{lem:cross-cut},
if we start PageRank with all mass on a vertex $v$ that is $s$-captured, and push
mass to the limit, we know that
$3/4-o(1)$ of the mass will stay in $S$. Since $\volC(S)\leq m(C)$, this
corresponds to
an excess of at least $3/4-o(1)-1/2>1/5$.
Thus, by Theorem \ref{thm:ACL} with $G=C$ and $\Cdiff=1/5$,
we find a cut
with small side $A=T$ and conductance
$\Phi_{C}(A)=O(\sqrt{(\alpha_0\log m)/\Cdiff})=O(1/\log^2 m)=o(\Phi_0)$
in time $\tO(\vol_{C}(T)/\alpha)=\tO(\vol_{C}(A))$.

Now, if $s\leq m(C)/32$ and $\volC(S)\leq s$, the mass $3/4-o(1)$
corresponds to an excess of at least $3/4-o(1)-1/64>1/2$. Thus
we can use $\Cdiff=1/2$ in  Theorem \ref{thm:ACL}, noting
that $s\leq m(C)/32=m(C)\Cdiff/16$. Then
$\volC(A)\leq 16 s$, and $\excessC {p^*_{v,C}} A\geq 1/(32\lg(4s))$.
\end{proof}

\subsection{Starting from set of non-captured vertices}
Next we consider the case where we somehow manage to guess a large
set $X$ of vertices that are not $s$-captured. If
$\vol_C(X)\geq
\min\{16 \delta m(C)/(\alpha_0s),m(C)/4\}$, then the
lemma below states that we have an algorithm providing all
the guarantees required for Theorem \ref{thm:recurse}.
\begin{lemma}\label{lem:not-s-captured} Let $s\in [s_0,m(C)]$
and $C$ be a trimmed component of $H$.
Let $X\subseteq V(C)$ be a set of at vertices, none of
which are $s$-captured in $C$, and with total volume $\vol_C(X)\geq
\min\{16 \delta m(C)/(\alpha_0s),m(C)/4\}$. We have an algorithm that, depending on the input, has one of the 
following outcomes:
\begin{itemize}
\item[(i)] Find a set $A\subseteq V(C)$ with $\volC(A)\leq m(C)$, $\excessC{p^*_{X,C}} A \geq 1/(128\lg (8m))$,  and
$\PhiC(A)=o(\Phi_0)$
in time $\tO(\volC(A))$. Recall 
that 
$p^*_{X,C}$ is the limit distribution when we run PageRank on $C$
starting with all mass spread on $X$ with uniform density $1/\vol_C(X)$.
\item[(ii)] Find a set $A\subseteq V(C)$ with $\volC(A)\leq m(C)$ and
$\PhiC(A)=o(\Phi_0)$ in time $\tO(m(C))$, and certify for the large side $B=V(C)\setminus A$,
that every set $S$ in $C$ with $|\boundC(S)|\leq \delta$ and $\volC(S)\leq s$
has $\volC(S\cap B)\leq s/2$. In particular, if 
the component $C$ was $s$-splittable, then the large side $B$ is certified 
$s/2$-splittable.
\item[(iii)] Certify in $\tO(m(C))$ time that there
is no set $S$ in $C$ with $|\boundC(S)|\leq \delta$ and $s/2<\volC(S)\leq s$.
In particular, if 
the component $C$ was $s$-splittable, then  $C$ is now certified $s/2$-splittable.
\end{itemize}
Above, it is only case (ii) and (iii) that depend on the
assumption that no vertex in $X$ is $s$-captured, that is, if
$X$ does contain some $s$-captured vertex, then the algorithm will still
return in one of three cases where (i) is exactly as stated above
while cases (ii) and (iii) are not to be trusted.
\end{lemma}
\begin{proof}
We are going to start PageRank from the initial
distribution $p^\circ_{X,C}$ with all mass spread on $X$ with
the uniform density $1/\vol_C(X)$.  The limit distribution is denoted
$p^*_{X,C}$. Since $C$ is trimmed, the minimum degree in $C$ is
$2\delta/5$.

Exploiting that no vertex $v\in X$ is $s$-captured, we will argue that
only little mass can end in a set $S\subseteq V(C)$ with
$|\boundC(S)|\leq \delta$ and $\volC(S)\leq s$. We assume for now that
such a set $S$ exists and that $\volC(S)>s/2$, hence that we are not
in case (iii).

First we bound the volume of the
vertices from $X$ in $S$. Consider a vertex $v\in X\cap S$.
Since $v$ is not $s$-captured in $C$, it
has at least $1/4$ of its  edges in $C$ leaving $S$, so we conclude that
$\vol_C(X\cap S)\leq 4|\boundC(S)|$. It follows that the total initial
mass in $S$ is $p^\circ_{X,C}(S)\leq 4|\boundC(S)|/\vol_C(X)$.

The maximal initial density on all vertices is 
$1/\vol_C(X)$, so by
Lemma \ref{lem:stationary}, the
net flow over any edge is at most $1/(2\vol_C(X)\alpha_0)$. Hence the total
net flow into $S$ 
at most $|\boundC(S)|/(2\vol_C(X)\alpha_0)$. The final limit mass on $S$
is thus 
\begin{align}\label{eq:limit-in-S}
p^*_{X,C}(S)&\leq (4+1/(2\alpha_0))|\boundC(S)|/\vol_C(X)\\
&<|\boundC(S)|/(\vol_C(X)\alpha_0)<
\delta/(\vol_C(X)\alpha_0).\nonumber 
\end{align}
The second inequality above uses that $\alpha_0=o(1)<1/8$. 
We
will now argue that 
\begin{equation}\label{eq:delta-bound}
\delta/(\vol_C(X)\alpha_0)\leq s/(16m(C)).
\end{equation}
This is trivially true if $\vol_C(X)\geq 16 \delta m(C)/(\alpha_0s)$.
It is also true if $\vol_C(X)\geq m(C)/4$ because $s\geq
s_0=64\delta/\alpha_0$.  The lemma assumes that $\vol_C(X)\geq
\min\{16 \delta m(C)/(\alpha_0s),m(C)/4\}$, so \req{eq:delta-bound} follows.
We conclude that
\[p^*_{X,C}(S)\leq s/(16m(C)).\]
Since $\volC(S)>s/2$, this means that vertices $u\in S$
with limit density $p^*_{X,C}(u)/d(u)\leq 1/(4m(C))$ represent more than
half the volume of $S$.

We now apply Theorem \ref{thm:endgame} with $\Cdiff=1/2$. We
get a set $A=T$ with $\volC(A)\leq m(C)$ and
$\PhiC(A)=O(\sqrt{(\alpha_0\log m)/\Cdiff)})=O(1/(\log m)^2)= o(\Phi_0)$. If we
end in case (i) of Theorem \ref{thm:endgame}, the set $A$ is found quickly
in time $\tO(\volC(A)/(\Cdiff\alpha_0))=\tO(\volC(A))$ and then $\excessC {p^*_{X,C}} A\geq \Cdiff/(64\lg (8m(C)))=1/(128\lg (8m(C)))$ as claimed in case (i) of the lemma.

If we end in case (ii) of Theorem \ref{thm:endgame}, the set $A$ is
found in time $\tO(m(C)/(\Cdiff\alpha_0))=\tO(m(C))$ with the guarantee that $A$
contains all vertices with $p^*_{X,C}(u)/d(u)\leq 1/(4m(C))$, which implies
that $\volC(A\cap S)>\volC(S)/2$. With $B=V(C)\setminus A$, this gives
$\volC(B\cap S)\leq\volC(S)/2$, and this holds for any set $S$ with
$|\boundC(S)|\leq \delta$ and $s/2<\volC(S)\leq s$, as required for case (ii) of the
lemma.

If we end in case (iii) of Theorem \ref{thm:endgame}, we know that
there is no vertex $u$ with $p^*_{X,C}(u)/d(u)\leq 1/(4m(C))$, but
then we conclude, by contradiction, that there was no set $S$ with $|\boundC(S)|\leq
\delta$ and $s/2<\volC(S)\leq s$, as required for case (iii) of the
lemma.

To finish the proof of case (ii) and (iii), suppose the component $C$ is
certified $s$-splittable and consider any cut of $\bbar G$ with at most $\delta$ cut edges. Let
$T$ be the side minimizing $\volC(C\cap T)$, and set $S=C\cap T$.  Since
$C$ is $s$-splittable, we know that $\volC(S)\leq s$. Moreover,
$|\boundC(S)|\leq|\boundary_{\bbar G}(T)|\leq \delta$.

In case (ii), the algorithm certifies that
$\vol_C(S\cap B)\leq s/2$. We also have 
$\vol_B(T\cap B)=\vol_B(S\cap B)\leq \vol_C(S\cap B)$. Therefore
$\vol_B(T\cap B)\leq s/2$, so we conclude that $B$ is $s/2$-splittable.

In case (iii), the algorithm certifies that we
cannot have $s/2<\volC(S)\leq s$, but $\volC(S)\leq s$,
so we conclude that 
$\volC(C\cap T)=\volC(S)\leq s/2$, implying that $C$ is $s/2$-splittable.

Above we assumed that no vertex from $X$ was $s$-captured, but
even if this is not the case, we only return case (i) if
we get it from case (i) of Theorem \ref{thm:endgame} whose properties
do not depend on any assumptions.
\end{proof}

\subsection{Proof of Theorem \ref{thm:recurse} for highly splittable
components}\label{sec:large} 

We will now prove Theorem \ref{thm:recurse}
for an $s$-splittable component $C$ where $s=\tOmega(m(C))$, that is, 
\begin{quote}{\em Let
$s\in [s_0,m(C)]$ where $s=\tOmega(m(C))$, and let $C$ be an
$s$-splittable trimmed component of $H$.  We have an algorithm that,
depending on the input, will do one of the following:
\begin{itemize}
\item[(i)] Find a set $A\subseteq V(C)$ with $\volC(A)\leq m(C)$ and $\PhiC(A)
=o(\Phi_0)$
 in time $\tO(\volC(A))$.
\item[(ii)] Find a set $A\subseteq V(C)$ with $\volC(A)\leq m(C)$ and
$\PhiC(A)=o(\Phi_0)$ in time $\tO(m(C))$ certifying that the large side
$B=V(C)\setminus A$ is $s/2$-splittable.
\item[(iii)] Certifying in $\tO(m(C))$ time that $C$ is $s/2$-splittable.
\end{itemize}}
\end{quote}
\begin{proof}[ of Theorem \ref{thm:recurse}
when \boldmath{$s=\tOmega(m(C))$}] Our algorithm
picks an arbitrary set $X\subseteq V(C)$
with $\ceil{40 m(C)/(s\alpha_0)}=\tO(1)$ vertices. Since $C$ is
trimmed, the minimum degree in $C$ is at least $2\delta/5$, so
$\vol_C(X)\geq 16 \delta m(C)/(s\alpha_0)$.
Now, in parallel alternation,
we run Lemma \ref{lem:s-captured} on every vertex $v\in X$, and we run
Lemma \ref{lem:not-s-captured} on the set  $X$. We terminate
if someone finds a set $A$ with $\PhiC(A)=o(\Phi_0)$ corresponding to
case (i) in Lemma \ref{lem:s-captured} or
(i) in Lemma \ref{lem:not-s-captured}, calling this early termination; otherwise
we continue until all processes have terminated.

In the early termination case, since we run only $\tO(1)$ processes
in parallel, the total running time is $\tO(\volC(A))$, so we match
Theorem \ref{thm:recurse} (i).

If no process reaches case (i), the total running time is
$\tO(m(C))$. We get from Lemma~\ref{lem:s-captured}
(ii) that no vertex $v\in X$ is $s$-captured, which means that we can
trust the certifications in case (ii) and (iii) of Lemma
\ref{lem:not-s-captured}, and they match case (ii) and (iii) of
Theorem \ref{thm:recurse}.
\end{proof}

\subsection{Proof of Theorem \ref{thm:recurse} for less splittable
components}\label{sec:small} 

We will now prove Theorem \ref{thm:recurse} for less splittable
components than those handled above in Section \ref{sec:large}. Our
new proof will work for a trimmed component $C$
that is $s$-splittable for any $s\in[s_0,m(C)/32]$. This
case is far more complicated, and requires several new lemmas.

First let us see what goes right and wrong if we try to do the same
as we did with the highly spilttable components.
The algorithm  would still be correct,
but now we have no good bound on the size of the set $X$. This means
that the multiplicative slowdown from running $|X|$ process is not
bounded.

It is instructive to note that if we apply Lemma \ref{lem:s-captured}
to all $v\in X$, then the total running time is $\tO(m(C))$, for the
lemma spends $\tO(s)$ time on each of the $|X|=\ceil{40 m(C)/(s\alpha_0)}$ vertices.
However, we cannot afford to spend this much time if for some $v\in
X$, we end in case (i) with no certification but a low-conductance cut
around a very small side. Our idea to circumvent the problem is to
exploit that case (i) implies a lower bound on the excess, both in
Lemma \ref{lem:s-captured} and in Lemma \ref{lem:not-s-captured}, and
we want to detect this efficiently in advance. This is
the most tricky part of our algorithm, and the motivation for
including excess guarantees on the low-conductance
set $A$ found in Theorem \ref{thm:ACL} and Theorem
\ref{thm:endgame}. The following two lemmas address the issue. The
reader who wants to fully understand the motivation for these lemmas
may want to skip to Theorem \ref{thm:recurse} and see how they are
used in its proof.
The first lemma is about identifying a large set of non-$s$-captured
vertices in a trimmed component $C$, but not blindly running Lemma
\ref{lem:s-captured} from each vertex. If a low conductance cut is found,
we will only use time near-linear in the volume of the smaller side.
\begin{lemma}\label{lem:many-non-captured} 
For $s\in [s_0,m(C)/32]$, let $Y$ be a set
of at most $m(C)/(1024 s\lg(4s))$ vertices from a trimmed component $C$. We have an algorithm that, depending
on the input,  will do
one of the following:
\begin{itemize}
\item[(i)] Find a set $A\subseteq V(C)$ with $\volC(A)\leq m(C)$ and $\PhiC(A)=o(\Phi_0)$ in time $\tO(\volC(A))$.
\item[(ii)] Identify a subset $X\subseteq Y$, $\vol_C(X)\geq\vol_C(Y)/2$ in
  $\tO(m(C))$ time, certifying that no vertex in $X$ is $s$-captured
  in $C$.
\end{itemize}
\end{lemma}
\begin{proof}
First we consider an optimimistic algorithm that in $\tO(m(C))$ time 
identifies a set $X\subseteq Y$ with no $s$-captured vertices. This
is, in itself trivial, since $X=\emptyset$ would do. However, the
optimistic algorithm applies Lemma \ref{lem:s-captured} to each $v\in Y$ in
$\tO(s)$ time. Some vertices will be certified as not $s$-captured in
Lemma \ref{lem:s-captured} (ii),
and they are the ones we place in $X$. The total time we
spend is $\tO(|Y|s)=\tO(m(C))$, so if $X$ ends up
with at least half the vertices from $Y$, then we are done as in case (ii).

The bad case for the optimistic algorithm is if we end up with
$\vol_C(X)<\vol_C(Y)/2$. We will now study the bad case, finding a way to
detect it without having to run the optimistic algorithm. Thus, below we
pretend we have run the optimistic algorihm ending in the bad case
with $\vol_C(X)<\vol_C(Y)/2$. For every
$v\in Y\setminus X$, when running PageRank from $v$ with Lemma
\ref{lem:s-captured}, we get Lemma \ref{lem:s-captured} (i) with a low conductance cut where the small
side $T_v$ has $\volC(T_v)\leq 16s$ and a limit excess above $1/(32\lg
(4s))$. We get the excess guarantee from Lemma \ref{lem:s-captured} (i) 
because $s\leq m(C)/32$ and it implies that $p^*_{v,C}(T_v)> 1/(32\lg
(4s))$. Let 
$S'=\bigcup_{v\in Y\setminus X} T_v$. Recall that pushing to the limit
is a non-negative linear transformation. This means if
we run PageRank from {\em any\/} distribution on $Y\setminus X$, 
then the limit mass on $S'$ is above $1/(32\lg
(4s))$. 

What happens in the real algorithm behind Lemma \ref{lem:many-non-captured} is that we first
run PageRank starting from the 
distribution $p^\circ_{Y,C}$ on $Y$ with uniform density $1/\vol_C(Y)$.  Assuming we are in the bad case
with $\vol_C(X)<\vol_C(Y)/2$, we have $p^\circ_{Y,C}(Y\setminus X)\geq 1/2$. It
follows that when we push to the limit, we
end up with mass $p^*_{Y,C}(S')> 1/(64\lg (4s))$.  We also have
$\volC(S')\leq 16s|Y|$. This means that $S'$ gets excess at least
$\Cdiff'=1/(64\lg (4s))-16s|Y|/(2 m(C))$.  However, we have $|Y|\leq
m(C)/(1024 s\lg(4s))$, and hence $\Cdiff'\geq 1/(128\lg (4s))$.  Thus,
if we are in the bad case for the optimistic algorithm, then we know
that $\excessC {p^*_{Y,C}} {S'}\geq \Cdiff'$. Applying 
Theorem \ref{thm:ACL} to $p^\circ_{Y,C}$ with excess parameter
$\Cdiff'$, we get one of two outcomes:
\begin{itemize}
\item We find a set $A$ with $\vol(A)\leq m(C)$ and $\Phi_C(A)\leq
  O(\sqrt{(\alpha_0\log m)/\Cdiff'})=O(1/(\log m))^{3/2})=o(\Phi_0)$ in
  time $\tO(\volC(A))$, satisfying Lemma \ref{lem:many-non-captured}
  (i).
\item In $\tO(m(C))$ time, we certify that there is no set $S'$ with
  excess $\Cdiff'$.  We now run the above optimistic
  algorithm with no risk of ending in the bad case, and then we
  satisfy Lemma \ref{lem:many-non-captured} (ii).
\end{itemize}
\end{proof}
The next lemma will be used to certify that if we try to
run Lemma \ref{lem:not-s-captured} on any set $X$ with at least half the
vertices from a given set $Y$ and we end with the low-conductance
cut of Lemma \ref{lem:not-s-captured} (i), then the smaller side
has volume  $\tOmega(m(C))$.  The new lemma may itself
yeild a low-conductance cut, but the difference is that the new lemma
does not make any assumptionas about vertices being $s$-captured as
required for Lemma \ref{lem:not-s-captured}.
\begin{lemma}\label{lem:half-concentration} Let $Y$ be any set of vertices from a trimmed component
$C$. We have an algorithm that, depending
on the input, will do
one of the following:
\begin{itemize}
\item[(i)] Find a set $A\subseteq V(C)$ with $\volC(A)\leq m(C)$ and
  $\PhiC(A)=o(\Phi_0)$ in time $\tO(\volC(A))$.
\item[(ii)] Certify in $\tO(m(C))$ time that
there is not a subset $X\subseteq Y$ with $\vol_C(X)\geq \vol_C(Y)/2$ and a set $A\subseteq V(C)$ with $\volC(A)\leq m(C)/(256\lg (8m))$ such that
  $\excessC{p^*_{X,C}} A \geq 1/(128\lg (8m))$. Recall 
that 
$p^*_{X,C}$ is the limit distribution when we run PageRank on $C$
starting with all mass uniformly spread on $X$.
\end{itemize}
\end{lemma}
\begin{proof}
Let us assume that there is a subset
$X\subseteq Y$ with $\vol_C(X)\geq \vol_C(Y)/2$ and set $A\subseteq V(C)$ with
$\volC(A)\leq m(C)/(256\lg (8m))$ such that 
  $\excessC{p^*_{X,C}} A \geq 1/(128\lg (8m))$.  Then we
have to end in case (i).

We are going to start PageRank with the distribution $p^\circ_{Y,C}$
on $Y$ with uniform density $1/\vol_C(Y)$. The limit distribution is
$p^*_{Y,C}$. We have $p^\circ_{X,C}$ and $p^*_{X,C}$ denoting the
corresponding distributions if we instead started with uniform
density $1/\vol_C(X)$ on $X$.

Since $X\subseteq Y$ and $\vol_C(X)\geq \vol_C(Y)/2$, we
get $p^\circ_{Y,C}(v)\geq p^\circ_{X,C}(v)/2$ for every vertex $v\in C$.
Since pushing to the limit is a non-negative linear transformation, we
conclude that that we also in the limit get $p_{Y,C}^*(v)\geq
p_{X,C}^*(v)/2$ for every vertex $v\in C$. In particular, we get that
\[p_{Y,C}^*(A)\geq p_{X,C}^*(A)/2>\excessC{p^*_{X,C}} A/2 \geq 1/(256\lg (8m)).\]
Therefore
\[\excessC{p^*_{Y,C}} A=p_{Y,C}^*(A)-\volC(A)/(2m(C))>1/(512\lg (8m)).\]
Thus, starting PageRank from $p^\circ_{Y,C}$ using Theorem \ref{thm:ACL} 
with $\Cdiff=1/(512\lg (8m))$, we will get a set
$A'$ for case (i) with $\volC(A')\leq m(C)$ and $\Phi_C(A')\leq O(\sqrt{(\alpha_0\log m)/\Cdiff})=O(1/(\log m)^{3/2})=o(\Phi_0)$ in time $\tO(\volC(A'))$. We then return $A'$ as in case (i).

If no such $A'$ is found, we terminate in
$\tO(m(C)/(\alpha_0\Cdiff))=\tO(m(C))$, certifying that our assumptions
were false as in case (ii).
\end{proof}

We are now ready to complete the proof of Theorem \ref{thm:recurse}:
\begin{quote}\em

 Let $s\in [s_0,m(C)]$ and $C$ be an $s$-splittable trimmed component of $H$.
We have an algorithm that, depending
on the input, will do one of the following:
\begin{itemize}
\item[(i)] Find a set $A\subseteq V(C)$ with $\volC(A)\leq m(C)$ and $\PhiC(A)
=o(\Phi_0)$
 in time $\tO(\volC(A))$.
\item[(ii)] Find a set $A\subseteq V(C)$ with $\volC(A)\leq m(C)$ and
$\PhiC(A)=o(\Phi_0)$ in time $\tO(m(C))$ certifying that the large side
$B=V(C)\setminus A$ is $s/2$-splittable.
\item[(iii)] Certifying in $\tO(m(C))$ time that $C$ is $s/2$-splittable.
\end{itemize}
\end{quote}
\begin{proof}[ of Theorem \ref{thm:recurse}]
We have already handled the case where $s=\tOmega(m(C))$ in Section
\ref{sec:large}, so we may assume that $s\leq m(C)/32$, as required
for Lemma \ref{lem:many-non-captured}.

First we assume that $C$ has at least $\ceil{80 m(C)/(s\alpha_0)}$
vertices, and then we let $Y$ be any set of $\ceil{80
  m(C)/(s\alpha_0)}$ from $C$.  Since $C$ is trimmed, the minimum
degree in $C$ is at least $2\delta/5$, so $\vol_C(Y)\geq 32 \delta
m(C)/(s\alpha_0)$. Next we cut $Y$ into $\ceil{80\cdot 1024 \lg
  (4s)/\alpha_0}=\tO(1)$ segments $Y_i$, each with at most $m(C)/(1024
s\lg(4s))$ vertices.

Recall from Section \ref{sec:representation} that we have
a balanced binary search tree over the vertex list from $C$, where each tree node
knows the number of descendants. Using this tree, we can cut into
segments of any desired size in $O(\log n)$ time per segment, so the the
total time spent so far is $\tO(1)$.

We will then, alternating in parallel,
apply Lemma~\ref{lem:many-non-captured} to every $Y_i$
while, also in parallel, applying Lemma \ref{lem:half-concentration} to $Y$.
If any one of these ends in case (i), then this corresponds to case (i)
of the theorem. The multiplicative $\tO(1)$ slowdown does not affect the
time bound. Thus we are done if we get Lemma~\ref{lem:many-non-captured} (i) for some $Y_i$ or Lemma \ref{lem:half-concentration} (i) for $Y=\bigcup_i Y_i$.

Assume instead that we get no case (i). Then for each $Y_i$, by
Lemma~\ref{lem:many-non-captured} (ii), we find a subset $X_i\subseteq
Y_i$ with at least half the volume, that is, $\vol_C(X_i)\geq \vol_C(Y_i)/2$,
and such that no vertex in $X_i$
is $s$-captured.  Then no vertex in $X=\bigcup_i X_i$ is $s$-captured,
and $\vol_C(X)\geq \vol_C(Y)/2$. Now by Lemma
\ref{lem:half-concentration} (ii), we know that there is no set
$A\subseteq V(C)$ with $\volC(A)\leq m(C)/(256\lg (8m))$ such that 
$\excessC{p_{X,C}^*} A \geq 1/(128\lg (8m))$.
We have spent $\tO(m(C))$ time so far. We know that $\vol_C(X)\geq \vol_C(Y)/2\geq
16 m(C)/(s\alpha_0)$ and that no vertex from $X$ is $s$-captured.

We now apply Lemma \ref{lem:not-s-captured} to $X$. 
If we get case (ii) or (iii) of Lemma~\ref{lem:not-s-captured}, then
they directly gives us case (ii) or (iii) of the theorem, so we are done.

Suppose instead  we get
Lemma \ref{lem:not-s-captured} (i). Then we find a set $A\subseteq
V(C)$ with $\volC(A)\leq m(C)$, $\excessC{p^*_{X,C}} A \geq 1/(128\lg
(8m))$, and $\PhiC(A)=o(\Phi_0)$ in time $\tO(\volC(A))$.  This would
be good enough for case (i) of the theorem, except that the total time
we have spent finding $A$ is $\tO(m(C))$, and we are only allowed
time $\tO(\volC(A))$.

We now combine with the above conclusion from Lemma
\ref{lem:half-concentration} (ii) which says that there is no set
$A\subseteq V(C)$ with $\volC(A)\leq m(C)/(256\lg (8m))$ and
$\excessC{p^*_{X,C}} A \geq 1/(128\lg (8m))$. We conclude that the set $A$
found by Lemma \ref{lem:not-s-captured} (i) has $\volC(A)>
m(C)/(256\lg (8m))=\tOmega(m(C))$, which means that the
total time we have spent is $\tO(m(C))=\tO(\volC(A))$, as required for
case (i) of the theorem. This completes the proof of 
Theorem \ref{thm:recurse} assuming that $C$ has at
least $\ceil{80 m(C)/(s\alpha_0)}$ vertices.

We will now complete the proof of Theorem \ref{thm:recurse} by handling
the case where $C$ has less than $80 m(C)/(s\alpha_0)$ vertices. The
proof is very similar to the case where $C$ had more
vertices, but this time, we let $Y$ consist of all the vertices from
$C$, that is, $Y=V(C)$.

Next, as above, we partition $Y$ into $\tO(1)$ sets $Y_i$, each with at most
$m(C)/(256 s\lg(4s))$ vertices. If $Y$ is small, we get fewer sets
$Y_i$, which is only an advantage.

We apply Lemma~\ref{lem:many-non-captured} in parallel to every
$Y_i$. If we get Lemma~\ref{lem:many-non-captured} (i) for any of $Y_i$, then this
corresponds to Theorem \ref{thm:recurse} (i), and we are done.

Thus, suppose we for every $Y_i$ get 
Lemma~\ref{lem:many-non-captured} (ii) with a subset $X_i$ with no
$s$-captured vertices and with 
$\volC(X_i)\geq \volC(Y_i)/2$. We consider now the set $X=\bigcup_i X_i$.
It has $\volC(X)\geq\volC(Y)/2=m(C)$, and there is no $s$-captured vertex
in $X$. 

Since $\volC(X)>m(C)/4$, we can now apply Lemma \ref{lem:not-s-captured} to $X$.
If we get case (ii) or (iii) of Lemma~\ref{lem:not-s-captured}, then
they directly gives us case (ii) or (iii) of the theorem, so we are done.

Suppose instead  we get
Lemma \ref{lem:not-s-captured} (i). Then we find a set $A\subseteq
V(C)$ with $\volC(A)\leq m(C)$, $\excessC{p^*_{X,C}} A \geq 1/(128\lg
(8m))$, and $\PhiC(A)=o(\Phi_0)$ in time $\tO(\volC(A))$.  This would
be good enough for Theorem \ref{thm:recurse} (i), except that the total time
we have spent finding $A$ is $\tO(m(C))$, and we are only allowed
time $\tO(\volC(A))$.

Like in the case when $C$ had more vertices, we want to argue that
$\volC(A)=\tOmega(m(C))$. This time we will not use Lemma
\ref{lem:half-concentration} (ii), but instead apply a direct
argument. Let $p^\circ_{X,C}$ be the uniform density distribution
on $X$. We have $\volC(X)\geq m(C)$, so $p^\circ_X$ is dominated by
the stationary distribution $\overline{1/m(C)}$, and hence so is
the limit distribution $p^*_{X,C}=\PRC(\alpha_0,p^\circ_{X,C})$.  This
means that any set $A'$ has $\excessC {p^*_{X,C}} A'\leq
\volC(A')(1/m(C)-1/(2m(C)))\leq \volC(A')/(2m(C))$. But the set $A$
from Lemma \ref{lem:not-s-captured} (i) had $\excessC{p^*_{X,C}}
A \geq 1/(128\lg(8m))$, so we conclude that $\volC(A)\geq
m(C)/(64\lg(8m)) =\tOmega(m(C))$. The set $A$ from Lemma
\ref{lem:not-s-captured} (i) is thus found within the $\tO(\volC(A))$
time required for Theorem \ref{thm:recurse} (i). This completes the proof of Theorem \ref{thm:recurse}.
\end{proof}

\begin{proof}[ of  Theorem \ref{thm:main-tech}]
We have now completed the proof of Theorem \ref{thm:recurse}, and by 
Lemma \ref{lem:recurse-main-tech}, this means that we are also done
with the proof of Theorem \ref{thm:main-tech}.
\end{proof}

\subsection{Log-factors} In this paper, we have not worried about
the number of log-factors in our near-linear time bound for solving
the min-cut problem, nor have we accounted for them. More precisely,
as described in Section \ref{sec:notation}, we have freely used
simplifications like $\tO(\tO(f(n)))=\tO(f(n))$ and
$\tO(f(n))\tO(g(n)))=\tO(f(n)g(n))$. For a proper accounting, we 
would have to undo these simplifications. 
Below we will sketch that 12 log-factors suffices. The purpose is
not to prove this, but rather to set up an estimated benchmark that other 
other researchers can improve on.

Currently, we have $\alpha_0=1/(\log m)^5$, but in fact it
suffices with $\alpha_0=1/(c_0(\log m)^4)$ for some sufficiently large constant
$c_0$. The place that puts the
biggest demand on $\alpha_0$ is in the end of the proof of
Lemmas \ref{lem:many-non-captured} and \ref{lem:half-concentration} where we
need that $\Phi_C(A)=O(\sqrt{(\alpha_0\log m)/\Cdiff})=
O\left(\sqrt{\alpha_0\log^2 m}\right)\leq\Phi_0=1/(20\lg m)$.
By definition of the $O$-notation, there exists a large enough
constant $c_0$ such that $\alpha_0=1/(c_0(\log m)^4)$
yields $\Phi_C(A)\leq 1/(20\lg m)$.

We can also reduce the requirement on $\delta$ to $\delta\geq c_1/\alpha_0$ and
set $\delta^*=c_1\delta/\alpha_0$ for some sufficiently large constant
$c_1$. The critical place is Lemma \ref{lem:cross-cut} which currently says that if we start the PageRank algorithm
from a vertex with a fraction $\eps$ of its edges leaving a certain set
$S$, then in the limit, the mass leaving $S$ is only $\eps+o(1)$.
Lemma  \ref{lem:cross-cut} was proved for $O(\delta)$ cut edges, but we
never need it for more than $2\delta$ cut edges. With this
concrete bound on the cut size, if we  parameterize by $c_1$ and change the
proof of Lemma \ref{lem:cross-cut} accordingly,
the mass leaving $S$ is at most $\eps+O(1/c_1)$.
When we later apply Lemma \ref{lem:cross-cut} to
the proof of Lemma \ref{lem:s-captured}, what we need is
that $3/4-O(1/c_1)-1/2> 1/5$, which is true for some sufficiently large
constant $c_1$.

The conclusion is that we can run our algorithm with parameters
$\alpha_0=O(\log^4 m)$ and $\delta^*=O(\delta\log^4 m)$. For Lemma \ref{lem:passive}, this implies that the number of
edges leaving passive super vertices is $O(m\delta^*/\delta^2)=O(m(\log^4
m)/\delta)$, which then also bounds the number of edges in $\bbar G$.

The bottleneck in time originates from Lemma
\ref{lem:many-non-captured} (i), where the set $A$ is really found in
time $O(\volC(A)(\log m)/(\Cdiff\alpha_0))=O(\volC(A)\log^6 m)$
time. In the proof of Theorem \ref{thm:recurse}, we run $O((\log
s)/\alpha_0)=O(\log^5 m)$ such experiments from Lemma
\ref{lem:many-non-captured} in parallel, so the cost is $O(\log^{11}
m)$ per edge, and the same edge may get charged $2\lg m$ times as it
ends up in components of half the size or half as splittable. Thus a total cost of $O(\log^{12} m)$ per edge is needed
in order to find the clusters. When we afterwards contract the cores,
we halve the number of edges, so it is the cost of the first cluster
finding round that dominates. Our total cost for finding $\bbar G$ is
$O(m \log^{12} m)$.  Since $\bbar G$ has only $O((m \log^4 m)/\delta)$
edges, using Gabow's algorithm, we can now find a minimum cut in
$O(m \log^5 m)$ time. Our overall time bound for finding the
minimum cut is thus $O(m \log^{12} m)=O(m \log^{12} n)$. This
was recently improved by Henzinger et al.~\cite{HRW17} to $O(m(\log
n)^2(\log\log n)^2)$.

\section{Approximate cuts with fewer vertices}\label{sec:approx}
In this section, we prove Theorem  \ref{thm:main-tech-strong}. To
do this, we 
will first generalize our min-cut contraction
algorithm to preserve approximate min-cuts. Afterwards we will
show how to reduce the number of vertices in the contracted
graph by a factor $\delta$.

\subsection{Preserving approximate cuts}\label{sec:approx-cuts}
We are now going to modify our min-cut contraction algorithm
to preserve, not only min-cuts, but also
approximate min-cuts:
\begin{theorem}\label{thm:main-tech-approx} Let $\eps\in(0,1]$ be a constant.
Given a simple input graph $G$ with $n$ vertices, $m$ edges, and
minimum degree $\delta$, and (unknown) edge
connectivity $\lambda$, in near-linear time, we can contract vertex
sets producing a multigraph $\bbar G$ which has only $\bbar
m=\tO(m/\delta)$ edges, yet which preserves all non-trivial cuts of
$G$ of size at most $\lambda+(1-\eps)\delta$.
\end{theorem}
It turns out that we only need minor modifications of our contraction
algorithm for min-cuts to preserve approximate min-cuts and
prove Theorem \ref{thm:main-tech-approx}.
We will first describe the algorithmic changes, later 
the changes to the analysis.

\paragraph{Algorithmic changes}
On the high level, we are going to reuse Algorithm \ref{alg:min-cut} 
for the approximate min-cuts of Theorem \ref{thm:main-tech-approx}. However,
we need to modify some definitions by changing
four thresholds.

The first obvious change is that before we said we could contract two
vertices if there were more than $\delta$ parallel edges between
them. Now there should be more than $2\delta$ parallel edges between
them.

More interestingly, we have to change the definition of loose vertices from Section
\ref{sec:semi-contract}. The original defintion was that a vertex $v$
in a cluster $C$ is loose if it is regular and at least $d(v)/2-1$ of
its edges in $\bbar G$ leave $C$. Now we say {\em $v$ is loose if at
  least $\eps d(v)/2-4$ of its edges in $\bbar G$ leave $C$.}  Recall
that loose vertices have to be shaved from $C$ to get down to
the core, which is scrapped if less than $1/4$ of the edges incident
to $C$ are internal to the core. With the new definition, we get more
loose vertices, hence smaller cores.  Essentially this will increase
the number of edges removed from cores by a factor $1/\eps$.

The third change is in Section  \ref{sec:super-degree} where
we said that we terminate the algorithm if
more than a fraction $1/20$ of the edges in $\bbar G$ are
incident to passive super vertices. Now will we terminate if more than a
fraction $\eps/20$ of the edges are incident to passive super
vertices. 

The fourth change is in Section \ref{sec:cut-trim} in the definition of
the threshold $\Phi_0=1/(20\lg m)$ for a low conductance cut.  We
change it to $\Phi_0'=\eps/(20\lg m)$. The last two changes will imply
that the fraction of edges cut from passive vertices and from low
conductance cuts is reduced by $\eps$ from $1/10$ to $\eps/10$.

Finally, we note an implicit change; namely that new asymptotic
calculations may affect the smallest number $n_0$ of vertices that our
algorithm can handle (c.f. Section \ref{sec:notation}). However, $n_0$ 
remains a constant and Theorem \ref{thm:main-tech-approx} is
trivially true for a constant sized graph.

This completes the changes to our contraction algorithm. We will
now analyze how the changes help us preserve approximate cuts below.

\paragraph{Modified analysis}
We want to prove that our modified contraction algorithm preserves cuts of size at most
$\lambda+(1-\eps)\delta<2\delta$. To do that, we will in many cases
focus on cuts of size at most $2\delta$ instead of $\delta$. The
first example is in the definition of a cluster
from Section \ref{sec:clusters}. Originally we said that a trimmed
vertex set $C$ was a cluster if for every cut of size at most
$\delta$ in $\bbar G$, one side contains at most two regular vertices
and no super vertices from $C$. Now we say {\em a trimmed
vertex set $C$ is cluster if 
for every cut of size at most
$2\delta$ in $\bbar G$, one side contains at most five 
regular vertices and no super vertices from $C$. }
Corrsponding to Lemma \ref{lem:big-sides}, the following
lemma says that more than five regular vertices implies many regular
vertices:
\begin{lemmanew}{\ref{lem:big-sides}}
Consider a trimmed vertex set $C$ and
a cut $(T,U)$ of $\bbar G$ of size at most $2\delta$ (before it was $\delta$). If $T\cap C$ has no super vertices
and at least 6 (before it was 3) regular
vertices, then $T\cap C$ has at least $\delta/3$ regular
vertices.
\end{lemmanew}
\begin{proof}
Consider $T\cap C$ which has no super vertices.  Since $C$ is trimmed,
the internal degree of regular vertices in $C$ is at least
$2\delta/5$, so the number of edges crossing from $T\cap C$ to $U\cap
C$ is at least $|C\cap T|(2\delta/5+1-|C\cap T|)$, but we have at most
$2\delta$ cut edges. Since $\delta=\omega(1)$, we conclude that
$|C\cap T|\leq 5$ or $|C\cap T|\geq 2\delta/5-4>\delta/3$.
\end{proof}
The most important change to the analysis, however, is to 
show an analogue to Lemma~\ref{lem:core-contract} which said that 
contracting cores preserved non-trivial min-cuts. With our new definition
of clusters and of the loose vertices not in the core, we want
to show that contracting cores preserves
non-trivial cuts of size up to $\lambda+(1-\eps)\delta$. 
\begin{lemmanew}{\ref{lem:core-contract}} If a non-trivial cut 
of $G$ of size at most $\lambda+(1-\eps)\delta$ (before it was
$\lambda$)  has
survived in $\bbar G$, then it will also survive  when
we contract the core of any cluster in $\bbar G$.
\end{lemmanew}
\begin{proof}
Consider a cut $(T,U)$ of $\bbar G$ of size at most
$\lambda+(1-\eps)\delta$ that was non-trivial in $G$. 
We must have at least two regular vertices
or one super vertex both in $T$ and in $U$.

We now consider a cluster $C$ in $\bbar G$ with a non-empty core $A$.
Since $(T,U)$ has size at most $\lambda+(1-\eps)\delta<
2\delta$, by the new
definition of a cluster, one side, say $T$, has at most five regular
vertices and no super vertices from $C$. We will argue that these
vertices in $C\cap T$ must be loose, hence that the vertices
identified by the contraction of $A$ are all in $U$; for then
the contraction preserves $(T,U)$.

Let $v$ be one of the regular vertices from $C\cap T$, and assume for
a contradiction that $v$ is not loose. By our new definition of loose,
this means that $v$ has at most $\eps d(v)/2-5$ edges leaving $C$ in $\bbar G$.

We will prove that we get a cut that is more than $(1-\eps)\delta$ edges
smaller by moving $v$ to $U$, 
contradicting that $(T,U)$ had size at most 
$\lambda+(1-\eps)\delta$.

Moving $v$ only affects the cutting of
edges incident to $v$. Recall that $T\cap C$ has at most five regular vertices,
including $v$, and no super vertices, so $v$ has at most
four edges to other vertices in $T\cap C$. When $v$ is in $T$, we cut 
all other edges
from $v$ to $C$. However,
since $v$ is not loose, it has least $(1-\eps/2)d(v)+5$ edges from $v$ into $C$, so
with $v$ in $T$, we cut at least $(1-\eps/2)d(v)+1$ edges incident to $v$.
Moving $v$ to $U$, we cut the at most $\eps d(v)/2-1$ other edges incident 
to $v$. Moving $v$ from $T$ to $U$ thus reduces the cut size by
at least $(1-\eps)d(v)+2$, yielding the desired contradiction.
\end{proof}

We also need an analogue of Lemma~\ref{lem:shaving} bounding the
number of edges incident to a cluster that are not internal to the core.
\begin{lemmanew}{\ref{lem:shaving}} If a cluster $C$ has $k$ edges leaving it,
then there are less than $3k/\eps$ (before it was $3k$) 
edges incident to $C$ that are not internal to the
core. In particular, if the core is empty, we have $\vol(C)<3k/\eps$.
\end{lemmanew}
\begin{proof}
Let $A$ be $C$ without the loose vertices, i.e., $A$ is the core unless
the core becomes empty. Let $\ell$ be the number of
edges leaving $C$ from loose vertices.  Then we have $k-\ell$ edges
leaving $C$ from vertices in $A$.  Other edges incident to $C$ but not
internal to $A$ are all incident to loose vertices.

Consider any loose vertex $v$ in $C$. It has at least $\eps d(v)/2-4=
(\eps/2-o(1))d(v)$ edges leaving $C$. Here we used that loose vertices are regular, so
$d(v)\geq \delta=\omega(1)$. It follows that the total number of edges
incident to loose vertices is at most
$\ell/(\eps/2-o(1))=(2/\eps+o(1))\ell$. This proves the lemma unless
the core becomes empty.

The core becomes empty if and only if at most $1/4$ of the edges
incident to $C$ are internal to $A$, but this implies that the number
of edges internal to $A$ is at most $1/3$ of the number of edges not
internal to $A$. Thus, if $A$ is not the core, there are at most
$(2/\eps+o(1))k/3$ edges internal to $A$, and then we have at most
$(2/\eps+o(1))k+(2/\eps+o(1))k/3<3k/\eps$ edges incident to $C$.
\end{proof}
There are no changes to Lemma \ref{lem:big-super} and \ref{lem:passive}. With our new
Lemma~\ref{lem:shaving}',  we easily get a corresponding change to
Lemma \ref{lem:cut-trim}:
\begin{lemmanew}{\ref{lem:cut-trim}}
 If the total number of edges cut is $c$,
then the total number of edges lost due to trimming, shaving, and scrapping
is at most $4c/\eps$ (before it was $4c$).\bigskip
\end{lemmanew}

Because of the new bounds, we need to reprove: 

\paragraph{Lemma \ref{lem:edges-left}}\hspace{-2ex}  {\em Cutting edges around passive vertices and edges of
low-conductance cuts, trimming, shaving, and scrapping, leaves at least
half the edges of $\bbar G$ in the resulting cluster cores of $H$.}
\smallskip

\begin{proof} Recall that we changed the algorithm to stop when more than a fraction $\eps/20$ of the edges in $\bbar G$ were incident 
to passive vertices, so this limits the fraction that gets
cut when we cut out these passive vertices.

Also, now we have defined low conductance cuts to have conductance
at most $\Phi_0'=\eps/(20\lg m)$,
and the fraction of edges cut by recursive low conductance
cuts is at most $\lg m$ times bigger, so at most $\eps/20$. Thus,
the total fraction of edges cut is at most $\eps/10$. Then, by
Lemma \ref{lem:cut-trim}', the total fraction of
edges lost due to trimming, shaving, and scrapping is 
at most $4(\eps/10)\eps=4/10$. All together, the total
fraction of edges lost is $\eps/10+4/10\leq 1/2$.
\end{proof}
The change of $\Phi_0$, multiplying it by $\eps=\Theta(1)$,
has no other impact. The reason is that whenever
our algorithm produces a low conductance cut, we proved it to be of conductance
$o(\Phi_0)<\Phi_0'$ (c.f.~Theorem \ref{thm:recurse}).

When it comes to the graph representations and modifications in Section
\ref{sec:representation}, there are only two minor changes. The first is
that we now need more than $2\delta$ parallel edges betweeen vertices
before we contract them. The second is the
new definition of loose vertices. Before a vertex $v$ was loose
if it had $d_H(v)\leq d_{\bbar G}(v)/2+1$. Now it
only needs $d_H(v)\leq (1-\eps/2)d_{\bbar G}(v) +4$. With
the above changes, we can immediately strengthen Lemma \ref{lem:basic-recurse} to
work for approximate min-cuts.
\begin{lemmanew}{\ref{lem:basic-recurse}}
Ignoring the cost of finding the low-conductance cuts, we can implement Algorithm
\ref{alg:min-cut} (with our modified definitions) in $O(m\log^2 n)$ time where $n$ and $m$ are the
number of vertices and edges of the simple input graph $G$. The result
is the contracted graph $\bbar G$ described in Theorem
\ref{thm:main-tech-approx}, that is, $\bbar G$ has only $\tO(m/\delta)$
edges, yet it preserves all non-trivial cuts of
$G$ of size at most $\lambda+(1-\eps)\delta$.\medskip
\end{lemmanew}

Getting to Section \ref{sec:cluster-or-cut}, like for clusters,
we now need to consider cuts of size at most $2\delta$ instead of just
$\delta$, so now we define that {\em a component $C$ of $H$ is
  $s$-splittable if every cut $(T,U)$ of $\bbar G$ of size at most
$2\delta$ has $\min\{\vol_{C}(T\cap C),\vol_{C}(U\cap C)\}<s$.} We do
not need to change that $s_0=64\delta/\alpha_0$.

Working with cuts of size up to $2\delta$, we need to reprove Lemma
\ref{lem:super-splittable} which gives us the main condition for arguing
that a trimmed component is a cluster.
\paragraph{Lemma \ref{lem:super-splittable}} {\em 
If a trimmed component $C$ of $H$ is $s_0$-splittable,
then $C$ is a cluster.}
\begin{proof}
Suppose that $C$ is not a cluster. Then there is a cut $(T,U)$ of $\bbar G$
of size at most $2\delta$ (before it was $\delta$) 
such that both $T\cap C$ and $U\cap C$ contain a super vertex or at least 6 (before it was 3) regular vertices. In
the case of at least 6 regular vertices and no super vertices, 
Lemma \ref{lem:big-sides}' tells us that there are at least $\delta/3$ 
regular vertices.
No
other changes are needed to the original proof of Lemma  \ref{lem:super-splittable}.
\end{proof}
The remaining changes are quite small and local, exploiting
that our original analysis has enough slack to accomodate the increase in
cut size from $\delta$ to $2\delta$ as well as the up to $2\delta$ parallel
edges between vertices. 

We will see that Theorem \ref{thm:recurse} holds unchanged like
Lemma \ref{lem:super-splittable}, and this means that we need no
changes to Lemma \ref{lem:oracle} or its proof.
Therefore, as analogue to Lemma \ref{lem:recurse-main-tech}, we
get 
\begin{lemmanew}{\ref{lem:recurse-main-tech}} 
Theorem \ref{thm:main-tech-approx} follows from Theorem \ref{thm:recurse} (with
our modified definitions).\medskip
\end{lemmanew}

Moving to Section \ref{sec:cross-cut}, we need no changes to the statement of Lemma
\ref{lem:cross-cut} as it is already handling $O(\delta)$ cut
edges. However, inside the proof of Lemma \ref{lem:cross-cut}, we 
need a very minor change. We currently say
that the maximal fraction of mass that can be pushed from a super
vertex to a neighbor is $\delta/(2\delta^*/5)=O(\alpha_0/\log n)$. 
With up to $2\delta$ parallel edges, this is fraction is increased to $2\delta/(2\delta^*/5)=O(\alpha_0/\log n)$. No
other changes are needed to the proof.

With the statement of  Lemma \ref{lem:cross-cut} unchanged,
 we do not need to make any changes
to Lemma \ref{lem:s-captured} or its proof. 
The final place were we need a more careful change is 
\begin{lemmanew}{\ref{lem:not-s-captured}} Lemma \ref{lem:not-s-captured}
holds with 
the new definition of $s$-captured for cuts of size up to $2\delta$ and where
we in case (ii) and (iii) double the
the cut size bound to $\delta_C(S)\leq 2\delta$.
\end{lemmanew}
\begin{proof} The proof is essentially the same as
that of Lemma \ref{lem:not-s-captured}, and we will only
describe the changes. We double the bound on $\delta_C(S)$ to
$2\delta$, and as a result, we may get twice as much limit mass in $S$.
More precisely, as in \req{eq:limit-in-S} from the proof of  Lemma \ref{lem:not-s-captured}, we 
get
\[p^*_{X,C}(S)\leq (4+1/(2\alpha_0))|\boundC(S)|/\vol_C(X).\]
In the proof of Lemma \ref{lem:not-s-captured}, we bounded
$(4+1/(2\alpha_0))\leq 1/\alpha_0$, but this time we tighten it
to $(4+1/(2\alpha_0))\leq 2/(3\alpha_0)$, using $\alpha_0=o(1)<1/24$,
and then we get 
\[p^*_{X,C}(S)\leq(2/(3\alpha_0))|\boundC(S)|/\vol_C(X)\leq 4\delta/(3\vol_C(X)\alpha_0).\]
As in \req{eq:delta-bound} from the proof of  Lemma \ref{lem:not-s-captured},
we have 
\[\delta/(\vol_C(X)\alpha_0)\leq s/(16m(C)).\]
so we conclude that 
\[p^*_{X,C}(S)\leq s/(12m(C))\textnormal{\hspace{3ex} (before it was $s/(16m(C))$)}.\]
Since $\volC(S)>s/2$, this means that vertices $u\in S$
with limit density $p^*_{X,C}(u)/d(u)\leq 1/(3m(C))$ (before it was 
$1/(2m(C))$) represent more than
half the volume of $S$. Finally we apply 
Theorem \ref{thm:endgame} with $\Cdiff=1/6$ (before it was $1/4$).
The rest of the proof is exactly as the proof of Lemma \ref{lem:not-s-captured}.
\end{proof}
Since Lemma \ref{lem:s-captured} and Lemma \ref{lem:not-s-captured}
both hold with our new notion of $s$-captured
for cuts of size up to $2\delta$, we do not need to make any changes
to the proofs of Lemma \ref{lem:many-non-captured}, Lemma \ref{lem:half-concentration}, and Theorem \ref{thm:recurse}. Then Theorem~\ref{thm:main-tech-approx} follows by Lemma \ref{lem:recurse-main-tech}'.

This completes the description of how the contraction algorithm and its
analysis for Theorem \ref{thm:main-tech} can be
modified to prove Theorem \ref{thm:main-tech-approx}.
We note that the proof is not harder, and we could have proved
the stronger Theorem \ref{thm:main-tech-approx} directly to start with.
However, the main result of this paper is about min-cuts, and they
are cleaner to work with, e.g., exploiting that every vertex has at least half its neighbors on the same side.

\subsection{Better edge and vertex bounds}
We will now show how we can improve the edge and vertex bounds of Theorem \ref{thm:main-tech-approx} to $\tO(n)$ and $\tO(n/\delta)$, respectively,
as required to complete the proof of Theorem \ref{thm:main-tech-strong}.

Let $k=2\delta$. First we take out $k$ disjoint forests
$F_1,\ldots,F_k$ from $G$ such that $F_k$ is maximal in $G\setminus
\bigcup_{j<i} F_j$. Nagamochi and Ibaraki \cite{NI92} have described
how to do this in linear time.  Then $H=\bigcup_{i} F_i$ has less than
$m_H=nk=O(n\delta)$ edges. Moreover, $H$ preserves all cuts of size
$\leq k$ and larger cuts preserve at least $k$ of their
edges. Finally, if an edge $(u,v)$ is left in $G\setminus H$, then $u$
and $v$ are at least $k+1$-edge connected in $G$. We are later going
to contract the edges from $G\setminus H$, but we cannot destroy simplicity
yet.

We now apply Theorem \ref{thm:main-tech-approx} to $H$, obtaining a
contracted graph $\bbar H$ with $\bbar m_H=\tO(m_H/\delta)=\tO(n)$
edges.  We claim that $\bbar H$ has only $\tO(n/\delta)$
vertices. Trivially $\bbar H$ has at most $2\bbar
m_H/\delta=\tO(n/\delta)$ vertices of degree $\geq \delta$, which was
the minimum degree in $G$. Lower degree vertices must be super
vertices, and by Lemma \ref{lem:big-super}, each super vertex in
$\bbar H$ contracts $\Omega(\delta^2)$ edges from $H$, so we can only
have $O(m_H/\delta^2)=O(n/\delta)$ super vertices in $\bbar H$.  Thus
$\bbar H$ has $\tO(n/\delta)$ vertices and $\tO(n)$ edges, and
it preserves all non-trivial cuts of $H$ of size at most
$\lambda+(1-\eps)\delta$.

Our final step is to contract in $\bbar H$ the end-points of each edge $(u,v)$ left in
$G\setminus H$. More precisely, we find the vertices or super vertices
representing $u$ and $v$ in $\bbar H$, and identify them in a single
super vertex. Since $u$ and $v$ were $2\delta+1$-connected in $G$, these
contractions preserve all cuts of size at most $2\delta$, and the
resulting graph $\bbar G$ is obtained from $G$ by contractions only. 
The contractions can only decrease the number of
edges and vertices, so we conclude that $\bbar G$ has $\tO(n/\delta)$ vertices and $\tO(n/\delta)$ edges, and that  it preserves all non-trivial 
cuts of $G$ of size of size at most $\lambda+(1-\eps)\delta$. This
completes the proof of Theorem \ref{thm:main-tech-strong}.

\section{Limit concentration and low conductance cuts: the proofs}\label{sec:pagerank-anal}
In this section we will prove Theorem \ref{thm:ACL} and
\ref{thm:endgame} from Section \ref{sec:thm:conc}. We recommond that
the reader reviews Section \ref{sec:PageRank} before continuing. 
 The
starting point for both theorems is a multigraph with $m$ edges, an
initial mass distribution $p^\circ$ of total mass $1$, and
a listing of the vertices with positive mass in order of non-increasing density. 

First, in Section \ref{sec:sweep}, we will show how to implement
PageRank followed by a sweep in linear time, improving by a
logarithmic factor the bound from \cite[\S 2.2]{ACL07:pagerank}.

In Sections 
\ref{sec:high}--\ref{sec:single-low}, we will study when 
we from a settled mass distribution get a low-conductance
cut in the sweep. Sections
\ref{sec:high} and \ref{sec:low} make symmetric studies of the cases 
of high versus low settled densities, while
Section \ref{sec:single-low} studies the case where a single
vertex has $\Omega(1/m)$ too little mass, as 
needed for our new endgame result in Theorem \ref{thm:endgame}. 
The results of Sections \ref{sec:high} and \ref{sec:low} are
essentially equivalent to those in \cite{ACL07:pagerank}, but our
proofs are easily modified to prove the new results in 
Section \ref{sec:single-low}. Our proofs 
are quite similar to those in \cite{AC07:sharp-drop}, but cast
in terms of the flows from Lemma \ref{lem:stationary}.

In Section \ref{sec:proof-ACL}, we use the results from
Sections \ref{sec:high} and \ref{sec:low} to prove Theorem \ref{thm:ACL},
which says that if some set has $\Omega(1)$ excess mass
in the limit, then we can find a low-conductance
cut. The techniques to prove this are not new, but in the
previous papers \cite{AC07:sharp-drop,ACL07:pagerank} they always
assume that the initial distribution is with all mass on a single
``good'' vertex, while our Theorem \ref{thm:ACL} holds for arbitrary initial
distributions, including our $p^\circ_X$ where the initial
mass is spread with uniform density on an arbitrary subset $X$ of the vertices.

Finally, in Section \ref{sec:proof-endgame}, we prove 
our new endgame Theorem \ref{thm:endgame} which identifies
a low-conductance cut if there is a single vertex whose
limit distribution is $\Omega(1/m)$ too small.

\subsection{Sweeping for low conductance cuts in linear time}\label{sec:sweep}
We will first present a simple variant of the PageRank in Algorithm
\ref{alg:apprpr} which makes the sweep for a low conductance cut run
in linear time. The issue is that in order to do the sweep, we need
the vertices with positive settled mass to be sorted in order of
non-increasing settled mass density. In \cite[\S 2.2]{ACL07:pagerank}
they used regular sorting to get this order, costing them a
logarithmic factor.

First we note that we can make the push more flexible, only pushing
part of the residual mass at a vertex as described in Algorithm \ref{alg:push'}.
\begin{algorithm}\label{alg:push'}
\caption{Push'$(\alpha,u,q)$---assumes $r(u)\geq q$}
$p(u)\asgn p(u)+\alpha q$\;
\lFor{$(u,v)\in E$}{$r(v)\asgn r(v)+(1-\alpha)q/(2d(u))$}
$r(u)\asgn r(u)-(1+\alpha)q/2$.
\end{algorithm}
This more flexible push still satisfies all the basic properties discussed
in Section  \ref{sec:PageRank}, e.g., we preserve Invariant \req{eq:inv-push},
Fact \ref{fact:flow}, and Lemma \ref{lem:stationary}, and
as the residual mass goes to 0, the settled mass converges to
the same same unique limit $\PR(\alpha,p^\circ)$ as 
with $\ApprPR$ in Algorithm \ref{alg:apprpr}.

Our basic idea is that when we push from a vertex, we will always
push $\eps d(u)$ of the residual mass as described in
$\ApprPR'$ in  Algorithm \ref{alg:apprpr'}.
\begin{algorithm}\label{alg:apprpr'}
\caption{\ApprPR'$(\alpha,\eps,p^\circ)$}
$r\asgn p^\circ$;$\quad p\asgn 0^V$\;
\lWhile{$\exists u: r(u)/d(u)\geq \eps$}{Push'$(\alpha,u,\eps d(u))$}
\end{algorithm}
This means that the settled density $p(u)/d(u)$ on any vertex
is always an integer multiple of $\eps$.
\begin{lemma}\label{lem:sweep}
We are given an initial mass distribution $p^\circ$ of total mass $1$
and a listing of the vertices with positive initial mass in order of
non-increasing density. For any $\alpha,\eps\leq 1$, we can implement
$\ApprPR'$ from Algorithm \ref{alg:apprpr'} in $O(1/(\eps\alpha))$
time, producing a settled mass distribion $p$ listing the vertices
with positive settled mass in order of non-increasing density. Within
this time-bound, we can also find which prefix of list defines the lowest
conductance cut.
\end{lemma}

\begin{proof}
The initial mass is $1$, so the sum of the degrees of
the pushes is bounded by $1/(\eps\alpha)$, which 
is within a constant factor of the time bound we want.

We will maintain a doubly-linked push-list with vertices of residual
density at least $\eps$.  These are the vertices from which $\ApprPR'$
will have to do a push. Initially the residual mass is the initial
mass, so the initial push-list is extracted as a prefix of the list of
vertices in order of non-increasing initial/residual density.

We are going to use the general trick that each vertex has a flag
telling if it has been affected by the pushes. Moreover, we will have
a list with all vertices affected. The maximal length of this list is
$1/(\eps\alpha)$. All vertices have fields for residual density and
settled density that we can start using first time they are affected,
copying the residual mass from the initial mass and setting the
settled mass to zero. At the end, we will have to clean up, going
through the list of affected vertices setting their flags to
zero.

When we do a push from $u$, we first note that it affects $u$ and all
its neighbors, and the neighbors may have to be initialized as described above.
Afterwards, we just implement the push as described in Algorithm 
\ref{alg:push'} in constant time per neighbor. Note
that when we push from $u$, its residual denisty may drop
below $\eps$ and then it has to be removed from the push-list. Likewise,
some neigbors may now get residual density $\eps$ and
hence get added to the push-list. This way, the total
time spent is $O(1/(\eps\alpha))$.

We now go through the list of affected vertices, collecting all
vertices with positive settled mass. They all have settled density at least
$\eps$. Next we go through this positive list extracting all vertices with
settled density at least $2\eps$, and continue this way, for
$i=2,3,...$ extracting the vertices with settled density at least
$i\eps$, stopping when no more vertices are extracted. This
splits the vertices into lists $L_i$ with vertices of settled
density $i\eps$. The vertices in $L_i$ are considered $i+1$ times,
but they were pushed $i$ times, and the total number of pushes was
bounded by $1/(\eps\alpha)$, so the total time for this
is $O(1/(\eps\alpha))$. Concatenating the lists in reverse order,
we get a list $(v_1,\ldots,v_\ell)$ containing the vertices with
positive settled mass ordered by non-increasing density, as
desired.

We now want to find the lowest conductance cut based on a prefix of
the list. The volumes are trivially computed, just adding up degrees
for all prefixes. To compute the cut sizes
$c_i=|\boundary(\{v_1,\ldots,v_i\})|$, we assume that each vertex has
a field that we for an affected vertex assign its index from the list,
or set to $\ell+1$ if the vertex is not in the list because it has no
settled mass. This means that we for any edge leaving $v_i$ can tell
if the other end-point is before or after $v_i$ in the list. We now
sweep the vertices. We start with $c_0=0$. When we get to $v_i$, we
first set $c_i=c_{i-1}$. Then we subtract from $c_i$ the number of
edges from $v_i$ to preceeding vertices and add the number of edges to
succeeding vertices. Now $c_i$ is the desired cut size, which we
divide by the volume to get the conductance.  At the end we return the
lowest conductance cut. The time spent on $v_i$ is $O(d(v_i))$ but we
know that $v_i$ was pushed, hence that $\sum_i d(v_i)\leq
1/(\eps\alpha)$.  Thus we conclude that the total time spent is
$O(1/(\eps\alpha))$.
\end{proof}

\subsection{High densities}\label{sec:high}
We will now prove that the sweep must find a sparse cut if
the settled mass $p$ is concentrated. We will only consider
vertices with settled densities above some $t_0$
such that
\begin{equation}\label{eq:high}
\volp {>t_0}\leq m.
\end{equation}
Recall here that $V^p_{> t_0}$ is the set of vertices with settled density 
$\pd v>t_0$. We will focus on cuts defined by
sets $V^p_{\geq t}$ consisting of vertices with settled density $\pd v\geq t$
for some $t>t_0$. This implies that vertices with the same settled density will always
be on the same side of a cut. It is among the cuts defined by these sets that we will look
for low conductance $\Phi(V^p_{\geq t})=\boundp {\geq t}\left/\volp {\geq t}\right.$.
\begin{lemma}\label{lem:conductance}
Assuming \req{eq:high}, for any $\tau\in(t_0,1]$, 
\begin{equation}\label{eq:conductance}
\min_{t\in(t_0,\tau]} \Phi(V^p_{\geq t})\leq
\sqrt{\frac{18\alpha}{(\tau -t_0)\volp{\geq \tau}}}.
\end{equation}
\end{lemma}
\begin{proof}
To prove \req{eq:conductance}, consider any 
$\phi\leq \min_{t\in(t_0,\tau]} \Phi(V^p_{\geq t})$. We will
prove \req{eq:conductance} in the following
equivalent form:
\begin{equation}\label{eq:conductance-ind}
\tau-t_0 \leq \frac{18\alpha}{\phi^2\volp{\geq \tau}}.
\end{equation}
Let $t\in(t_0,\tau]$.  By \req{eq:high}, we have $\volp {\geq t}\leq  m$, 
so by definition, $\boundp {\geq t}\geq \phi\,\volp {\geq
  t}$. Consider any edge $(u,v)$ leaving $\Vp {\geq t}$ (so $u \in \Vp
{\geq t}$ but $v \not\in \Vp {\geq t}$). By Fact \ref{fact:flow}, the net flow over this edge from $u$ to $v$ is
at least $(\pd u-\pd v)/(3\alpha)$. Here
we used that $\alpha < 1/3$, so $(1-\alpha)/(2\alpha)>1/(3\alpha)$.
Since $\pd u\geq t>\pd v$ this
flow is always positive away from $\Vp {\geq t}$. 

Let $t'$ be the
median density $\pd v$ of a neighbor of $\Vp {\geq t}$, counting $\pd
v$ with the multiplicity of the number of edges from $\Vp {\geq t}$ to
$v$. We say that {\em $t'$ is derived from $t$ using median expansion}.
We note that $|\Vp {\geq t'}|>|\Vp {\geq t}|$ since there
is some vertex with median density $\pd v=t'$. We also note that
we may have $t'<t_0$ even though $t>t_0$.

We have at least $\boundp {\geq t}/2$ edges
from $\Vp {\geq t}$ to vertices $v$ with $\pd v\leq t'$, so
the net flow out of
$\Vp {\geq t}$ is at least
\begin{equation}\label{eq:flow-out}
\left(\boundp {\geq t}/2\right)(t-t')/3\alpha\geq
\phi\,\vol(\Vp{\geq t})(t-t')/6\alpha.
\end{equation}
But this can be no more than the total mass, which is $1$,
so
\begin{equation}\label{eq:tt'}
(t-t')\leq \frac{6\alpha}{\phi\,\volp {\geq t}}.
\end{equation}
By the median definition of $t'$, at least half the edges leaving $\Vp{\geq t}$
land in $\Vp{\geq t'}$, so 
\begin{align}
\volp {\geq t'}&\geq \volp {\geq t}+\boundp {\geq t}/2
\geq (1+\phi/2)\volp {\geq t}.\label{eq:growth}
\end{align}
We want to prove for any $t\leq \tau$ that
\begin{equation}\label{eq:ind-t}
t-t_0 \leq
\frac{18\alpha}{\phi^2\volp {\geq t}}.
\end{equation}
If $t\leq t_0$, the statement of \req{eq:ind-t} is trivially true, so
we can assume $t>t_0$. We want to use induction, inductively
assuming that \req{eq:ind-t} holds for the above defined $t'<t$. However,
we cannot just use induction over the reals. Instead, formally,
the induction is over the positive integer $n-|\Vp{\geq t}|$. Since
$|\Vp{\geq t'}|>|\Vp{\geq t}|$, it is inductively valid
to assume
that \req{eq:ind-t} holds for $t'$.
Combining this with \req{eq:tt'} and \req{eq:growth},  we 
get
\begin{align*}
t-t_0&\leq (t-t')+(t'-t_0)\\
&\leq \frac{6\alpha}{\phi\,\volp{\geq t}}+
\frac{18\alpha}{\phi^2\volp {\geq t'}}\\
&\leq \frac{6\alpha}{\phi\,\volp {\geq t}}+
\frac{18\alpha}{\phi^2(1+\phi/2)\volp {\geq t}}\\
&\leq
\frac{18\alpha}{\phi^2\volp {\geq t}}\left(
\phi/3+1/(1+\phi/2)\right)\\
&\leq \frac{18\alpha}{\phi^2\volp {\geq t}}.
\end{align*}
The last inequality uses that $\phi/3+1/(1+\phi/2)\leq 1$ for
all $\phi\in[0,1]$. This completes the proof of \req{eq:ind-t}. Now \req{eq:conductance-ind}
follows with $t=\tau$.
\end{proof}

\subsection{Low densities}\label{sec:low}
We will now make a symmetric study of vertices with settled density below some
$t_0$ such that
\begin{equation}\label{eq:low}
\volp {<t_0}\leq m.
\end{equation}
Note that if \req{eq:high} is false then \req{eq:low} is true, so
for any value of $t_0$, the analysis below applies if the
analysis from the previous subsection did not apply. We also note that
we can pick a median $t_0$ satisfing both \req{eq:high} and  \req{eq:low}.

We will consider the 
conductance $\Phi(V^p_{\leq t})=\boundp {\leq t}\left/\volp {\leq  t}\right.$
of the sets $V^p_{\leq t}$ of vertices with 
settled density $\pd v\leq t$ for some $t<t_0$.
Symmetric to the results from Section \ref{sec:high}, we will prove
\begin{lemma}\label{lem:conductance-low} Assuming \req{eq:low}, for any 
$\tau\in[0,t_0)$, 
\begin{equation}\label{eq:conductance-low}
\min_{t\in[\tau,t_0)} \Phi(V^p_{\leq t})\leq
\sqrt{\frac{18\alpha}{(t_0-\tau)\volp{\leq \tau}}}.
\end{equation}
\end{lemma}
\begin{proof}
The proof is symmetric to that of Lemma \ref{lem:conductance}, which
we assume that the reader is already familar with.
For any given $\phi\leq \min_{t\in[\tau,t_0)} \Phi(V^p_{\leq t})$,
we are going to prove \req{eq:conductance-low} in the following
equivalent form.
\begin{equation}\label{eq:conductance-ind-low}
t_0-\tau \leq \frac{18\alpha}{\phi^2\volp{\leq \tau}}.
\end{equation}
Consider any $t\in[\tau,t_0)$. By \req{eq:low}, we have $\volp {\leq t}\leq m$, so by definition,
$\boundp {\leq  t}\geq \phi\,\volp {\leq  t}$. Consider
any edge $(u,v)$ leaving $\Vp {\leq t}$ (so $u \in \Vp {\leq t}$ but $v \not\in \Vp {\leq t}$). By Fact \ref{fact:flow} 
the net flow over this edge from $v$ to $u$
is at least $(\pd v-\pd u)/(3\alpha)$. 
Since $\pd u\leq t<\pd v$ this
flow is always positive into $\Vp {\leq t}$. Let $t'$ be the median density $\pd v$
of a neighbor of $\Vp {\leq t}$, counting $\pd v$ with the multiplicity of the number of edges from $\Vp {\leq t}$ to $v$. This is the symmetric analogue of the
median expansion used in the proof of Lemma \ref{lem:conductance}.

We then have at least $\boundp {\leq t}/2$ edges
from $\Vp {\leq t}$ to vertices $v$ with $\pd v\geq t'$, so
the net flow into
$\Vp {\leq t}$ is at least
\begin{equation}\label{eq:flow-in}
\left(\boundp {\leq  t}/2\right)(t'-t)/(3\alpha)\geq
\phi\,\volp{\leq t}(t'-t)/(6\alpha).
\end{equation}
But this can be no more than the total mass, which is $1$,
so symmetric to \req{eq:tt'}, we get
\begin{equation}\label{eq:tt'-low}
(t'-t)\leq \frac{6\alpha}{\phi\,\volp {\leq t}}.
\end{equation}
Also, since $t'$ was the median neighboring density, corresponding
to \req{eq:growth}, we get
\begin{align}
\volp {\geq t'}&\geq \volp {\geq t}+\boundp {\geq t}/2
\geq (1+\phi/2)\volp {\geq t}\textnormal,\label{eq:growth-low}
\end{align}
The rest of the argument for Lemma \ref{lem:conductance-low}
is exactly the same as
the argument for Lemma \req{eq:conductance}.
\end{proof}

\subsection{A single low density}\label{sec:single-low}
In this section we will show that just a
single vertex with low density makes a big difference if
we have a good bound on the residual densities $r(v)/d(v)$ for every vertex $v$. More precisely, we will prove
\begin{lemma}\label{lem:conductance-low-single}
Assume \req{eq:low}, that is,  $\volp {<t_0}\leq m$. If $r(v)/d(v)\leq
\eps\leq 1/(2m)$ for all $v\in V$ and 
there is a vertex $u$ with density $\pd u\leq \tau$,
then 
\begin{equation}\label{eq:conductance-low-single}
\min_{t\in[\tau,t_0)} \Phi(V^p_{\leq t})\leq
\sqrt{\frac{12(t_0+\eps)\alpha \lg m}{t_0-\tau}}.
\end{equation}
\end{lemma}
\begin{proof}
Let $\phi\leq \min_{t\in[\tau,t_0)} \Phi(V^p_{\leq t})$.
We shall reuse a lot of the analysis from Section \ref{sec:low} based
on some $t<t_0$ and the median neighboring density $t'$ from
the median expansion in the proof of Lemma \ref{lem:conductance-low}.
In Section \ref{sec:low},   symmetric to the high density case,
we said that the total flow into $\Vp{\leq t}$ is at most
$1$. However, for the current proof we know that the residual density on every
vertex is bounded by $\eps$. Then the total mass on
$\Vp{\leq t}$ is at most $(t+\eps)\volp{\leq t}$. This
gives us a different bound on the net flow into $\Vp{\leq t}$, which by \req{eq:flow-in}
is at least $\phi\,\volp{\leq t}(t'-t)/(6\alpha)$. Thus,
as an alternative to \req{eq:tt'-low}, we have
\begin{equation}\label{eq:tt'-low-single}
\phi\,\volp{\leq t}(t'-t)/(6\alpha)\leq (t+\eps)\volp{\leq t}
\iff (t'-t)\leq  6(t+\eps)\alpha/\phi\leq  6(t_0+\eps)\alpha/\phi.
\end{equation}
Starting from $t=\tau$, we consider how many times
we can step from $t$ to $t'$ using median expansion before reaching or
passing $t_0$. We know that $u\in\Vp{\leq\tau}$ and $d(u)\geq 1$, so
we start with $\volp{\leq \tau}\geq 1$.

Now, every time we go from $t$ to $t'$, we  know from \req{eq:growth-low} that
the volume grows by at least a factor $(1+\phi/2)$, and by
definition, $\volp {< t_0}\leq m$, so we can have at most
\[\log_{(1+\phi/2)} m<(2/\phi)\lg m\]
iterations before we reach or pass $t_0$. Therefore
\[t_0-\tau\leq (2/\phi)(\lg m)6(t_0+\eps)\alpha/\phi=12(t_0+\eps)\alpha(\lg m)/\phi^2.\]
Thus we have
\[\phi\leq \sqrt{12(t_0+\eps)\alpha(\lg m)/(t_0-\tau)}.\]
This also holds for $\phi=\min_{t\in[\tau,t_0)} \Phi(V^p_{\leq t})$,
completing the proof of \req{eq:conductance-low-single}.
\end{proof}

\subsection{Exploiting concentration}\label{sec:proof-ACL}
Our goal in this subsection is to provide an algorithm performing as
in Theorem \ref{thm:ACL}, restated below for convenience.
\begin{quote}\it
We are given a multigraph with $m$ edges, an initial mass
distribution $p^\circ$ of total mass $1$,  and
a listing of the vertices with positive mass in order of non-increasing density. 
Let
$p^*=\PR(\alpha,p^\circ)$. We are also given an excess parameter
$\Cdiff<1$.  We have a PageRank algorithm that staring from $p^\circ$ 
will either find a set $T$ with $\vol(T)\leq m$ and conductance
\[\Phi(T)=O(\sqrt{(\alpha\log m)/\Cdiff})\textnormal,\]
or certify that there is no set
$S$ with 
\[\excess {p^*} S\geq \Cdiff.\]
The maximal running time is $O(m/(\Cdiff \alpha))$,
but if a set $T$ is returned, then the time is also bounded by $O(\vol(T)(\log m)/(\Cdiff \alpha))$.

If we are further given a volume parameter $s\leq m\Cdiff/16$, the 
algorithm will either find the above $T$ with the additional guarantees that
$\vol(T)\leq 8s/\Cdiff$ and $\excess {p^*} T\geq \Cdiff/(16\lg (4s))$,
or certify 
there is no set $S$ with $\vol(S)\leq s$ and $\excess {p^*} S\geq \Cdiff$.
The maximal running time is $O(s/(\Cdiff \alpha))$,
but if a set $T$ is returned, then the time is also bounded by $O(\vol(T)(\log m)/(\Cdiff \alpha))$.
\end{quote}

\paragraph{With a volume bound}
We will first address the case where we have a volume parameter $s\leq
m\Cdiff/16$ to bound $\vol(S)$. In this case, we will apply Algorithm
\ref{alg:nibble-s}. Below we analyze this algorithm, but some of the
lemmas will be more general so that we can reuse them on other
algorithms.
\begin{algorithm}\label{alg:nibble-s}
\caption{BoundedNibble$(\alpha,p^\circ,\Cdiff,s)$--assumes $s\leq m\Cdiff/16$}
$\eps\asgn \Cdiff/2$\;
\Repeat{$\eps\leq \Cdiff/(8s)$}{
$\eps\asgn \eps/2$\;
$p\asgn\ApprPR'(\alpha,\eps,p^\circ)$\;
\If{$\volp{\geq 1/(2m)+\eps}\geq \Cdiff/(8\eps\lg(4s))$}{\Return{
$T=\Vp{\geq t}$ where $t\in (1/(2m)+\eps/2,1/(2m)+\eps]$ minimizes
$\Phi(\Vp{\geq t})$.}}
}
\Return{``There is no set $S$ with
$\excess {p^*} S\geq \Cdiff$ and $\vol(S)\leq s$.''}
\end{algorithm}

Consider an iteration of the loop in Algorithm
\ref{alg:nibble-s} based on some $\eps> \Cdiff/(8s)$. We know from Lemma~\ref{lem:sweep} that it
takes $O(1/(\eps\alpha))$ time to run $\ApprPR'$ including a sweep for
low conductance cuts. Therefore the last iteration will dominate our
bound for the total running time. In particular it
follows that the total time is at most
  $O(1/((\Cdiff/(8s))\alpha))= O(s/(\Cdiff\alpha))$.

Consider again an arbitrary iteration with some $\eps\geq \Cdiff/(8s)$.
We claim that
\begin{equation}\label{eq:half-volume}
\volp{> 1/(2m)+\eps/2}\leq  m.
\end{equation}
To see this, first note by definition that 
we have settled density $p(u)/d(u) > 1/(2m)+\eps/2>\eps/2$ for all $u \in \Vp{> 1/(2m)+\eps/2}$. Therefore
the settled mass in $\Vp{> 1/(2m)+\eps/2}$ is bigger than
$\volp{>  1/(2m)+\eps/2}\eps/2$, but the mass cannot be bigger
than $1$ and $s\leq m\Cdiff/16$, so $\volp{>1/(2m)+\eps/2}<2/\eps\leq 8s/\Cdiff\leq m$, as claimed in
\req{eq:half-volume}. 

Suppose the iteration 
passes the condition of the if-statement and
returns a set $T$. This is then the final iteration, so the total
time bound is $O(1/(\eps\alpha))$. The condition of the if-statement together
with \req{eq:half-volume} means that we satisfy the conditions supposed by  the lemma below. In return, the lemma states that the set $T$ 
satisfies all the conditions of Theorem \ref{thm:ACL} in the case where
we have a volume bound parameter $s$.
\begin{lemma}\label{lem:find-T}
Let $p^*=\PR(\alpha,p^\circ)$ and $p\asgn\ApprPR'(\alpha,\eps,p^\circ)$.
Suppose $\volp{\geq  1/(2m)+\eps}\geq \Cdiff/(8\eps\lg(4s))$
and $\volp{>  1/(2m)+\eps/2}\leq m$.
Let $T=\Vp{\geq t}$ where $t\in (1/(2m)+\eps/2,1/(2m)+\eps]$ minimizes
$\Phi(\Vp{\geq t})$. Then
$\vol(T)\leq m$, $1/(\alpha\eps)=O(\vol(T)(\log m)/(\Cdiff\alpha))$,
$\excess {p^*} T\geq \Cdiff/(16\lg(4s)$,
and $\Phi(T)\leq \tO\left(\sqrt{(\alpha\log m)/\Cdiff}\right)$.
\end{lemma}
\begin{proof}
First we note that 
\[\Cdiff/(8\eps\lg(4s))\leq \volp{\geq  1/(2m)+\eps}\leq\vol(T)\leq\volp{> 1/(2m)+\eps/2}\leq m.\]
From this we immediately get that $\vol(T)\leq m$ and 
$1/(\eps\alpha)=O(\vol(T)(\log m)/(\Cdiff\alpha))$. Moreover, we get
\[\excess {p^*} T> \vol(T)(1/(2m)+\eps/2)-\vol(T)/2m= (\eps/2)\vol(T)\geq \Cdiff/(16\lg(4s)).\]
Finally we need to argue about the conductance of $T$. Set $t_0^+=1/(2m)+\eps/2$
and   $\tau^+=1/(2m)+\eps$. We have $\volp{> t_0^+}\leq m$ as in \req{eq:high}, so
by Lemma \ref{lem:conductance},
\begin{align}
\min_{t\in(t_0^+,\tau^+]} \Phi(V^p_{\geq t})&\leq
\sqrt{\frac{18\alpha}{(\tau^+ -t_0^+)\volp{\geq \tau^+}}}
\leq \sqrt{\frac{36\alpha}{\eps\volp{\geq 1/(2m)+\eps}}}\nonumber\\
&\leq \sqrt{\frac{36\alpha}{\Cdiff/(8\lg(4s))}}
\leq O\left(\sqrt{\frac{\alpha\log m}{\Cdiff}}\right).\label{eq:cond+}
\end{align}
\end{proof}
This completes the proof of Theorem \ref{thm:ACL} with volume bound $s$ assuming
the algorithm terminates satisfying the
condition $\volp{>  1/(2m)+\eps}\geq \Cdiff/(8\eps\lg(4s))$ of the if-statement for some $\eps\geq \Cdiff/(8s)$.
We need to prove that this happens if there is a set $S$ with $\excess
{p^*} S\geq \Cdiff$ and $\vol(S)\leq s$, as stated by the following
lemma:
\begin{lemma}\label{lem:s-term} Let $p^*=\PR(\alpha,p^\circ)$ and suppose there is a set $S$ with $\excess
{p^*} S\geq \Cdiff$ and $\vol(S)\leq s\leq 2m$. Then
there is an integer $i\leq\ceil{\lg (2s)}$ such that if 
we set $\eps=\Cdiff 2^{-i-1}$ and $p\asgn\ApprPR'(\alpha,\eps,p^\circ)$, then
$\volp{>  1/(2m)+\eps}\geq \Cdiff/(8\eps\lg(4s))$.
\end{lemma}
\begin{proof}
A vertex $u\in S$
contributes $d(u)\max\{0,\ppd u-1/(2m)\}$ to $\excess {p^*} S$, so all together,
the vertices $u$ with $\ppd u\leq 1/(2m)+\Cdiff/(2s)$ contribute less than
$\Cdiff/2$.  Let
\[S_1=\{u\in S\mid \ppd u>1/(2m)+\Cdiff/2\}\]
and for $i=2,...,\lceil\lg (2s)\rceil$, define
\[S_i=\{u\in S\mid 1/(2m)+\Cdiff2^{-i}< \ppd u\leq 1/(2m)+\Cdiff2^{1-i}\}\]
Then
\[\sum_{i=1}^{\lceil\lg (2s)\rceil}(p^*(S_i)-\vol(S_i)/(2m))>\Cdiff/2.\]
Thus,
for some $i=\{1,...,\lceil\lg (2s)\rceil\}$, we have \[p^*(S_i)-\vol(S_i)/(2m)> \Cdiff/(2\lg (4s)).\]
This will be the value of $i$ we chose for the lemma.
If $i>1$ then
\[p^*(S_i)-\vol(S_i)/(2m)\leq \Cdiff2^{1-i}\vol(S_i)\leq \Cdiff2^{1-i}
\volpp{> 1/(2m)+\Cdiff2^{-i}}.\]
So
\begin{equation}\label{eq:ideal-vol}
\volpp{> 1/(2m)+\Cdiff2^{-i}} > 2^{i-2}/\lg (4s).
\end{equation}
This equation is also satisfied if $i=1$, for then
$S_1\neq\emptyset$, and hence
$\volpp{> 1/(2m)+\Cdiff2^{-1}}\geq 1$.

As stated in the lemma, we set $\eps=\Cdiff 2^{-i-1}$ and
$p\asgn\ApprPR'(\alpha,\eps,p^\circ)$. Then our settled distribution
$p$ satisfies $\ppd u -\eps\leq \pd u\leq \ppd u$ for all vertices
$u$.  Therefore
\[\Vp{> 1/(2m)+\eps}\supseteq \Vpp{> 1/(2m)+2\eps}=\Vpp{>  1/(2m)+
\Cdiff2^{-i}}\textnormal,\]
and by \req{eq:ideal-vol},
\[\volpp{> 1/(2m)+\Cdiff 2^{-i}} \geq
2^{i-2}/\lg (4s)=\Cdiff/(8\eps\lg(4s)).\]
\end{proof}
This completes the proof of Theorem \ref{thm:ACL}
when a volume bound $s$ is given.

\paragraph{Without a volume bound}
With no volume parameter bound $s$, we will run Algorithm
\ref{alg:nibble} below, claiming that it satisfies that the statement
of Theorem~\ref{thm:ACL}.
\begin{algorithm}\label{alg:nibble}
\caption{Nibble$(\alpha,p^\circ,\Cdiff)$}
$\eps\asgn \Cdiff/2$\;
\Repeat{$\eps\leq \Cdiff/(16m)$}{
$\eps\asgn \eps/2$\;
$p\asgn\ApprPR'(\alpha,\eps,p^\circ)$\;
\If{$\volp{> 1/(2m)+\eps/2}\leq m$ and $\volp{\geq 1/(2m)+\eps}\geq \Cdiff/(8\eps\lg(8m))$}{\Return{
$T=\Vp{\geq t}$ where $t\in (1/(2m)+\eps/2,1/(2m)+\eps]$ minimizes
$\Phi(\Vp{\geq t})$.}}
\If{$\volp{<1/(2m)-\eps}\leq m$ and $\volp{< 1/(2m)-2\eps}\geq \Cdiff/(8\eps\lg(8m))$}{\Return{
$T=\Vp{\leq t}$ where $t\in [1/(2m)-2\eps,1/(2m)-\eps)$ minimizes
$\Phi(\Vp{\leq t})$.}}
}
\Return{``There is no set $S$ with
$\excess {p^*} S\geq \Cdiff$.''}
\end{algorithm}
Algorithm \ref{alg:nibble} has a lot of similarities with Algorithm
\ref{alg:nibble-s} applied with the trivial volume parameter bound
$s=2m$, but instead of always returning a set $T$ of high density
vertices, it may also return a set of low density vertices. The first
condition in each if-statement ensures that the set $T$ returned has
$\vol(T)\leq m$. The running time analysis is exactly the
same as that for Algorithm \ref{alg:nibble-s} with $s=2m$. 

Concerning the first if-statement, we note that the conditions matches
exactly the conditions supposed in Lemma \ref{lem:find-T}, which implies
that the set returned satisfies all the requirements of Theorem \ref{thm:ACL}.

We now need a corresponding lemma for the case where a set is returned
from the second if-statement, noting that now we have no no volume bound
expect the trival $2m$.
\begin{lemma}\label{lem:find-T-}
Let $p^*=\PR(\alpha,p^\circ)$ and $p\asgn\ApprPR'(\alpha,\eps,p^\circ)$.
Suppose $\volp{\leq 1/(2m)-\eps}\leq m$ and $\volp{< 1/(2m)-2\eps}\geq \Cdiff/(8\eps\lg(8m))$.
Let $T=\Vp{\leq t}$ where $t\in [1/(2m)-2\eps,1/(2m)-\eps)$ minimizes
$\Phi(\Vp{\leq t})$. Then
$\vol(T)\leq m$, $1/(\alpha\eps)=O(\vol(T)(\log m)/(\Cdiff\alpha))$,
and $\Phi(T)\leq \tO\left(\sqrt{\frac{\alpha\log m}{\Cdiff}}\right)$.
\end{lemma}
\begin{proof}
First we note that we have
\[\Cdiff/(8\eps\lg(8s))< \volp{\leq 1/(2m)-2\eps}\leq\vol(T)\leq
\volp{<1/(2m)-\eps}\leq m.\]
From this we immediately get that $\vol(T)\leq m$ and 
$1/(\eps\alpha)=O(\vol(T)(\log m)/(\Cdiff\alpha))$. 

Concerning the conductance, with $t_0^-=1/(2m)-\eps$, we have 
$\volp{<t_0^-}\leq m$ as in \req{eq:low}. 
With $\tau^-=1/(2m)-2\eps$, it follows
from Lemma \ref{lem:conductance-low}  that
\begin{align*}
\min_{t\in[\tau^-,t_0^-)} \Phi(V^p_{\leq t})&\leq
\sqrt{\frac{18\alpha}{(t_0^- - \tau^-)\volp{\leq \tau^-}}}
\leq \sqrt{\frac{18\alpha}{\eps\volp{\leq 1/(2m)-2\eps}}}\\
&\leq \sqrt{\frac{18\alpha(8\lg(8ms))}{\Cdiff}}
\leq O\left(\sqrt{\frac{\alpha\log m}{\Cdiff}}\right).
\end{align*}
\end{proof}
The above two lemmas imply that if a set $T$ is returned by any of the
two if-statements in Algorithm \ref{alg:nibble}, then $T$ satisfies
all the requirements of Theorem \ref{thm:ACL} in the case without a
volume bound. The total running time is
$O(1/(\eps\alpha))=O(\vol(T)(\log m)/(\Cdiff\alpha))$.

To complete the proof of Theorem \ref{thm:ACL}, we need to show that
if there is a set $S$ with $\excess {p^*} S\geq \Cdiff$, then there
will be some iteration where we pass the conditions of one of the two
if-statements.

Assume first that
\begin{equation}\label{eq:head*}
\volpp{\geq 1/(2m)}\leq m
\end{equation}
Then we always have
\begin{equation}\label{eq:head}
\volp{\geq 1/(2m)+\eps/2}\leq \volp{\geq 1/(2m)}\leq \volpp{\geq 1/(2m)}\leq m
\end{equation}
In particular this means that we always satisfy the
first condition of the first if-statement. Moreover, by Lemma \ref{lem:s-term}
with $s=2m$, there will be some iteration satisfying the second
condition if there is a set $S$ with $\excess
{p^*} S\geq \Cdiff$. This completes the proof of Theorem \ref{thm:ACL}
when \req{eq:head*} is satisfied. It remains to consider the case where \req{eq:head*} is false, hence where 
$\volpp{< 1/(2m)}< m$. Then
\begin{equation}\label{eq:tail}
\volp{< 1/(2m)-\eps}\leq \volpp{< 1/(2m)}< m.
\end{equation}
The first condition of the second if-statement is thus always 
satisfied. We need a lemma corresponding to Lemma \ref{lem:s-term}, stating
that there will always be an iteration satisfying the last condition of
the second if-statement:
\begin{lemma}\label{lem:s-term-} Let $p^*=\PR(\alpha,p^\circ)$ and 
suppose there is a set $S$ with $\excess {p^*} S\geq \Cdiff$. Then
there is an integer $i\leq\ceil{\lg (2s)}$ such that if we set
$\eps=\Cdiff 2^{-i-1}$ and $p\asgn\ApprPR'(\alpha,\eps,p^\circ)$, then
$\volp{< 1/(2m)-2\eps}\geq \Cdiff/(8\eps\lg(4s))$.
\end{lemma}
\begin{proof}
Since $p^*(V)=1$ and $\vol(V)=2m$,
we have
\begin{align*}
\volpp {< 1/(2m)}/(2m)-\ppV{< 1/(2m)}&=\ppV{\geq 1/(2m)}-\volpp {\geq
 1/(2m)}/(2m)\\
&\geq \vol(S)/(2m)-p^*(S)
\geq \Cdiff.
\end{align*}
We can now make an analysis of $p^*$ for densities below $1/(2m)$
which is symmetric to the one we did in the proof of Lemma \ref{lem:s-term-} 
with densities above $1/(2m)$
but based on $\Vp{\leq 1/(2m)}$ instead of $S$ and with $2m$ instead of $s$.
Corresponding to \req{eq:ideal-vol}, we find an
$i\leq \lceil\lg (4m)\rceil$ such that
\begin{equation}\label{eq:ideal-vol-neg}
\volpp{< 1/(2m)-\Cdiff2^{-i}} \geq 2^{i-2}/\lg (8m).
\end{equation}
Since $p^*$ is non-negative, we must have $i\geq \lg(2\Cdiff m)$, but we
will not exploit this in the analysis.
As stated in the lemma, we set $\eps=\Cdiff 2^{-i-1}$ and
$p\asgn\ApprPR'(\alpha,\eps,p^\circ)$. Then our settled distribution
$p$ satisfies $\ppd u -\eps\leq \pd u\leq \ppd u$ for all vertices
$u$.  Since
$p^*$ dominates $p$, we get
\[\Vp{< 1/(2m)-2\eps}=\Vp{< 1/(2m)-\Cdiff2^{-i}}\supseteq
\Vpp{<  1/(2m)- \Cdiff2^{-i}}\]
so
\begin{equation}\label{eq:set-size-tail}
\volp{< 1/(2m)-2\eps}\geq 2^{i-2}/\lg (8m)=\Cdiff/(8\eps\lg(8m)).
\end{equation}
\end{proof}
Summing up, suppose there is a set
$S$ with $\excess {p^*} S\geq \Cdiff$. If $\volpp{\geq 1/(2m)}\leq m$
as in \req{eq:head*}, then \req{eq:head} and Lemma \ref{lem:s-term}
implies that some iteration passes the conditions of the first
if-statement.  Otherwise $\volpp{< 1/(2m)}\leq m$. Then
\req{eq:tail} and Lemma \ref{lem:s-term-} implies that some
iteration passes the conditions of the second if-statement. In either
case we find a set $T$.

Even if the set $S$ does not exist, the quality of $T$ is given by
Lemma \ref{lem:find-T} if it comes from the first if-statement, and by
Lemma \ref{lem:find-T-} if it comes from the second if-statement.
More precisely, the lemmas state that $\vol(T)\leq m$ and
$\Phi(T)\leq \tO\left(\sqrt{\frac{\alpha\log m}{\Cdiff}}\right)$.
Moreover, for the finding iteration, they state that
$1/(\eps\alpha)=O(\vol(T)(\log m)/(\Cdiff\alpha))$, which
then bounds the total time spent on finding $T$.

If $S$ does not exist we may not find the set $T$. In that case,
we report the non-existence of $S$. The total time
spent is dominated by the last round, which takes time
$O(1/(\eps\alpha))=O(1/(\Cdiff/(16m))\alpha))=
O(m/(\Cdiff\alpha))$. This completes the proof of Theorem \ref{thm:ACL}.

\subsection{Exploiting single low density}\label{sec:proof-endgame}
In this subsection we present Algorithm \ref{alg:some-small}, proving
that it performs as
stated in Theorem \ref{thm:endgame}:
\begin{quote}\it
We are given a multigraph with $m$ edges, an initial mass
distribution $p^\circ$ of total mass $1$,  and
a listing of the vertices with positive mass in order of non-increasing density. 
Let $p^*=\PR(\alpha,p^\circ)$. We are also given a parameter $\Cdiff<1$.
We have a PageRank algorithm that staring from $p^\circ$ 
will either 
find a set $T$ with $\vol(T)\leq m$ and
conductance 
\[\Phi(T)=O(\sqrt{(\alpha\log m)/\Cdiff})\textnormal,\]
or certify that there is no vertex $u$ with 
\[\ppd u\leq (1-\Cdiff)/(2m).\]
The running time of the algorithm is $O(m/(\Cdiff\alpha))$, and, depending
on the input, it will always end in one of the following cases:
\begin{itemize}
\item[(i)] The set $T$ is found in time $O(\vol(T)(\log m)/(\Cdiff\alpha))$ and
has $\excess {p^*} T\geq \Cdiff/(64\lg (8m))$.
\item[(ii)] The set $T$ is guaranteed to contain all small density vertices $u$ with
$\ppd u\leq (1-\Cdiff)/(2m)$.\\
In this case, even if $T$ is small, we have no better time bound than $O(m/(\Cdiff\alpha))$.
\item[(iii)] A certificate that there is no vertex
$u$ with $\ppd u \leq (1-\Cdiff)/(2m)$.
\end{itemize}
\end{quote}
\begin{algorithm}\label{alg:some-small}
\caption{SomeSmall$(\alpha,p^\circ,\Cdiff)$}
$\Cdiff'=\Cdiff/4$\;
$\eps\asgn \Cdiff'/2$\;
\Repeat{$\eps<\Cdiff'/(16m)$\hspace{2ex}}{
$\eps\asgn \eps/2$\;
$p\asgn\ApprPR'(\alpha,\eps,p^\circ)$\;
\If{$\volp{> 1/(2m)+\eps/2}\leq m$ and $\volp{\geq 1/(2m)+\eps}\geq \Cdiff'/(8\eps\lg(8m))$}{\Return{Case (i):
$T=\Vp{\geq t}$ where $t\in (1/(2m)+\eps/2,1/(2m)+\eps]$ minimizes
$\Phi(\Vp{\geq t})$.}}
}
$\eps\asgn\Cdiff/(8m)$\;
$p\asgn\ApprPR'(\alpha,\eps,p^\circ)$\;
\eIf{$\exists u:\pd u\geq (1-\Cdiff)/(2m)$}{\Return{Case (ii):
$T=\Vp{\leq t}$ where $t\in [(1-\Cdiff)/(2m),(1-0.75\Cdiff)/(2m))$ minimizes
$\Phi(\Vp{\leq t})$.}}{\Return{Case (iii): ``There is no
vertex $u$ with $\ppd u\leq (1-\Cdiff)/(2m)$.''}}
\end{algorithm}
Like with the previous algorithms, we note that the running time 
of Algorithm \ref{alg:some-small} is $O(1/(\eps\alpha))$ where $\eps$ has
the smallest value encountered. It follows that the maximal running time
is $O(m/(\Cdiff\alpha))$. Also, it follows from Lemma \ref{lem:find-T} with
$\Cdiff'$ instead of $\Cdiff$ and $s=2m$ that
if the if-statement in the repeat-loop returns a set $T$, then
$T$ satisfies all the requirements of Theorem \ref{thm:endgame} (i).

Assume now that
\begin{equation}\label{eq:high-end}
\volpp{<(1-\Cdiff/2)/(2m)}> m.
\end{equation}
If so, we have negative concentration
\[\volpp {<(1-\Cdiff/2)/(2m)}/(2m)-\ppV{<(1-\Cdiff/2)/(2m)}>
m\Cdiff/(4m)=\Cdiff/4=\Cdiff'\textnormal,\]
where $\Cdiff'$ has the value from Algorithm \ref{alg:some-small}.
We can only have positive concentration, or excess, on the set 
$S=\Vpp {\geq 1/(2m)}$, and negative concentration ourside, so
\begin{align*}
\excess {p^*} S& =\ppV{\geq 1/(2m)}-\volpp {\geq 1/(2m)}/(2m)\\
&= \volpp {<1/(2m)}/(2m)-\ppV{< 1/(2m)}\\
&\geq \volpp {<(1-\Cdiff/2)/(2m)}/(2m)-\ppV{<(1-\Cdiff/2)/(2m)}>\Cdiff'.
\end{align*}
It follows from Lemma \ref{lem:s-term} with $\Cdiff'$ instead of
$\Cdiff$ that the second condition of the if-statement is satisfied by
some iteration of the repeat-loop in Algorithm
\ref{alg:some-small}. Moreover, since the total volume is $2m$, by
\req{eq:high-end},
\[m>\volpp{\geq (1-\Cdiff/2)/(2m)}\geq \volp{\geq (1-\Cdiff/2)/(2m)}\geq\volp{\geq 1/(2m)+\eps/2}.\]
This means that the first condition of if-statement is satisfied
by all iterations of the repeat-loop. Thus, if \req{eq:high-end} is true, then some iteration
will return a set $T$, and we saw above that this $T$ 
satisfies Theorem \ref{thm:endgame} (i).

We now assume that the repeat-loop fails to find a set $T$, implying that
\req{eq:high-end} is false.
As in Algorithm \ref{alg:some-small},
we set $\eps=\Cdiff/(8m)$ and $p\asgn \ApprPR'(\alpha,\eps,p^\circ)$.
We also set $\tau=(1-\Cdiff)/(2m)$.
Since $p^*$ dominates $p$,
if $\ppd u\leq \tau$, then $\pd u\leq \tau$. Thus, if there is no $u$ with
$\pd u\leq\tau$, then we certify that 
there is no $u$ with $\ppd u\leq \tau$ as in Theorem \ref{thm:endgame} (iii). Thus we may assume that there is a $u$ with $\pd u\leq\tau=(1-\Cdiff)/(2m)$.

Since \req{eq:high-end} is false, with $t_0=(1-\Cdiff/2)/(2m)-\eps=(1-0.75\Cdiff)/(2m)$, we get
\begin{equation}
\label{eq:end-small}
\volp{<t_0}\leq \volpp{<(1-\Cdiff/2)/(2m)}\leq m.
\end{equation}
Since there is a $u$ with
$\pd u\leq\tau$, Algorithm \ref{alg:some-small} 
returns a set $T=\Vp{\leq t}$ where $t\in [\tau,t_0)$ minimizes
$\Phi(\Vp{\leq t})$. By Lemma \ref{lem:conductance-low-single}, we
have 
\[\Phi(T)
\leq\sqrt{\frac{12(t_0+\eps)\alpha \lg m}{t_0-\tau}}
=\sqrt{\frac{12((1-0.5\Cdiff)/(2m))\alpha \lg m}{\Cdiff/(8m)}}
=O\left(\sqrt{\frac{\alpha \log m}{\Cdiff}}\right)
\]
Moreover, the returned set $T=\Vp{\leq t}$ has $\vol(T)\leq\volp
{<t_0}\leq m$ and since $t\geq\tau$, $T$ includes every $u$
with $\ppd u\leq \pd u\leq\tau=(1-\Cdiff)/(2m)$, as required for Case (ii).
This completes the proof of Theorem \ref{thm:endgame}.

\section{Cactus}
\label{cactuss}

Recall that
 the set $\partial (U)$ of edges connecting $U$ and $T=V \backslash U$ is called a \emph{cut} while $U$ and $T$ are the \emph{sides} of the cut.

We call a loopless and 2-edge-connected graph $G$ a \emph{cactus} if each edge belongs to exactly one
cycle. This is equivalent to saying that all blocks are
cycles (allowing two-element cycles). For example,
a cactus may be obtained by duplicating each edge of
a tree.
Note that the minimum cuts of a cactus $C$ are exactly
those pairs of edges which belong to the same cycles
of $C$.

The following result states that the minimum cuts
of an arbitrary graph have the same structure as the
minimum cuts of a cactus.

\begin{theorem}[Dinits, Karzanov, and Lomonosov, \cite{DKL76}]
\label{cactus1}
Let $\lambda$ be an integer and
$G = (V,E)$ a loopless graph for which the cardinality of a minimum cut is $\lambda$. There is a cactus $C = (U, F)$ and a mapping $\phi$ from $V$ to $U$ so that the pre-images $\phi^{-1}(U_1)$ and $\phi^{- 1}(U_2)$ are the two sides of a minimum cut of $G$ for every 2-element cut of $C$ with sides
$U_1$ and $U_2$. Moreover, every minimum cut of $G$ arises this way. 
\end{theorem}

Gabow's cactus algorithm \cite{Gab16} can construct the
cactus and mapping of Theorem \ref{cactus1} in $\tO(\lambda m)$ time,
but here we will do the construction in near-linear time if the input
graph is simple.

\begin{theorem}
\label{main2}
There is a near-linear time algorithm that given
a simple graph $G=(V,E)$ constructs
a cactus $C = (U, F)$ and a mapping $\phi$ from $V$ to $U$
as described in Theorem \ref{cactus1}.
\end{theorem}

\drop{
In order to show this theorem we need some definitions.
Two subsets
$X$ and $Y$ of vertices are called \emph{crossing} if none of $X \backslash Y$, $Y \backslash X$,
$X \cap Y$, $V \backslash (X \cup Y)$
is empty. Two cuts $\partial(X)$ and $\partial(Y)$ are \emph{crossing} if $X$ and $Y$ are crossing.
}

\begin{proof}
Let $\delta$ be the minimum degree in $G$. First we apply our
main technical result, Theorem~\ref{thm:main-tech}, which contracts
vertex sets in near-linear time, producing a graph $\bbar G=(\bbar
V,\bbar E)$ with $\bbar m=\tO(m/\delta)$ edges such that all
non-trivial min-cuts of $G$ are preserved in $\bbar G$.

As in the proof of Corollary \ref{cor:min-cut}, we now run Gabow's
edge-connectivity algorithm \cite{Gab95} on $\bbar G$, asking it to fail if the
edge-connectivity $\bbar \lambda$ of $\bbar G$ is above $\delta$. This
takes $\tO(\delta \bbar m)=\tO(m)$ time, and now we compare $\bbar \lambda$
with the min-degree $\delta$.

If $\bbar \lambda>\delta$, then all min-cuts in $G$ are trivial and
then the edge connectivity $\lambda$ of $G$ is the minimum degree
$\delta$. Let $v_1,\dots, v_h$ denote the vertices of degree
$\delta$. Let $U = \{u_0, u_1,\dots, u_h\}$ be the vertex set of
a cactus $C=(U,F)$ in which $u_0$ and $u_i$ are connected by two
parallel edges in $F$ for each $i = 1,\dots, h$.  Let $\phi$ be the
mapping from $V$ to $U$ defined by $\phi(v_i)=u_i$ for $i = 1,\dots,
h$ and $\phi(v) = u_0$ for $v \in V \setminus \{v_1,\dots,
v_h\}$. Then $C$ and $\phi$ form a cactus for $G$ as described in
Theorem \ref{cactus1}.

Suppose instead that $\bbar \lambda\leq \delta$. Then $\bbar\lambda$
is also the edge connectivity $\lambda$ of $G$. We then
apply the cactus algorithm of Gabow \cite{Gab16} to $\bbar G$. In
$\tO(\bbar \lambda \bbar m)=\tO(m)$ time, it produces a cactus $C = (U, F)$ of $\bbar G$ and
a mapping $\bbar \phi$ from $\bbar V$ to $U$ as in Theorem \ref{cactus1}.
Next we turn $\bbar \phi$ into a mapping $\phi:V\fct U$ from the
original vertex set $V$ by reversing the contractions from
Theorem \ref{thm:main-tech}, that is, if $v$ got contracted into
the super vertex $\bbar v$, then $\phi(v)=\bbar \phi(\bbar v)$. Now
we have a cactus representing some min-cuts of $G$, including all non-trivial
min-cuts of $G$. If $\bbar\lambda< \delta$, then there are no trivial min-cuts
of $G$, and then our cactus $C$ is the final cactus for $G$.

Finally, if $\bbar\lambda=\delta$, there may be some trivial min-cut of $G$
that are not yet represented. The min-cut around a
min-degree vertex $v$ is represented if and only if
there is a vertex $u\in U$ such that $\{v\}=\phi^{-1}(u)$ and $u$ has only
two incident edge. If this
is not the case, let $u=\phi(v)$. To include the min-cut around $v$,
we introduce a new vertex $u'$ in $U$ and set $\phi(v)=u'$. The only
neighbor of $u'$ is $u$ and we add two parallel edges between them to $F$.
This adds the desired trivial min-cut but no other cuts to the
cactus representation.
We repeat this process for all min-degree vertices whose min-cut is not yet
represented. Now $C$ and $\phi$ is a cactus for $G$ as described in
Theorem \ref{cactus1}. Adding the trivial min-cuts took $O(m)$ time,
so the whole construction time is $\tO(m)$. This completes the proof of Theorem
 \ref{main2}.
\end{proof}

\section{Acknowledgments}
We would like to thank Hal Gabow and Yuzhou Gu as well as anonymous
referees from STOC'15 and J. ACM for patiently reading earlier
versions of this paper, pointing out issues and coming with useful
suggestions for the presentation.


\begin{thebibliography}{10}

\bibitem{AC07:sharp-drop}
R.~Andersen and F.~R.~K. Chung.
\newblock Detecting sharp drops in pagerank and a simplified local partitioning
  algorithm.
\newblock In {\em Proc. 4th TAMC}, pages 1--12, 2007.

\bibitem{ACL07:pagerank}
R.~Andersen, F.~R.~K. Chung, and K.~J. Lang.
\newblock Using pagerank to locally partition a graph.
\newblock {\em Internet Mathematics}, 4(1):35--64, 2007.

\bibitem{Google98:pagerank}
S.~Brin and L.~Page.
\newblock The anatomy of a large-scale hypertextual web search engine.
\newblock {\em J. Comput. Networks and ISDN Systems archive}, 30:107--117,
  1998.

\bibitem{CFN04:plan-mincut}
P.~Chalermsook, J.~Fakcharoenphol, and D.~Nanongkai.
\newblock A deterministic near-linear time algorithm for finding minimum cuts
  in planar graphs.
\newblock In {\em Proc. 15th SODA}, pages 828--829, 2004.

\bibitem{DKL76}
E.~A. Dinitz, A.~V. Karzanov, and M.~V. Lomonosov.
\newblock On the structure of a family of minimum weighted cuts in a graph.
\newblock In A.~A. Fridman, editor, {\em Studies in Discrete Optimization},
  pages 290--306. Nauka, Moskow, 1976.
\newblock (in Russian).

\bibitem{ET75}
S.~Even and R.~E. Tarjan.
\newblock Network flow and testing graph connectivity.
\newblock {\em SIAM Journal of Computing}, 4:507--518, 1975.

\bibitem{FF56}
L.~R. Ford and D.~R. Fulkerson.
\newblock Maximal flow through a network.
\newblock {\em Canadian Journal of Mathematics}, 8:399--404, 1956.

\bibitem{frank94}
A.~Frank.
\newblock On the edge-connectivity algorithm of {Nagamochi and Ibaraki}
  (unpublished).
\newblock Laboratoire Artemis, IMAG, Universite J. Fourier, Grenoble, 1994.

\bibitem{Gab95}
H.~N. Gabow.
\newblock {A matroid approach to finding edge connectivity and packing
  arborescences}.
\newblock {\em J. Comp. Syst. Sc.}, 50:259--273, 1995.
\newblock Announced at STOC'91.

\bibitem{Gab16}
H.~N. Gabow.
\newblock The minset-poset approach to representations of graph connectivity.
\newblock {\em {ACM} Trans. Algorithms}, 12(2):24:1--24:73, 2016.
\newblock Announced at FOCS'91.

\bibitem{GH61}
R.~E. Gomory and T.~C. Hu.
\newblock Multi-terminal network flows.
\newblock {\em J. SIAM}, 9(4):551--570, 1961.

\bibitem{GHT16:inc-edge-conn}
G.~Goranci, M.~Henzinger, and M.~Thorup.
\newblock Incremental exact min-cut in poly-logarithmic amortized update time.
\newblock In {\em Proc. 24th ESA}, pages 46:1--46:17, 2016.

\bibitem{HO94}
J.~Hao and J.~Orlin.
\newblock A faster algorithm for finding the minimum cut in a directed graph.
\newblock {\em Journal of Algorithms}, 17(3):424--446, 1994.
\newblock announced at SODA'92.

\bibitem{HRW17}
M.~Henzinger, S.~Rao, and D.~Wang.
\newblock Local flow partitioning for faster edge connectivity.
\newblock In {\em Proc. 28th SODA}, pages 1919--1938, 2017.

\bibitem{HW96}
M.~Henzinger and D.~Williamson.
\newblock On the number of small cuts in a graph.
\newblock {\em Inf. Process. Lett.}, 59(1):41--44, 1996.

\bibitem{HLT01}
J.~Holm, K.~Lichtenberg, and M.~Thorup.
\newblock Poly-logarithmic deterministic fully-dynamic algorithms for
  connectivity, minimum spanning tree, 2-edge and biconnectivity.
\newblock {\em J. ACM}, 48(4):723--760, 2001.
\newblock Announced at STOC'98.

\bibitem{Kar99}
D.~R. Karger.
\newblock Random sampling in cut, flow, and network design problems.
\newblock {\em Math. Oper. Res.}, 24(2):383--413, 1999.
\newblock Announced at STOC'94.

\bibitem{Kar00}
D.~R. Karger.
\newblock Minimum cuts in near-linear time.
\newblock {\em J. ACM}, 47(1):46--76, 2000.
\newblock Announced at STOC'96.

\bibitem{DBLP:conf/soda/KargerP09}
D.~R. Karger and D.~Panigrahi.
\newblock A near-linear time algorithm for constructing a cactus representation
  of minimum cuts.
\newblock In {\em Proc. 20th SODA}, pages 246--255, 2009.

\bibitem{KS96}
D.~R. Karger and C.~Stein.
\newblock A new approach to the minimum cut problem.
\newblock {\em J. ACM}, 43(4), 1996.
\newblock Announced at SODA'92 and STOC'93.

\bibitem{KT86}
A.~Karzanov and E.~Timofeev.
\newblock Efficient algorithm for finding all minimal edge cuts of a
  nonoriented graph.
\newblock {\em Kibernetika}, 2:8--12, 1986.
\newblock Translated in {\em Cybernetics}, (1986), pp. 156-162.

\bibitem{KT14-C}
K.~Kawarabayashi and M.~Thorup.
\newblock Deterministic global minimum cut of a simple graph in near-linear
  time.
\newblock In {\em Proc. 47th STOC}, pages 665--674, 2015.

\bibitem{Ma87}
D.~W. Matula.
\newblock Determining the edge connectivity in {$O(mn)$} time.
\newblock In {\em Proc. 28th FOCS}, pages 249--251, 1987.

\bibitem{Mat93}
D.~W. Matula.
\newblock A linear time {$2 + \epsilon$} approximation algorithm for edge
  connectivity.
\newblock In {\em Proc. 4th SODA}, pages 500--504, 1993.

\bibitem{menger}
K.~Menger.
\newblock Zur allgemeinen kurventheorie.
\newblock {\em Fund. Math.}, 10:96--115, 1927.

\bibitem{NI92A}
H.~Nagamochi and T.~Ibaraki.
\newblock Computing edge connectivity in multigraphs and capacitated graphs.
\newblock {\em SIAM J. Discr. Math.}, 5(1):54--66, 1992.
\newblock Announced at SIGAL'90.

\bibitem{NI92}
H.~Nagamochi and T.~Ibaraki.
\newblock Linear time algorithms for finding a sparse $k$-connected spanning
  subgraph of a $k$-connected graph.
\newblock {\em Algorithmica}, 7:583--596, 1992.

\bibitem{Po73}
V.~D. Podderyugin.
\newblock An algorithm for finding the edge connectity of graphs.
\newblock {\em Vopr. Kibern.}, 2:136, 1973.

\bibitem{Lex03}
A.~Schrijver.
\newblock {\em Combinatorial Optimization: Polyhedra and Efficiency}.
\newblock Springer, NY, 2003.

\bibitem{ST11:spectral}
D.~Spielman and S.~Teng.
\newblock Spectral sparsification of graphs.
\newblock {\em {SIAM} J. Comput.}, 40(4):981--1025, 2011.

\bibitem{SW97}
M.~Stoer and F.~Wagner.
\newblock A simple min-cut algorithm.
\newblock {\em J. ACM}, 44:585--591, 1997.
\newblock Announced at ESA'94.

\end{thebibliography}
\end{document}